%% file: journal_split_without_blue.tex
\documentclass[journal,10pt,onecolumn,draftclsnofoot]{IEEEtran}

\usepackage{amsmath}
\usepackage{amssymb}
\usepackage{cite}
\usepackage{graphicx}
\usepackage{epstopdf}
\usepackage{url}
\usepackage{cite}
\usepackage{color}
\usepackage{amsthm}
\usepackage{algorithm}
\usepackage{algorithmic}
\usepackage{subcaption}

\usepackage{accents}

\usepackage{tikz}
\usetikzlibrary{positioning,chains,fit,calc,shapes,arrows}

\newtheorem{problem}{Problem}
\newtheorem{example}{Example}
\newtheorem{definition}{Definition}
\newtheorem{theorem}{Theorem}
\newtheorem{lemma}{Lemma}

\newtheorem{remark}{Remark}

\newtheorem{corollary}{Corollary}
\newtheorem*{example2}{Example 2}
\newcommand{\Z}{\mathbb{Z}}
\newcommand{\R}{\mathbb{R}}

\newcommand{\B}{\mathcal{B}}
\newcommand{\F}{\mathcal{F}}
\newcommand{\D}{\mathcal{D}}

\newcommand{\PPB}{\mathcal{P}(\mathcal{B})}
\newcommand{\PPBbar}{\overline{\mathcal{P}}(\mathcal{B})}
\newcommand{\PPBo}{\accentset{\circ}{\mathcal{P}}(\mathcal{B})}

\newcommand{\CB}{\mathcal{C}(\mathcal{B})}

\newcommand{\CBONE}{\mathcal{C}^{1}_{\PPB}}
\newcommand{\CBZERO}{\mathcal{C}^{0}_{\PPB}}

\newcommand{\hx}{\hat{x}}
\newcommand{\hz}{\hat{z}}

\newcommand{\bias}{p}

\newcommand{\CN}{\mathcal{T}_f}

\DeclareMathOperator*{\vol}{Vol}

\usepackage{xinttools} % for the \xintFor***

\def\biglen{20cm} % playing role of infinity (should be < .25\maxdimen)
\tikzset{
  half plane/.style={ to path={
       ($(\tikztostart)!.5!(\tikztotarget)!#1!(\tikztotarget)!\biglen!90:(\tikztotarget)$)
    -- ($(\tikztostart)!.5!(\tikztotarget)!#1!(\tikztotarget)!\biglen!-90:(\tikztotarget)$)
    -- ([turn]0,2*\biglen) -- ([turn]0,2*\biglen) -- cycle}},
  half plane/.default={1pt}
}

\title{Neural network approaches \\ to point lattice decoding}

\author{Vincent Corlay, Joseph J. Boutros, Philippe Ciblat, and Lo\"ic Brunel 
\thanks{V. Corlay is with Mitsubishi Electric R\&D Centre Europe, Rennes, France, and Telecom Paris,  Palaiseau, France (v.corlay@fr.merce.mee.com). 
J. J. Boutros is with the Department of Electrical and Computer Engineering, Texas A\&M University at Qatar, Doha, Qatar (boutros@tamu.edu).
P.~Ciblat is with Telecom Paris, Palaiseau, France (philippe.ciblat@telecom-paris.fr). 
L. Brunel is with Mitsubishi Electric R\&D Centre Europe, Rennes, France (l.brunel@fr.merce.mee.com).
}
\thanks{Part of this paper was presented at the IEEE International Symposium on Information Theory, Paris, France, July 2019.}
}

\begin{document}

\maketitle

%\tableofcontents

%\newpage
%%-----------------------------------------------------------------------------------------------------

\begin{abstract}
\input{0_0_abstract}
\end{abstract}

\input{0_1_intro}

%\subsection{Main results}

We summarize below the main contributions of the paper.

\input{0_main_results}

%%-----------------------------------------------------------------------------------------------------
\section{Preliminaries}
\label{sec_prele}

This section is intended to introduce the notations for readers with a sufficient background in lattice theory.
It is also useful as a short introduction to lattices for newcomers to whom we suggest reading chapters 1-4 in \cite{Conway1999}.
Additional details on all elements of this section are found in \cite{Conway1999} and \cite{Forney1988}. 

\noindent \textbf{Lattice.} 
A lattice $\Lambda$ is a discrete additive subgroup of $\R^n$.
For a rank-$n$ lattice in $\R^n$, the rows of a $n\times n$ generator matrix $G$ constitute
a basis of $\Lambda$ and any lattice point $x$ is obtained via $x=z \cdot G$, where $z \in \Z^n$.
The Gram matrix is $\Gamma=G\cdot G^T=(GQ) \cdot (GQ)^T$, 
where $Q$ is any $n \times n$ orthogonal matrix.
All bases defined by a Gram matrix are equivalent modulo rotations and reflections.
A lower triangular generator matrix is obtained from the Gram matrix 
by Cholesky decomposition \cite[Chap.~2]{Cohen1996}.
For a given basis $\mathcal{B}=\{ g_i \}_{i=1}^n$ forming the rows of $G$,
the fundamental parallelotope of $\Lambda$ is defined by
\begin{equation}
\label{equ_PB}
\mathcal{P}(\mathcal{B}) = \{ y \in \R^n : y=\sum_{i=1}^{n}\alpha_{i}g_{i}, \ 0 \leq \alpha_{i} < 1  \}.
\end{equation}
%For $\PPB$ in (\ref{equ_PB}), we also define its closure denoted by $\PPBbar$.
%define its interior denoted by $\PPBo$ and 
The Voronoi region of $x$ is:
\begin{equation}
\mathcal{V}(x)=\{ y \in \R^n : \|y-x\| \le \|y-x'\|, \forall x' \neq x,  \ x,x' \in \Lambda  \}.
\end{equation}
A Voronoi facet denotes a subset of the points
\begin{equation}
\{ y \in \R^n : \|y-x\| = \|y-x'\|, \forall x' \neq x,  \ x,x' \in \Lambda  \},
\end{equation}
which are in a common hyperplane. \\
$\PPB$ and $\mathcal{V}(x)$ are fundamental regions of the lattice: one can perform a tessellation of $\R^n$ with these regions.
The fundamental volume of $\Lambda$ is $\vol(\mathcal{V}(x))=\vol(\mathcal{P}(\mathcal{B}))=|\det(G)|$.

The minimum Euclidean distance of $\Lambda$ is $d(\Lambda)=2\rho(\Lambda)$, where $\rho(\Lambda)$ is the packing radius.
The nominal coding gain $\gamma$ of a lattice $\Lambda$ is given by the following ratio \cite{Forney1988}
\begin{equation}
\gamma(\Lambda) = \frac{d^{2}(\Lambda)}{\text{vol}(\Lambda)^{\frac{2}{n}}}.
\end{equation}
A vector $v \in \Lambda$ is called Voronoi vector if the hyperplane \cite{Conway1992}
\begin{equation} 
\{y \in \mathbb{R}^{n} \ : \ y \cdot v = \frac{1}{2}||v||^2 \}
\end{equation}
has a non empty intersection with $\mathcal{V}(0)$. The vector is said relevant \cite[Chap.~2]{Conway1999} if the intersection includes a $(n-1)$-dimensional face of $\mathcal{V}(0)$.
We denote by $\tau_{f}$ the number of relevant Voronoi vectors, referred to as the Voronoi number in the sequel.
For root lattices \cite{Conway1999}, the Voronoi number is equal to the kissing number~$\tau$, defined as the number of points at a distance $d(\Lambda)$ from the origin.
For random lattices, we have $\tau_{f}=2^{n+1}-2$ (with probability 1) \cite{Conway1992}.
The set $\CN(x)$, for $x \in \Lambda$, denotes the set of lattice points having a common Voronoi facet with~$x$.
The theta series of $\Lambda$ is \cite[Chap.~2, Section~2.3]{Conway1999}
\begin{align}
\label{equ_theta_series}
\Theta_{\Lambda}(q)=\sum_{x \in \Lambda} q^{\|x\|^2}=\sum_{\ell=0}^\infty \tau_\ell q^\ell,
\end{align} 
where $\tau_\ell$ represents the number of lattice points of norm $\ell$ in $\Lambda$ (with $\tau_{4\rho^2}=\tau$).
 Moreover, a lattice shell denotes the set of $\tau_i$ lattice points at a distance $i$ from the origin.  
For instance, the first non-zero term of the series is $\tau q^{4\rho^2}$ as there are $\tau$ lattice points at a distance $d(\Lambda)$ from the origin. These lattice points constitute the first lattice shell.

For any lattice $\Lambda$ the dual lattice $\Lambda^{*}$ is defined as follows \cite[Chap.~2, Section~2.6, (65)]{Conway1999}:
\begin{align}
\Lambda^{*} = \{u \in \mathbb{R}^{n}: u \cdot x \in \mathbb{Z}, \ \forall \ x \in \Lambda \}.
\end{align}
Hence if $G$ is a square generator matrix for $\Lambda$, then $(G^{-1})^{T}$ is a generator matrix for $\Lambda^{*}$. Moreover, if a lattice is $equivalent$ to its dual, it is called a self-dual (or unimodular) lattice. For instance, $E_{8}$ and $\Lambda_{24}$ are self-dual. 

The main lattice parameters are depicted on Figure~\ref{fig_prez_lat}.
The black arrows represent a basis $\mathcal{B}$.  The shaded area is the parallelotope $\PPB$. The facets of the Voronoi region are shown in red. In this example, the Voronoi region has six facets generated by the perpendicular bisectors with six neighboring points.
The two circles represent the packing sphere of radius $\rho(\Lambda)$ and the covering sphere of radius $R(\Lambda)$ respectively, $R(\Lambda)>\rho(\Lambda)$.  The kissing number $\tau$ of this lattice is 2 and the Voronoi number $\tau_f$ is 6. In this case, all Voronoi vectors are relevant. \\

\begin{figure}
  \centering
	\includegraphics[width=0.45\columnwidth]{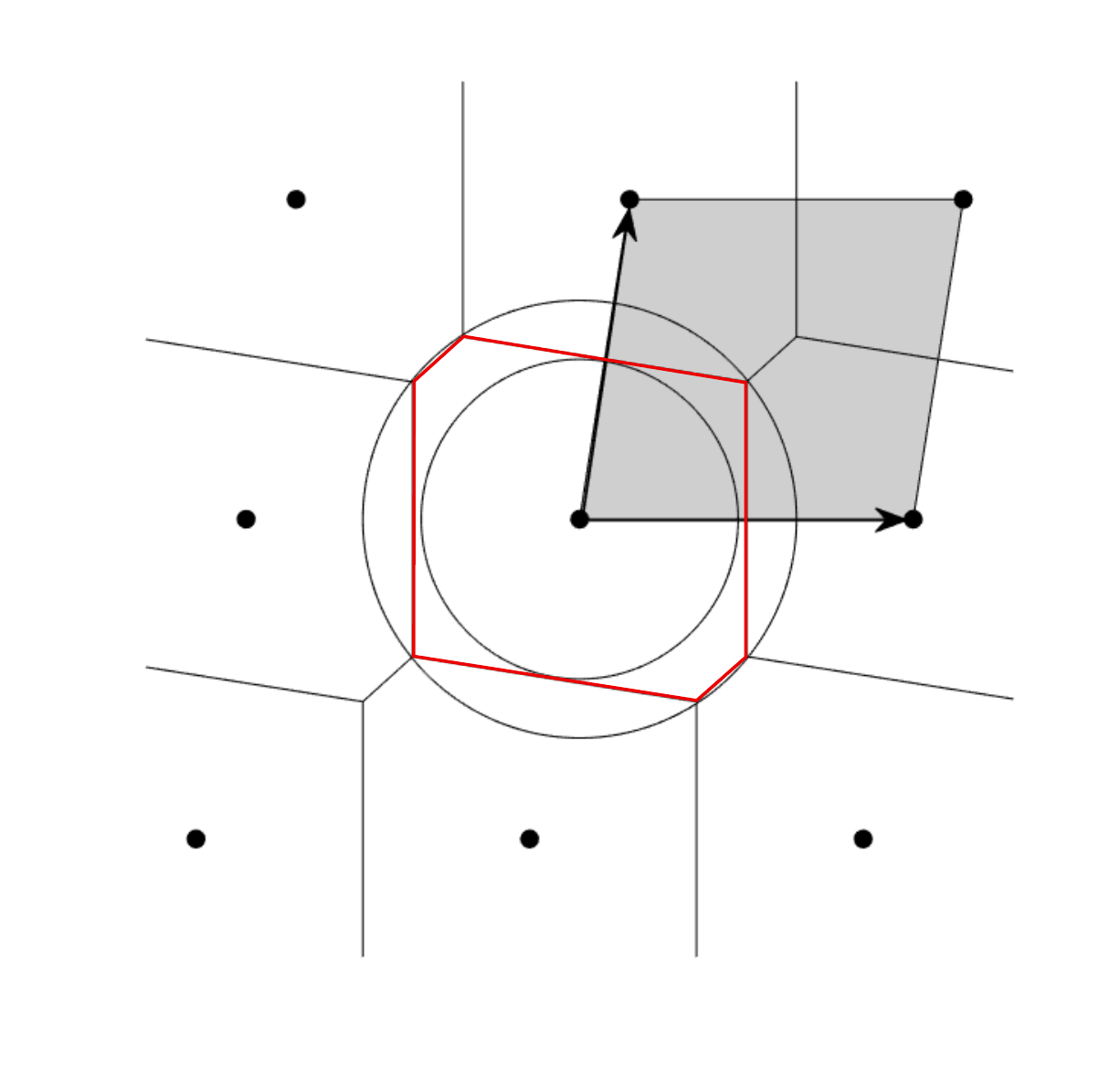}
  \caption{ Illustration of the main parameters of a lattice. }
  \label{fig_prez_lat}
\end{figure}

%The neural lattice decoder employs $\PPB$ as its main compact region \cite{Corlay2018},
%thus it is important to characterize $\PPB$ as made below.

%-----------------------------------------------------------------------------------------------------
\noindent \textbf{Geometry.} 
Let $\PPBbar$ be the topological closure of $\PPB$ and $\PPBo$ the interior of $\PPB$. A $k$-dimensional element of $\PPBbar\setminus \PPBo$
is referred to as $k$-face of $\PPB$. There are $2^n$ 0-faces, called corners or vertices.
This set of corners is denoted $\mathcal{C}_{\PPB}$.
The subset of $\mathcal{C}_{\PPB}$ obtained with $z_i=1$ is
$\mathcal{C}^1_{i,\PPB}$ and $\mathcal{C}^0_{i,\PPB}$ for $z_i=0$. 
To lighten the notations, we shall sometimes use $\mathcal{C}^1_{\PPB}$ and $\mathcal{C}^0_{\PPB}$.

The remaining $k$-faces of $\PPB$, $k>0$, are parallelotopes. For instance, a $(n-1)$-face of $\PPB$,
say $\F_{i}$, is itself a parallelotope of dimension $n-1$ defined by $n-1$ vectors of $\B$. 
Throughout the paper, the term facet refers to a $n-1$-face.

%There are $\binom{n}{n-1}$ ways to choose $n-1$ vectors of the basis, and each facet has a parallel facet. Therefore, we get that there are $2n$ facets. 
%The number of $k$-faces, say $N_{k}$, of $\PPB$ can be computed as \cite{Coxeter1973} $N_{k}=2^{n-k} \binom{n}{k}.$

%A polyhedron is defined as the intersection of a finite number of half-spaces generated by hyperplanes:
%\[
%P_{o}=\{x \in \R^{n} : \ xA \leq b, \ A \in \R^{n \times m}, \ b \in \mathbb{R}^{m}\}.
%\]
%A polyhedron is $convex$ if none of the bounding hyperplanes crosses its interior.
%A convex polyhedron is called a {\em polytope} \cite{Coxeter1973}.
Let $v_j$ denote the vector orthogonal to the hyperplane 
\begin{align}
\label{equ_vec_vj}
\{ y \in \R^n : \ y \cdot v_j - \bias_j =0  \}.
\end{align}
A polytope (or convex polyhedron) is defined as the intersection of a finite number of half-spaces %bounded by hyperplanes 
(as in e.g. \cite{Coxeter1973})
\begin{align}
P_{o}=\{x \in \R^{n} : \ x \cdot A \leq b, \ A \in \R^{n \times m}, \ b \in \mathbb{R}^{m}\},
\end{align}
where the columns of the matrix $A$ are $m$ vectors $v_j$. \\
Since a parallelotope is a polytope, it can be alternatively defined from its bounding hyperplanes. Note that the vectors orthogonal to the facets of $\mathcal{P}(\mathcal{B})$ are basis vectors of the dual lattice. Hence, a second useful definition for $\mathcal{P}(\mathcal{B})$ is obtained through the basis of the dual lattice:
\begin{align}
\begin{split}
\mathcal{P}(\mathcal{B})= &\{x \in \mathbb{R}^{n} : \ x \cdot G^{-1} \geq 0 \ , \ \ x \cdot G^{-1}  \leq 1, \\
& G \in \mathbb{R}^{n \times n}\},
\end{split}
\end{align}
where each column vector of $G^{-1}$ is orthogonal to two facets of $\PPB$ and  $(G^{-1})^{T}$ is a basis for the dual lattice of $\Lambda$.

%In this paper, parallelotopes are not the only polytopes considered as we also use simplices.
%A $i$-simplex associated with $\mathcal{U}=\{u_{j}\}_{j=0}^{i}$ is given by
%\begin{align}
%\label{eq_fake_simp}
%\begin{split}
%\mathcal{S}(\mathcal{U})= \{ & y \in \R^{n} :  \ y= \sum_{j=1}^{i} \alpha_{j}(u_j-u_0), \\
%&\sum_{j=1}^{i} \alpha_{j} \leq1, \ \alpha_{j} \geq 0 \  \forall \ j   \}.
%\end{split}
%\end{align}
%By abuse of terminology, the definition of \eqref{eq_fake_simp} is maintained even if the 
%vectors in the set $\mathcal{U}$ are not affinely independent. In this latter case, we refer to $i$ as the size of the simplex whereas it is its dimension otherwise.
%It is clear that the corners of $S(\mathcal{U})$ are the $i+1$ points of $\mathcal{U}$. 
%The number of $(k)$-faces of a $(n)$-simplex is $\binom{n+1}{k+1}$ .

%A reflection of a point $y$ through a hyperplane perpendicular to a vector $u$ is  
%\begin{equation}
%\label{eq_ref}
%s_{u}(y) = y- 2 \cdot \frac{ y\cdot u}{||u||^{2}}\cdot u.
%\end{equation}

We say that a function $g : \mathbb{R}^{n-1} \rightarrow \mathbb{R}$ is CPWL 
if there exists a finite set of polytopes covering $\mathbb{R}^{n-1}$,
and $g$ is affine over each polytope. 
The number of pieces of $g$ is the number of distinct polytopes partitioning its domain.

$\vee$ and $\wedge$ denote respectively the maximum and the minimum operator. 
We define a convex (resp. concave) CPWL function formed by a set of affine functions
related by the operator $\vee$ (resp. $\wedge$). 
If $\{g_{k}\}$ is a set of $K$ affine functions,
the function $f=g_{1} \vee ... \vee g_{K}$ is CPWL and convex. \\

\noindent \textbf{Lattice decoding.} 
Optimal lattice decoding refers to finding the closest lattice point,
the closest in Euclidean distance sense.
This problem is also known as the CVP.
Its associated decision problem is NP-complete \cite[Chap.~3]{Micciancio2002}.

Let $x \in \Lambda$ and $\eta$ be a Gaussian vector where each component is i.i.d $\mathcal{N}(0,\sigma^2)$.  Consider $y \in \R^n$ obtained as
%Find the closest lattice point to
\small
\begin{align}
\label{eq_gauss_channel}
y = x + \eta.
\end{align}
\normalsize
Since this model is often used in digital communications, $x$ is referred to as the transmitted point, $y$ the received point, and the process described by \eqref{eq_gauss_channel} is called a Gaussian channel.
Given equiprobable inputs, maximum-likelihood decoding (MLD) on the Gaussian channel is equivalent to solving the CVP.
Moreover, we say that a decoder is quasi-MLD (QMLD) if $\mathcal{P}_{dec}(\sigma^2) \leq \mathcal{P}_{opt}(\sigma^2) \cdot (1 + \epsilon)$, where $\epsilon>0$.

In the scope of (infinite) lattices, the transmitted information rate and the signal-to-noise ratio based on the second-order moment are pointless. %infinite lattice constellations
Poltyrev introduced the generalized
capacity \cite{Poltyrev1994} \cite{Zamir2014}, the analog of Shannon
capacity for lattices. The Poltyrev limit corresponds to a noise variance of
$\sigma^2_{max}=\vol(\Lambda)^{\frac{2}{n}}/(2 \pi e)$. 
The point error rate on the Gaussian channel is therefore evaluated 
with respect to the distance to Poltyrev limit, also called the volume-to-noise ratio (VNR) \cite{Zamir2014}, i.e. 
\begin{align}
\Delta= \frac{\sigma^2_{max}}{\sigma^2}.
\end{align} 
The reader should not confuse this VNR $\Delta$ with the standart notation of the lattice sphere packing density as in Section~1.2 of \cite{Conway1999}.
Using the union bound with the Theta series (see \eqref{equ_theta_series}), the MLD probability of error
per lattice point of lattice $\Lambda$ can be bounded from above by \cite[Chap.~3, Section~1.3, (19)]{Conway1999}
\begin{equation}
\label{eq_theta_s}
P_e(opt) \le P_e(ub),
\end{equation}
where \cite[Chap.~3, Section~1.4, (19) and (35)]{Conway1999}
\begin{align}
\label{equ_bound_prel}
P_e(ub) = \frac{1}{2}\Theta_{\Lambda}\left(\exp(-\frac{1}{8\sigma^2})\right) - \frac{1}{2}=\frac{1}{2} \sum_{x \in \Lambda \setminus \{ 0 \} } \exp\left( -\frac{\|x\|^2}{8\sigma^2} \right).
\end{align}
It can be easily shown that $\frac{\rho^2}{2\sigma^2}=\frac{\pi e \Delta\gamma}{4}$.
For $\Delta \rightarrow \infty$, the term $\tau q^{4\rho^2}$ dominates the sum in $\Theta_{\Lambda}(q)$ \cite[Chap.~3, Section~1.4, (21)]{Conway1999}. As proven in Appendix~\ref{App_proof_small_o}, \eqref{equ_bound_prel} becomes
\begin{align}
\label{eq_pe_opt}
P_e(ub)  =~\frac{\tau}{2} \exp(-\frac{\pi e \Delta\gamma}{4})
  + o\left(\exp(-\frac{\pi e \Delta\gamma}{4})\right).
\end{align}
%For moderate $\Delta$, only the lattice shells within $\approx 3$ dB of the first lattice shell (i.e. at a squared distance of $2\cdot d^2(\Lambda)$ from the origin)  need to be considered to get a good estimate of the MLD performance.

Finally, lattices are often used to model MIMO channels \cite[Chap.~15]{Proakis2008}.
Consider a flat  quasi-static MIMO channel with  $n/2$ transmit antennas and $n/2$
receive antennas. 
Any complex matrix of size $n/2$ can be trivially transformed into a real matrix of size $n$.
Let $G$ be the  $n \times n$ real matrix representing the
channel coefficients.  
Let $z  \in \Z^n$ be  the channel
input,  i.e., $z$  is  the uncoded  information  sequence.  The  input
message yields  the output  $y \in  \R^n$ via  the standard  flat MIMO
channel equation,
\vspace{-1mm}
\begin{align*}
y = \underset{x}{\underbrace{z \cdot G}} + \eta.
\end{align*}
%where $\eta$  Gaussian vector with i.i.d. $\mathcal{N}(0, \frac{N_0}{2})$ components. 
A MIMO lattice shall refer to a lattice generated by a matrix $G$ representing a MIMO channel. \\

\noindent \textbf{Neural networks.}
Given $n$ scalar inputs $y_1,...,y_n$ a perceptron performs the operation $\sigma(\sum_i w_i \cdot y_i)$ \cite[Chap.~1]{Goodfellow2016}.
The parameters $w_i$ are called the weights or edges of the perceptron and $\sigma(\cdot)$ is the activation function.
The activation function $\sigma(x)=\max(0,x)$ is called ReLU. A perceptron can alternatively be called a neuron. \\
Given the inputs  $y=(y_1,...,y_n)$, a feed-forward neural network simply performs the operation \cite[Chap.~6]{Goodfellow2016}:
\begin{align}
\label{eq_ff_nn}
\hat{z} = \sigma_d(... \sigma_2(\sigma_1(y \cdot G_1 + b_1)\cdot G_2 +b_2) \cdot ... \cdot G_d+b_d),
\end{align}
where:
\begin{itemize}
\item $d$ is the number of layers of the neural network.
\item Each layer of size $m_i$ is composed of $m_i$ neurons. The weights of the neurons in the $i$th layer are stored in the $m_i$ columns of the matrix $G_i$. The vector $b_i$ represents $m_i$ biases. 
\item The activation functions $\sigma_i$ are applied componentwise.
\end{itemize}

%%-----------------------------------------------------------------------------------------------------

\input{1_main_content}

\appendix

\input{2_Appendix}

%%%%%%%%%%%%%%%%%%%%%%%%%%%%%%%%%

%%%%%%%%%%%%%%%%%%%%%%%%%%%%%%%%%%

\end{document}

%% file: 0_0_abstract.tex
We characterize the complexity of the lattice decoding problem from a neural network perspective.
The notion of Voronoi-reduced basis is introduced to restrict the space of solutions to a binary set.
On the one hand, this problem is shown to be equivalent to computing a continuous piecewise linear (CPWL) function restricted to the fundamental parallelotope. 
On the other hand, it is known that any function computed by a ReLU feed-forward neural network is CPWL.
As a result, we count the number of affine pieces in the CPWL decoding function to characterize the complexity of the decoding problem.
It is exponential in the space dimension $n$, which induces shallow neural networks of exponential size. 
For structured lattices we show that folding, a technique equivalent to using a deep neural network, enables to reduce this complexity from exponential in $n$ to polynomial in $n$. Regarding unstructured MIMO lattices, in contrary to dense lattices many pieces in the CPWL decoding function can be neglected for quasi-optimal decoding on the Gaussian channel. This makes the decoding problem easier and it explains why shallow neural networks of reasonable size are more efficient with this category of lattices (in low to moderate dimensions).

%% file: 0_1_intro.tex
\section{Introduction}

In 2012 Alex Krizhevsky and his team presented a revolutionary deep neural network in the ImageNet Large Scale Visual Recognition Challenge \cite{Krizhevsky2012}. The network largely outperformed all the competitors. This event triggered not only a revolution in the field of computer vision but has also affected many different engineering fields, including the field of digital communications. 

In our specific area of interest, the physical layer, countless studies have been published since 2016. 
For instance, reference papers such as \cite{Hoydis2017} gathered more than 800 citations in less than three years.
However, most of these papers present simulation results: 
e.g. a decoding problem is set and different neural network architectures are heuristically considered.
Learning via usual gradient-descent-like techniques is performed and the results are presented.

Our approach is different: we try to characterize the complexity of the decoding problem that should be solved by the neural network.

Neural network learning is about two key aspects: first, finding  a function class
$\Phi=\{f\}$ that contains a function ``close enough'' to a target function $f^*$.
Second, finding a learning algorithm  for the class $\Phi$.
%The choices of (i) and (ii) can be either done jointly or
%separately but in either case they impact each other.
Naturally, the less ``complex'' the target function $f^*$, the easier the problem is.
We argue that understanding this function $f^*$ encountered in the scope of the decoding problem
is of interest to find new efficient solutions.

Indeed, the first attempts to perform decoding operations with “raw” neural networks (i.e. without
using the underlying graph structures of existing sub-optimal algorithms, as done in \cite{Nachmani2016}) were unsuccessful. 
For instance, an exponential number of neurons in the network is needed in \cite{Gruber2017} to achieve satisfactory performance when decoding small length polar codes. 
We made the same observation when we tried to decode dense lattices typically used for channel coding \cite{Corlay2018}.
So far, it was not clear whether such a behavior is due to either an unadapted learning algorithm or a consequence of the complexity of the function to learn.
However, unlike for channel decoding (i.e. dense lattice decoding), neural networks can sometimes be successfully trained in the scope of multiple-input multiple-output (MIMO) detection \cite{Samuel2017}\cite{Corlay2018}. 
Note that it is also possible to unfold existing iterative algorithms to establish the neural network structure for MIMO detection as done in \cite{He2020}. For lattices in reasonable number of dimensions it is possible to maintain sphere decoding but tune its parameters via a neural network \cite{Mohammadkarimi2019}, this is outside the context of our study.

In this paper, the problem of neural-network lattice decoding is investigated.
Lattices are well-suited to understand these observed differences as they can be used both for channel coding and for modelling MIMO channels.

We embrace a feed-forward neural network perspective. %: the function class $\Phi$ is restricted 
%to the functions computable by this category of neural networks.
%All the operations performed by these neural networks are discriminations with respect to hyperplanes:
These neural networks are aggregation of perceptrons and compute a composition of the functions executed by each perceptron.
% where each perceptron performs a projection on a vector orthogonal on a hyperplane followed by a discrimination which depends on the activation function used. 
%If a ReLU activation function is used feed-forward neural networks.
For instance, if the activation functions are rectified linear unit (ReLU), each perceptron computes a piecewise affine  function.
%Since a ReLU neural network computes a
%composition  of   piecewise affine  functions,  
Consequently, all  functions   in  the function class $\Phi$ of this feed-forward neural network  are CPWL.

We shall see that, under some assumptions, the lattice decoding problem is equivalent to computing a CPWL function. 
The target $f^*$ is thus CPWL.
The complexity of $f^*$ can be assessed, for instance, by counting its number of affine pieces.

It has been shown that the minimum size of shallow neural networks, such that $\Phi$ contains a given CPWL function $f^*$, directly depends on the number of affine pieces of $f^*$ whereas deep neural networks can ``fold'' the function and thus benefit of an exponential complexity reduction \cite{Montufar2014}.
On the one hand, it is critical to determine the number of affine pieces in $f^*$ to figure out if shallow neural networks can solve the decoding problem.
On the other hand, when this is not the case, we can investigate if there exist preproccessing  techniques to reduce the number of pieces in the CPWL function. We shall see that these preprocessing techniques are sequential and thus involve deep neural networks.

Due to the nature of feed-forward neural networks, our approach is mainly geometric and combinatorial. 
It is restricted to low and moderate dimensions. 
Again, our main contribution is not to present new decoding algorithms but to provide a better understanding of the decoding/detection problem from a neural network perspective.

The paper is organized as follows.
Preliminaries are found in Section~\ref{sec_prele}. 
We show in Section~\ref{sec_deco_step} how the lattice decoding problem can be restricted to the compact set $\PPB$.
This new lattice decoding problem in $\PPB$ induces a new type of lattice-reduced basis. The category of basis, called Voronoi-reduced basis, is presented in Section~\ref{sec_vr_proof}. In Section~\ref{sec_deco}, we introduce the decision boundary to decode componentwise. The discrimination with respect to this boundary can be implemented via the  hyperplane logical decoder  (HLD) also presented in this section. It is proved that, under some assumptions, this boundary is a CPWL function with an exponential number of pieces. Finally, we show in Section~\ref{sec_complex_reduc} that this function can be computed at a reduced complexity via folding with deep neural networks, for some famous dense lattices. We also argue that the number of pieces to be considered for quasi-optimal decoding is reduced for MIMO lattices on the Gaussian channel, which makes the problem easier.

%In  this  scope, several  papers  investigate  specifically deep  ReLU
%neural networks \cite{Pascanu2013}\cite{Montufar2014}\cite{Raghu2016}\cite{Telgarsky2016}\cite{Safran2017}\cite{Arora2018}
%(See \cite[Section~2.1]{Montufar2014} for a short introduction to ReLU neural networks).
%Since a ReLU neural network computes a
%composition  of   piecewise affine  functions,  all  functions   in  $\Phi$  are
%continuous  piecewise  linear   (CPWL).
%Hence,  the efficiency of $\Phi$ can be evaluated by checking whether
%a CPWL function with a lot of affine pieces belongs to $\Phi$.
%For example, there exists at least  one function  in
%$\R^n$  with $\Omega  \left( \left(  w/n \right)^{L-1}  w^n \right)$
%affine pieces that can be computed with a $w$-wide deep ReLU neural network
%having $L$ hidden layers~\cite{Montufar2014}.
%A two-layer neural network would need an exponential number of parameters for this same function.

%% file: 0_main_results.tex
\begin{itemize}
\item We first state a new closest vector problem (CVP), where the point to decode is restricted to the fundamental parallelotope $\mathcal{P}(\B)$. See Problem~\ref{pb_CVP_P}. 
This problem naturally induces a new type of lattice basis reduction, where the corresponding basis is called Voronoi-reduced basis. See Definition~\ref{def_Voronoi-reduced}.
In Section~\ref{sec_vr_proof}, we prove that some famous dense lattices admit a Voronoi-reduced basis. We also show that it is easy to get quasi-Voronoi-reduced bases for random MIMO lattices up to dimension $n=12$.
\item A new paradigm to address the CVP problem in  $\mathcal{P}(\B)$ is presented.
We introduce the notion of decision boundary in order to decode componentwise in $\mathcal{P}(\B)$. This decision boundary partition $\mathcal{P}(\B)$ into two regions. The discrimination of a point with respect to this boundary enables to decode.
The hyperplane logical decoder (HLD, see Algorithm~\ref{alg_hld}) is a brute-force algorithm which computes the position of a point with respect to this decision boundary. The HLD can be viewed as a shallow neural network.
\item In Section~\ref{sec_complex_ana}, we show that the number of affine pieces in the decision boundary grows exponentially with the dimension for some basic lattices such as $A_n$, $D_n$, and $E_n$ (see e.g. Theorem~\ref{theo_nbReg_Lin}). This induces both a HLD of exponential complexity and a shallow (one hidden layer) neural network of exponential size (Theorem~\ref{theo_shallow}).
% Moreover, Theorem~\ref{theo_shallow} shows that any ReLU neural network with one hidden layer needs an exponential number of neurons to compute this function.
\item  In Section~\ref{sec_fold}, in order to compute the decision boundary function in polynomial time,  the folding strategy is utilized (see Theorems~\ref{theo_An_linear}-\ref{theo_En_folding} for new results of folding applied to lattices). The folding strategy can be naturally implemented by a deep neural network.
%\item Using some of the above results, we highlight the advantage of depth over width in neural network architectures. We prove the following dimension dependent separation result between shallow and deep neural networks: There exists a function (the decision boudary of $A_n$), computable by a deep neural network of polynomial size, where any shallow neural network induces an exponential (in the dimension) approximation error unless it is of polynomial size. See Theorem~\ref{theo_separation}.
\item Regarding less structured lattices such as those considered in the scope of MIMO, we argue that the decoding problem on the Gaussian channel, to be addressed by a neural network, is easier compared to decoding dense lattices (in low to moderate dimensions). Namely, only a small fraction of the total number of pieces in the decision boundary function should be considered for quasi-optimal decoding. As a result, smaller shallow neural networks can be considered for random MIMO lattices, which makes the training easier and the decoding complexity reasonable.
\end{itemize}

%% file: 1_main_content.tex
\section{From the CVP in $ \mathbb{R}^n$ to the CVP in $\PPB$.}
\label{sec_deco_step}

It is well known in lattice theory that $\R^n$ can be partitioned as $\R^n=\bigcup_{x \in \Lambda} (\mathcal{P}(\mathcal{B})+x)$.
The parallelotope to which a point $y_{0} \in \R^n$ belongs is:
\begin{align}
y_{0} \in  \mathcal{P}(\mathcal{B})+x,
\end{align}
with
\begin{align}
	 x = \lfloor y_0G^{-1} \rfloor \cdot G,
\end{align}
where the floor function $\lfloor \cdot \rfloor$ is applied componentwise. 
This floor function should not be confused with the round function $\lfloor \cdot \rceil$. 
%to a vector corresponds to its application on all its coordinates.
Hence, a translation of $y_0$ by $-x$ results in a point $y$ located in the fundamental parallelotope $\PPB$. 
An instance of this operation is illustrated on Figure~\ref{fig_CVP_P}. 
The point $y \in \PPB + x$ is translated in the fundamental parallelotope (in red on the figure) to get the point $y' \in \PPB$. 
The blue arrows represent a basis $\B$ of the lattice. \\
As a result, a point $y_{0}$ to decode (e.g. in the scope of the CVP) can be processed as follows: 

\noindent {\bf Parallelotope-Based Decoding. } 
\begin{itemize}
\itemsep=-1mm
\item Step 0: a noisy lattice point $y_0=x+\eta$ is observed, where  $x \in \Lambda$ and $\eta \in \R^n$ is any additive noise.
\item Step 1: compute $t=\lfloor y_0 \cdot G^{-1} \rfloor$ and get $y=y_0-t\cdot G$ which now belongs to $\mathcal{P}(\mathcal{B})$. 
\item Step 2: find $\hz$, where $\hx=\hz \cdot G$ is the closest lattice point to $y$.
\item Step 3: the closest point to $y_0$ is $\hx_0=\hx+t \cdot G$. 
\end{itemize}

Since Step 1 and Step 3 have negligible complexity, an equivalent problem to the CVP (in $\R^n$) is the CVP in $\PPB$ (Step 2 above), which can simply be stated as follows.

\begin{problem}(CVP in $\PPB$)
\label{pb_CVP_P}
Given a point $y \in \PPB$, find the closest lattice point $\hat{x}= \hat{z} \cdot G$.
\end{problem}

 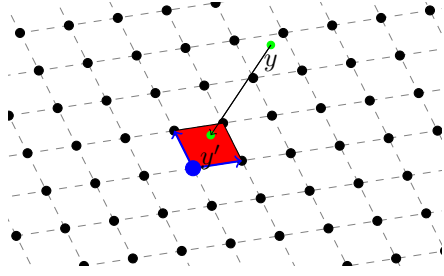
\begin{figure}
     %  \begin{center}
	\centering
          \begin{tikzpicture}[scale=0.25]
            \coordinate (Origin)   at (0,0);
            \coordinate (XAxisMin) at (-3,0);
            \coordinate (XAxisMax) at (5,0);
            \coordinate (YAxisMin) at (0,-2);
            \coordinate (YAxisMax) at (0,5);
            
            \clip (-3,-2) rectangle (20cm,12cm); % Clips the picture...
            
            \pgftransformcm{0.5}{-1}{1.3}{0.2}{\pgfpoint{0cm}{0cm}} %The lattice generator matrix
            
            \coordinate (Bone) at (0,2);
            \coordinate (Btwo) at (2,-2);
            \draw[style=help lines,dashed] (-20,-20) grid[step=2cm] (20,20);

            \foreach \x in {-10,-9,...,10}{
              \foreach \y in {-10,-9,...,10}{ 
                \node[draw,circle,inner sep=1.2pt,fill] at (2*\x,2*\y) {};
              }
            }

		\node[draw,circle,inner sep=1pt,fill,green] at (-7.5,11.3)  {};%node[below] {$y$};%{};
		\draw  (-7.5,11.3)  node[below] {$y$};
		\node[draw,circle,inner sep=1pt,fill,green] at (-3.5,7.3)  {};
		\draw [->] (-7.5,11.3) --  (-3.5,7.3) ;% node[below] {$y'$};
		\draw (-3.5,7.3)  node[below] {$y'$};

	\draw[fill=red]  (-2,6) -- (-2,8) -- (-4,8) -- (-4,6) -- cycle;

	  \node[draw,circle,inner sep=2pt,fill,blue] at (-2,6) {};
	 \draw[blue,->,thick](-2,6) -- (-2,8);
            \draw[blue,->,thick](-2,6) -- (-4,6);
	\node[draw,circle,inner sep=1pt,fill,green] at (-3.5,7.3)  {};
	\draw [->] (-7.5,11.3) --  (-3.5,7.3) ;% node[below] {$y'$};
	\draw (-3.5,7.3)  node[below] {$y'$};
	
	%\node[below=2cm] at (current bounding box.base) {caption 1};
	    %\node[below=2cm] at (current bounding box.base) {caption 1};

          \end{tikzpicture}

	%\centering
  	\caption{Translation of a noisy point into the fundamental parallelotope.}  
	%\end{center}
	 \label{fig_CVP_P}
     \end{figure}

\begin{remark}
\label{rem_ZF_1}
Consider a point $y=x + \eta$, where $\eta = \epsilon_1 g_1 + ... + \epsilon_n g_n$, $x \in \Lambda$, $0 \leq \epsilon_1, ... , \epsilon_n <1$, $g_1,...,g_n \in \B$.
Obviously, $y \in x + \PPB$. 
The well-known Zero-Forcing (ZF) decoding algorithm computes 
\begin{align}
\hat{z}=\lfloor y \cdot G^{-1} \rceil= \lfloor y_0 \cdot G^{-1} \rceil + x G^{-1}.
\end{align}
In other words, it simply replaces each $\epsilon_i$ by the
closest integer, i.e. 0 or 1.
The solution provided by this algorithm is one of the corners of the parallelotope $x + \PPB$.
\end{remark}

\begin{remark}
From a complexity theory view point, Problem~\ref{pb_CVP_P} is NP-hard. Indeed, since the above Steps 0, 1, and 3 are of polynomial complexity, the CVP, which is known to be NP-hard  \cite[Chap.~3]{Micciancio2002}, is polynomially reduced to Problem~\ref{pb_CVP_P}.
%Problem~\ref{pb_CVP_P} can be polynomially reduced to the CVP, which is known to be NP-hard (as stated in the previous section, \textcolor{blue}{see also \cite[Chap.~3]{Micciancio2002}}).
\end{remark}

\section{Voronoi-reduced lattice basis}
\label{sec_vr_proof}

\subsection{Voronoi- and quasi-Voronoi-reduced basis}

The natural question arising from Problem~\ref{pb_CVP_P} is the following: 
Is the closest lattice point to any point $y \in \PPB$  one of the corners of $\PPB$?
Unfortunately, as illustrated in Figure~\ref{fig_notVR}, this is not always the case.
The red arrows in the figure represent the basis vectors. The orange area in $\PPB$ belongs to the Voronoi region of the point $x=z \cdot G$, where $z=(-1,1)$ (in red on the figure). 
Since this lattice point is not a corner of $\PPB$, any point in this orange area, such as $y$, is not decoded to one of the corner of $\PPB$ (the four blue points on the figure).
Consequently, we introduce a new type of basis reduction.

\begin{figure}
\vspace{7mm}
\begin{center}
 \begin{tikzpicture}[scale = 1.4]

%%%%%%%%%%%%%VORO%%%%%%%%%%%%%%%%%%%%%%
    % draw the points and their cells
\def\maxxy{3.5} %2.5
  \fill [black] (-0.5, -0.5) rectangle (1.7,\maxxy);

  \def\pts{}

    \xintFor* #1 in {\xintSeq {-1}{2}} \do{
		\xintFor* #2 in {\xintSeq {-5}{5}} \do{
      \pgfmathsetmacro{\ptx}{(#1)+0.5*(#2)} % random x in [-.9\maxxy,.9\maxxy]
      \pgfmathsetmacro{\pty}{2*(#1)} % random y in [-.9\maxxy,.9\maxxy]
	     \edef\pts{\pts, (\ptx,\pty)} % stock the random point
  	}
}

%%%%%%%%%%%%%VORO%%%%%%%%%%%%%%%%%%%%%%
    % draw the points and their cells

   \xintForpair #1#2 in \pts \do{
    \edef\pta{#1,#2}
      \begin{scope}
	
        \xintForpair #3#4 in \pts \do{
          \edef\ptb{#3,#4}
          \ifx\pta\ptb\relax % check if (#1,#2) == (#3,#4) ?
            \tikzstyle{myclip}=[];
          \else
            \tikzstyle{myclip}=[clip];
          \fi;
          \path[myclip] (#3,#4) to[half plane] (#1,#2);
        }
        \clip (-0.5,-0.5) rectangle (1.7,\maxxy); % last clip
        \pgfmathsetmacro{\randhue}{rnd}
%       %\definecolor{randcolor}{hsb}{\randhue,.5,1}
        \fill[white] (#1,#2) circle (4*\biglen); % fill the cell with random color
%        %\fill[draw=red,very thick] (####1,##f2) circle (1.4pt); % and draw the point
      \end{scope}
      \pgfresetboundingbox
     %\draw (-\maxxy,-\maxxy) rectangle (\maxxy,\maxxy);
}

           \foreach \x in {0,1,...,1}{% Two indices running over each
              \foreach \y in {-1,0,...,1}{% node on the grid we have drawn 
                \node[draw,circle,inner sep=1.2pt,fill] at (\x+\y*0.5,\x*2) {};
%%                % Places a dot at those points
              }
            }
	\node[draw,circle,inner sep=1.2pt,fill] at (1,0) {};
	\node[draw,circle,inner sep=1.2pt,fill] at (-0,2) {};

	\fill[white] (0.78,1.1)--(0.78,2.9) -- (1.2,2.9)-- (1.2,1.1)-- (0.78,1.1);

	\fill[orange] (0.75,1)--(0.75,1.5) -- (0.5,1)-- (0.75,1);
			 \node[draw,circle,inner sep=1pt,fill,green] at (0.68,1.2)  {};%node[below] {$y$};%{};
		\draw  (0.68,1.2)  node[below] {$y$}; %non-VR part
	
	 \node[draw,circle,inner sep=2pt,fill,red] at (0.5,2)  {};%node[below] {$y$};%{};
	  \draw  (0.5,2)  node[above] {\scriptsize $z = (-1,1)$};  

	\node[draw,circle,inner sep=3pt,fill,blue!50] at (0,0) {};
           \node[draw,circle,inner sep=2pt,fill,blue!50] at (0.5,0) {};
	  \node[draw,circle,inner sep=2pt,fill,blue!50] at (1.5,2) {};
	\node[draw,circle,inner sep=2pt,fill,blue!50] at (1,2) {};
	  \draw[red,->,thick](0,0) -- (1,2);
           \draw[red,->,thick](0,0) -- (0.5,0);
         \end{tikzpicture}
        \end{center}
	\centering
	\vspace{4mm}
  	\caption{Example of a non-VR basis.} %As a result, this basis is not VR. 
	 \label{fig_notVR}
     \end{figure}
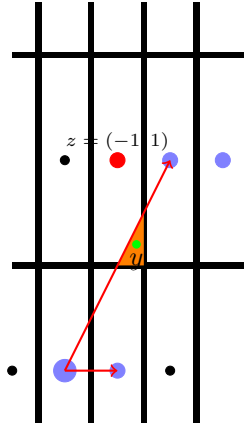

\begin{definition}
\label{def_Voronoi-reduced}
Let $\mathcal{B}$ be the $\Z$-basis of a rank-$n$ lattice $\Lambda$ in~$\R^n$.
$\mathcal{B}$ is said Voronoi-reduced if, for any point $y \in \mathcal{P}(\mathcal{B})$,
the closest lattice point $\hx$ to $y$ is one of the $2^n$ corners of $\mathcal{P}(\mathcal{B})$,
i.e. $\hx=\hz G$ where $\hz \in \{0, 1\}^n$.
\end{definition}
We will use the abbreviation {\em VR basis} to refer to a Voronoi-reduced basis.
Figure~\ref{fig_A2} shows the hexagonal lattice $A_2$, its Voronoi regions, and the fundamental parallelotope of the basis
$\mathcal{B}_1=\{v_1, v_2 \}$, where $v_1=(1, 0)$ corresponds to $z=(1, 0)$
and $v_2=(\frac{1}{2}, \frac{\sqrt{3}}{2})$ corresponds to $z=(0, 1)$.
$\mathcal{P}(\mathcal{B}_1)$ is partitioned into 4 parts included in the Voronoi regions of its corners.
$\mathcal{P}(\mathcal{B}_2)$ has 10 parts involving 10 Voronoi regions.
The small black dots in $\PPB$ represent Gaussian distributed points in $\R^2$ that have been
aliased in $\PPB$.
The basis $\mathcal{B}_1$ is Voronoi-reduced because
\begin{align}
\mathcal{P}(\mathcal{B}_1) \subset \mathcal{V}(0) \cup \mathcal{V}(v_1) \cup \mathcal{V}(v_2) \cup\mathcal{V}(v_1+v_2).
\end{align}

Lattice basis reduction is an important field in Number Theory. In general, a lattice basis is said to be of good quality when the basis vectors are relatively short and close to being orthogonal.
We cite three famous types of reduction to get a good basis:
Minkowski-reduced basis, Korkin-Zolotarev-reduced (or Hermite-reduced) basis, and $LLL$-reduced basis
for Lenstra-Lenstra-Lov\'asz \cite{Micciancio2002}\cite{Cohen1996}. A basis is said to be $LLL$-reduced if it has been processed by the $LLL$ algorithm. This algorithm, given an input basis of a lattice, outputs a new basis in polynomial time where the new basis respects some criteria, see e.g. \cite{Cohen1996}. The $LLL$-reduction is widely used in practice to improve the quality of a basis. 
%The reader may notice that a basis with all its vectors on the first lattice shell is Minkowski-reduced.
%In general, non-dense lattices do not admit a basis from the first shell.
The basis $\mathcal{B}_1$ in Figure~\ref{fig_A2} is Minkowski-, KZ-, and Voronoi-reduced.
%The basis $\{v_1, v_1+v_2\}$ of $A_2$ is not Minkowski, however it is Voronoi-reduced. 

Note that this new notion ensures that the 
closest lattice point $\hx$ to any point $y \in \PPB$ is obtained with a vector $\hz$ having only binary values (where $\hx = \hz \cdot G$). 
As a result, it enables to use a decoder with only binary outputs to optimally solve the CVP in $\PPB$.

\begin{figure}[!t]
\centering
\includegraphics[scale=0.65]{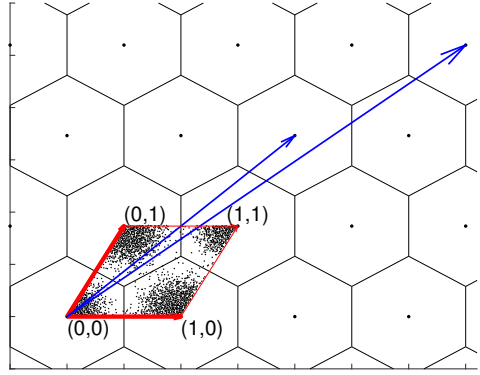}
\caption{Voronoi-reduced basis $\mathcal{B}_1$ for $A_2$ (in red) and a non-reduced basis $\mathcal{B}_2$ (in blue).
\label{fig_A2}}
\end{figure}

Unfortunately, not all lattices admit a VR basis (see the following subsection).
Nevertheless, as we shall see in the sequel, some famous dense lattices listed in \cite{Conway1999} admit a VR basis. Also, in some cases the $LLL$-reduction leads to a $quasi$-VR basis. 
Indeed, the strong constraint defining a VR basis can be relaxed as follows.

\begin{definition}
Let $\mathcal{C}(\mathcal{B})$ be the set of the $2^n$ corners of $\PPB$. Let $O$ be the subset of $\PPB$
that is covered by Voronoi regions of points not belonging to $\CB$, namely
\begin{equation}
\label{eq_reg_not_vr}
O=\PPB \setminus \left(\PPB \bigcap \left(\bigcup_{x \in \CB} V(x)\right) \right).
\end{equation}
The basis $\mathcal{B}$ is said quasi-Voronoi-reduced if $\vol(O) \ll \vol(\Lambda)$. 
\end{definition}

Let 
\begin{align}
d^2_{OC}(\mathcal{B})=\min_{x \in O, x' \in \CB} \|x-x'\|^2
\end{align}
be the minimum squared Euclidean distance between $O$ and $\CB$.
The sphere packing structure associated to $\Lambda$ guarantees that $d^2_{OC} \ge \rho^2$. 
Let $P_e(\mathcal{B})$ be the probability of error for a decoder
where the closest corner of $\PPB$ to $y$ is decoded. 
In other words, the space of solution for this decoder is restricted to $\mathcal{C}_{\PPB}$. 
%Here, we assume that $y_0=x+\eta$
%with $\eta_i \sim \mathcal{N}(0,\sigma^2)$, for $i=1, \ldots, n$. 
The following lemma tells us
that a quasi-Voronoi-reduced basis exhibits quasi-optimal performance on a Gaussian channel
at high signal-to-noise ratio.
In practice, the quasi-optimal performance is also observed at moderate values of signal-to-noise ratio. 
\begin{lemma}
\label{lem_quasi-Voronoi-reduced}
The error probability on the Gaussian channel when decoding a lattice $\Lambda$ in $\PPB$ can be bounded from above as
\begin{align}
\begin{split}
  P_e(\mathcal{B}) \le & ~P_e(ub) +  \frac{\vol(O)}{\det(\Lambda)} \cdot (e\Delta)^{n/2} \cdot
  \exp(-\frac{\pi e \Delta\gamma}{4} \cdot \frac{d^2_{OC}}{\rho^2}),\label{equ_PeO}
  \end{split}
\end{align}
for $\Delta$ large enough and where $P_e(ub)$ is defined by \eqref{eq_pe_opt}.
%where $\Delta=\frac{\det(\Lambda)^{2/n}}{2\pi e \sigma^2}$ is the distance to Poltyrev limit \cite{Poltyrev1994}, 
%$\gamma$ is the Hermite constant of $\Lambda$ \cite{Conway1999},
%and $o()$ is the small o Bachmann-Landau notation.
\end{lemma}
\begin{proof}
%Consider $P_e(opt)$ as defined in \eqref{eq_pe_opt}.
If $\mathcal{B}$ is Voronoi-reduced and the decoder works inside $\PPB$
to find the nearest corner, then the performance is given by $P_e(opt)$.\\
If $\mathcal{B}$ is quasi-Voronoi-reduced and the decoder only decides a lattice point from $\CB$,
then an error shall occur each time $y$ falls in $O$. We get
\begin{align}
\label{equ_error_term}
\begin{split}
P_e(\mathcal{B}) & \le P_e(opt) + P_e(O), \\
  &\le P_e(ub) + P_e(O).
  \end{split}
\end{align}
where
\begin{align*}
P_e(O) & =\idotsint_O \frac{1}{\sqrt{2\pi \sigma^2}^n} \exp(-\frac{\|x\|^2}{2\sigma^2}) \,dx_1 \dots dx_n\\
& \le \frac{1}{\sqrt{2\pi \sigma^2}^n} \exp(-\frac{d^2_{OC}}{2\sigma^2})~\vol(O)\\
& = \frac{\vol(O)}{\det(\Lambda)} \cdot (e\Delta)^{n/2} \cdot
  \exp(-\frac{\pi e \Delta\gamma}{4} \cdot \frac{d^2_{OC}}{\rho^2}).
\end{align*}
This completes the proof.
\end{proof}
\noindent

\subsection{Some examples}

\subsubsection{Structured lattices}
\label{sec_proofs_VR}

We first state the following three theorems on the existence of VR bases for some famous lattices. 
The proofs are provided in Appendix~\ref{App_proof_VR_bases}.

Consider a basis for the lattice $A_n$ with all vectors from the first lattice shell.
Also, the angle between any two basis vectors is $\pi/3$. Let $J_n$ denote the $n \times n$ all-ones matrix and $I_n$ the identity matrix.
The Gram matrix is
\begin{equation}
\label{eq_basis_An}
\Gamma_{A_n}=G \cdot G^T= J_n+I_n=
\left(
\begin{array}{cccccccc}
2 & 1 &1 &. . . &1 \\
1 &2& 1& . . . & 1 \\
1 &1& 2&  . . . & 1 \\
. & . & .& . .  . & . \\
1 & 1 & 1  & . . . & 2
\end{array}
\right).
\end{equation}
%\begin{equation}
%\label{eq_basis_An}
%\Gamma=GG^T=
%\left(
%\begin{array}{cccccccc}
%2 & 1 &1 &. . . &1 \\
%1 &2& 1& . . . & 1 \\
%1 &1& 2&  . . . & 1 \\
%. & . & .& . .  . & . \\
%1 & 1 & 1  & . . . & 2
%\end{array}
%\right).
%\end{equation}

\begin{theorem}
\label{theo_An}
A lattice basis of $A_n$ defined by the Gram matrix~\eqref{eq_basis_An} is Voronoi-reduced.
\end{theorem}

%\subsubsection{A VR basis for $E_8$}
Consider the following Gram matrix of $E_8$.
\begin{equation}
\label{eq_E8}
\Gamma_{E_8}=
\left(
\begin{array}{cccccccc}
  4  & 2  & 0  & 2  & 2  & 2  & 2  & 2 \\
  2  & 4  & 2  & 0  & 2  & 2  & 2  & 2 \\
  0  & 2  & 4  & 0  & 2  & 2  & 0  & 0 \\
  2  & 0  & 0  & 4  & 2  & 2  & 0  & 0 \\
  2  & 2  & 2  & 2  & 4  & 2  & 2  & 0 \\
  2  & 2  & 2  & 2  & 2  & 4  & 0  & 2 \\
  2  & 2  & 0  & 0  & 2  & 0  & 4  & 0 \\
  2  & 2  & 0  & 0  & 0  & 2  & 0  & 4
\end{array}  
\right).
\end{equation}
%We state the next theorem with respect to $\PPBo$ rather than $\PPB$ to exclude its facet (see the proof).

\begin{theorem}
\label{theo_E8}
A lattice basis of $E_8$ defined by the Gram matrix~\eqref{eq_E8} is Voronoi-reduced with respect to $\PPBo$.
\end{theorem}

%\subsubsection{The Leech lattice $\Lambda_{24}$ has no VR basis}

\begin{theorem}
\label{no_VR_24}
There exists no Voronoi-reduced basis for $\Lambda_{24}$.
\end{theorem}

%For intermediate dimensions, e.g. $n=6$, and for higher dimensions, e.g. $n=16$,
%when the existence of a VR basis cannot be proved via algebraic tools or via a tractable computer search,

Unfortunately, for most lattices such theorems can not be proved. 
However, quasi-Voronoi-reduced bases can sometimes be obtained.
For instance, the following Gram matrix corresponds to a quasi-Voronoi-reduced basis of~$E_6$:

\begin{equation}
\Gamma_{E_6}=\left(
\begin{array}{cccccc}
  3 & \frac{3}{2} & 0 & 0 &\frac{3}{2} &\frac{3}{2} \\
  \frac{3}{2} & 3 & 0 & 0 &\frac{3}{2} &\frac{3}{2} \\  
  0 & 0 & 3 & \frac{3}{2} &\frac{3}{2} &\frac{3}{2} \\
  0 & 0 &\frac{3}{2} & 3 &\frac{3}{2}&\frac{3}{2} \\
  \frac{3}{2} & \frac{3}{2} &\frac{3}{2} &\frac{3}{2} & 3 & \frac{3}{2} \\
  \frac{3}{2} & \frac{3}{2} & \frac{3}{2} &\frac{3}{2} &\frac{3}{2} & 3
\end{array}
\right),
\end{equation}
with $\frac{d^2_{OC}}{\rho^2}=1.60$ (2dB of gain) and $\frac{\vol(O)}{\det(\Lambda)}=2.47\times 10^{-3}$.
The ratio of $P_e(ub)$ by  the second term of the right-hand side of (\ref{equ_PeO}) is about $10^{-4}$ at $\Delta=1$ (0 dB) then vanishes further for increasing $\Delta$. 

Obviously, the quasi-VR property is good enough to allow the application
of a decoder working with $\CB$.% such as the hyperplane
%logical decoder presented in the next section.
If an optimal decoder is required, e.g. in specific applications such as lattice shaping and cryptography,
the user should let the decoder manage extra points outside $\CB$. For example,
the disconnected region $O$ (see \eqref{eq_reg_not_vr}) for $E_6$ defined by $\Gamma_{E_6}$ includes extra points
where $z_i \in \{ -1, 0, 1, +2\}$ instead of $\{ 0, 1\}$ as for $\CB$. 

\subsubsection{Unstructured MIMO lattices}

We investigate the VR properties of typical random MIMO lattices where the lattice is generated by a real matrix $G$ whose associated $n/2 \times n/2$ complex matrix has i.i.d. circular symetric  $\mathcal{CN}(0, 1)$ entries.
%A random lattice can be encountered, for instance, in the scope of the standard  flat MIMO channel, where
%the output  $y \in  \R^n$ is obtained from the input $z$ as $y = z \cdot G + \eta$,
%%\label{equ_MIMO_channel}
%where $\eta$ is a Gaussian vector $\mathcal{N}(0, \frac{N_0}{2})$ and 
%The matrix $G$ generates a random lattice.
The basis obtained via this random process is  in general of poor quality.
As mentioned in the previous subsection, the standard and cheap process to obtained a basis of better quality is to apply the $LLL$ algorithm.
%The resulting basis is then said to be LLL-reduced. 
As a result, we are interested in the following question: Is a $LLL$-reduced random MIMO lattice quasi-Voronoi-reduced?

In the previous subsection, we highlighted that two specific quantities characterize the loss in the error probability on the Gaussian channel ($P_e(O)$, see Equation~$\eqref{equ_error_term}$)
due to non-VR parts of $\PPB$: Vol$(O)$ and $d_{OC}(\B)$. Unfortunately, for a given basis, these quantities are in general difficult to compute because it requires
sampling in a $n$-dimensional space. In fact, one can directly estimate the term $P_e(O)$, without evaluate numerically these two quantities via Monte Carlo simulations. It is simpler to directly compute $P_e(O)$. Noisy points $y_0=x+\eta$ are generated as in Step~0 of the parallelotope-based decoding in Section~\ref{sec_deco_step}, then the shifted versions of $\PPB$ containing $y_0$ are determined as in Step~1 of the parallelotope-based decoding, and finally $y_0$ points are decoded with an optimal algorithm. 
If the decoded point is not a corner of $\PPB$, i.e. $\hat{z} \not\in \{0,1\}^n$,
we declare an error. However, if the decoded point is a corner of $\PPB$ but it is different from the transmitted lattice point $x$, we also declare an error.   
This is shown by the curves with caption named CP (for Corner Points) in Figure~\ref{fig_perf_quasi_voro_reduced}.
Comparing the resulting performance with the one obtained with the optimal algorithm enables 
to assess the term $P_e(O)$ and observe the loss in the error probability on the Gaussian channel caused by the non-VR parts of $\PPB$.

The simulation results are depicted on Figure~\ref{fig_perf_quasi_voro_reduced}  
where we show performance loss, on the Gaussian channel, due to non-VR parts of $\PPB$ for $LLL$-reduced random MIMO lattices. 
%The random generator matrix $G$ is generated with i.i.d. $\mathcal{N}(0, 1)$ components and the $LLL$-reduced basis is used for decoding. 
For each point, we average the performance over 1000 random generator matrices $G$.
Up to dimension $n=12$, considering only the corners of $\PPB$ yields no significant loss in performance. 
We can conclude that, on average for the considered model, a $LLL$-reduced basis for $n \leq 12$ is quasi-VR.
However, for larger dimensions, the loss increases and becomes significant.
On the figure, we also added the performance of the dense lattice $\Lambda_{16}$ (also called Barnes-Wall lattice in dimension 16 \cite[Chap.~4]{Conway1999}) for comparison. Obviously, the basis considered in not VR.

\begin{figure}
\begin{minipage}{.4\textwidth}
    \centering
    \includegraphics[scale=0.6]{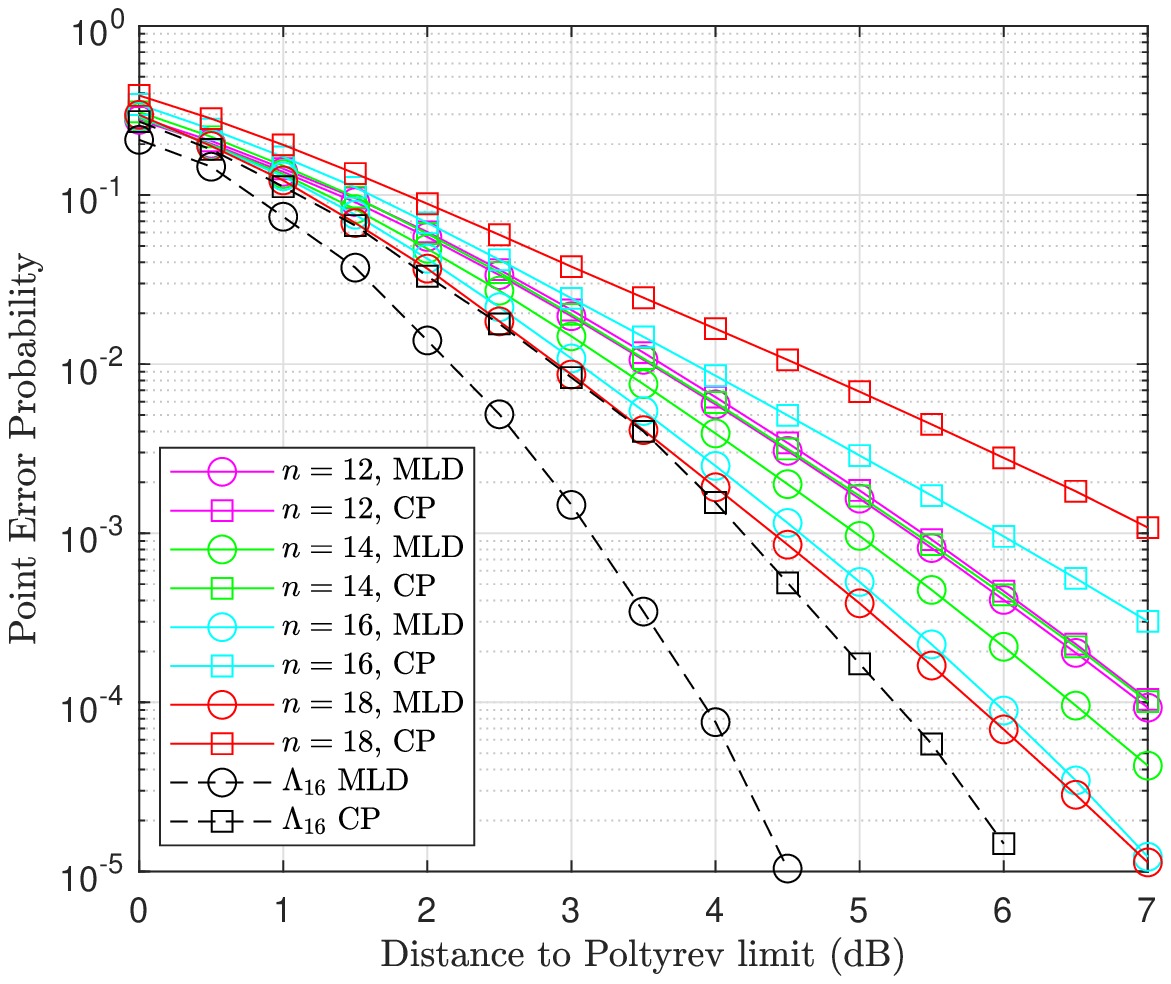}
     \caption{Assessment of the performance loss due to non-VR parts of $\PPB$.}
     \label{fig_perf_quasi_voro_reduced}
\end{minipage}%
\hspace{8mm}
\begin{minipage}{.4\textwidth}
    \centering
    \includegraphics[scale=0.6]{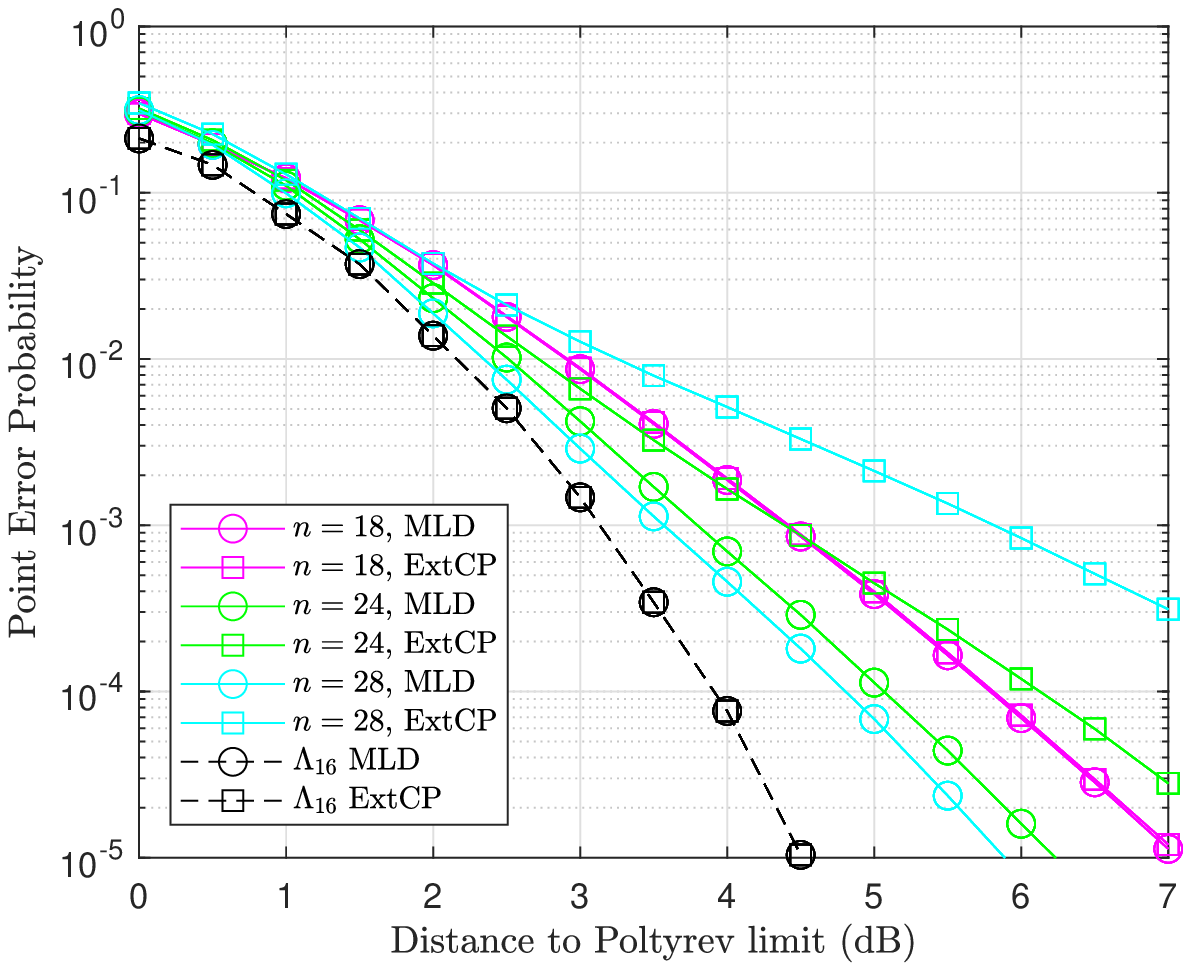}
     \caption{Assessment of the performance loss with the extended corner points.}
     \label{fig_perf_quasi_voro_reduced_ExtCP}
\end{minipage}
\end{figure}

%\begin{figure}
%    \centering
%     \includegraphics[scale=0.7]{VR_LLL_simu_true_MIMO_CP.eps}
%     \caption{\textcolor{blue}{Assessment of the performance loss due to non-VR parts of $\PPB$.}}
%     \label{fig_perf_quasi_voro_reduced}
%\end{figure}
%
%\begin{figure}
%    \centering
%     \includegraphics[scale=0.7]{VR_LLL_simu_true_MIMO_ExtCP.eps}
%     \caption{\textcolor{blue}{Assessment of the performance loss with the extended corner points.}}
%     \label{fig_perf_quasi_voro_reduced}
%\end{figure}
Figure~\ref{fig_perf_quasi_voro_reduced_ExtCP} shows the performance of a decoder with extended corner points (ExtCP) versus the maximum-likelihood decoder (MLD).
The VR concept assumes $z_i \in \{0,1 \}$. Here, the ExtCP decoder looks for the nearest lattice point slightly beyond the corners of $\PPB$ by considering $z_i \in \{-1,0,1,2\}$. 
This illustrates that the VR notion can be extended to consider $z_i$ values belonging to a larger set.

In summary,  the VR approximation can be made for a $LLL$-reduced random MIMO lattice up to dimension 12 (6 antennas)
and the extended corner-points decoding is quasi-optimal up to dimension 18 (9 antennas).

\section{Finding the closest corner of $\PPB$ for decoding}
\label{sec_deco}

Thanks to the previous section, we know that the CVP in $\PPB$, with a VR basis, can be optimally solved with an algorithm having only binary outputs.
In this section, we show how each $z_i$ can be decoded independently  in $\PPB$ via a decision boundary.
Our main objective shall be to characterize this decision boundary.
The decision boundary enables to find, componentwise, the closest corner of $\PPB$ to any point $y \in \PPB$. 
This process exactly solves the CVP if the basis is VR.
This discrimination can be implemented with the hyperplane logical decoder (HLD).
It can also be applied to lattices admitting only a quasi-VR basis
to yield quasi-MLD performance in presence of additive white Gaussian noise. 
The complexity of the HLD depends on the number of affine pieces in the decision boundary,
which is exponential in the dimension.
More generally, we shall see that this exponential number of pieces induces shallow neural networks of exponential size.
%for small dimensions without learning.

%In this section, we introduce a new lattice decoder to find the closest point

\subsection{The decision boundary}

We show how to decode one component of the vector $\hz$.
Without loss of generality, if not specified, the integer coordinate to be decoded for the rest of this section is $\hat{z}_1$.
The process presented in this section should be repeated for each $z_i$, $1 \leq i \leq n$ to recover all the components of $\hat{z}$.
Given a lattice with a VR basis, exactly half of the corners of $\PPB$ are obtained with $z_1=1$ and the other half
with $z_1=0$. Therefore, one can partition $\PPB$ in two regions, where each region is:

\begin{align}
\mathcal{R}_{\mathcal{C}^{i}_{\PPB}} = \bigcup_{x \in \mathcal{C}^{i}_{\PPB}}  \mathcal{V}(x) \cap \PPB, %\right(  \left), 
\end{align}
with $i=1$ or $0$. The intersections between $\mathcal{R}_{\mathcal{C}^{1}_{\PPB}}$ and $\mathcal{R}_{\mathcal{C}^{0}_{\PPB}}$
define a boundary. %which cuts $\PPB$ into two regions.
%It is composed of Voronoi facets of the corner points. 
This boundary splitting $\PPB$ into two regions $\mathcal{C}^0_{\PPB}$ %on one side
and $\mathcal{C}^1_{\PPB}$, %on the other side 
is the union of some of the Voronoi facets of the corners of $\PPB$.
Each facet can be defined by an affine function over a compact subset of $\mathbb{R}^{n-1}$,
and the boundary is locally described by one of these functions. 

Obviously, the position of a point to decode with respect to this boundary
determines whether $\hat{z}_1$ should be decoded to 1 or 0. For this reason,
we call this boundary the decision boundary. 
Moreover, the hyperplanes involved in the decision boundary are called boundary hyperplanes.
An instance of a decision boundary is illustrated on Figure~\ref{fig_exa_boundary} where the green point $y$, $\hat{z}_1$ should be decoded to $1$ because $y$ is above the decision boundary.

\begin{figure}
\begin{center}
\vspace{5mm}
 \begin{tikzpicture}[scale=2]
	\def\maxxy{1.6}
   \fill [black] (-0.48, -0.5) rectangle (1.7,\maxxy);

  \def\pts{}

    \xintFor* #1 in {\xintSeq {-2}{2}} \do{
		\xintFor* #2 in {\xintSeq {-2}{2}} \do{
      \pgfmathsetmacro{\ptx}{(#1)+0.5*(#2)} % random x in [-.9\maxxy,.9\maxxy]
      \pgfmathsetmacro{\pty}{0.866*(#2)} % random y in [-.9\maxxy,.9\maxxy]
	     \edef\pts{\pts, (\ptx,\pty)} % stock the random point
  	}
}

%%%%%%%%%%%%%VORO%%%%%%%%%%%%%%%%%%%%%%

   \xintForpair #1#2 in \pts \do{
    \edef\pta{#1,#2}
      \begin{scope}
	
        \xintForpair #3#4 in \pts \do{
          \edef\ptb{#3,#4}
          \ifx\pta\ptb\relax % check if (#1,#2) == (#3,#4) ?
            \tikzstyle{myclip}=[];
          \else
            \tikzstyle{myclip}=[clip];
          \fi;
          \path[myclip] (#3,#4) to[half plane] (#1,#2);
        }
        \clip (-0.5,-0.5) rectangle (1.7,\maxxy); % last clip
        \pgfmathsetmacro{\randhue}{rnd}
       %\definecolor{randcolor}{hsb}{\randhue,.5,1}
        \fill[white] (#1,#2) circle (4*\biglen); % fill the cell with random color
        %\fill[draw=red,very thick] (####1,##f2) circle (1.4pt); % and draw the point
      \end{scope}
      \pgfresetboundingbox
     %\draw (-\maxxy,-\maxxy) rectangle (\maxxy,\maxxy);
}

           \foreach \x in {0,1,...,1}{% Two indices running over each
              \foreach \y in {0,1,...,1}{% node on the grid we have drawn 
                \node[draw,circle,inner sep=2pt,fill] at (\x+\y*0.5,\y*0.866) {};
%%                % Places a dot at those points
              }
            }
	%\node[draw,circle,inner sep=1.2pt,fill] at (1,0) {};
	%\node[draw,circle,inner sep=1.2pt,fill] at (-0,2) {};
	% \node[draw,circle,inner sep=2pt,fill,red] at (0,0) {};
            \draw[red,->,thick](0,0) -- (1,0);
           \draw[red,->,thick](0,0) -- (0.5,0.866);

	 \node[draw,circle,inner sep=2pt,fill,green] at (0.5,0.55)  {};%node[below] {$y$};%{};
	  \draw  (0.5,0.52)  node[above] {$y$};

\node[draw,circle,inner sep=2pt,fill,blue] at (0,0) {};
			\draw   (0,0)  node[below] {$z_1 = 0 $};
        \node[draw,circle,inner sep=2pt,fill,blue] at (1,0) {};
			\draw (1,0)  node[below] {$z_1 = 0 $};
           \node[draw,circle,inner sep=2pt,fill,red] at (0.5,0.866) {};
			\draw  (0.5,0.866)  node[above] {$z_1 = 1 $};
	 \node[draw,circle,inner sep=2pt,fill,red] at (1.5,0.866) {};
			\draw  (1.5,0.866)  node[above] {$z_1 = 1 $};

	 \draw[orange,line width=1mm](0,0.58) -- (0.5,0.30);
	\draw[orange,line width=1mm](0.5,0.30) -- (1,0.58);
	\draw[orange,line width=1mm](1,0.58) -- (1.5,0.30);

    	\draw  (0.25,0.45)  node[below] {$h_1$};
	\draw  (0.75,0.45)  node[below] {$h_2$};
	\draw  (1.25,0.45)  node[below] {$h_3$};
         \end{tikzpicture}
        \end{center}
        \vspace{3mm}
\caption{The hexagonal lattice $A_2$ with a VR basis. 
The two upper corners of $\PPB$ (in red) are obtained with $z_1=1$ and the two other ones with
$z_1=0$ (in blue).
The decision boundary is illustrated in orange.}
\label{fig_exa_boundary}
      \end{figure}
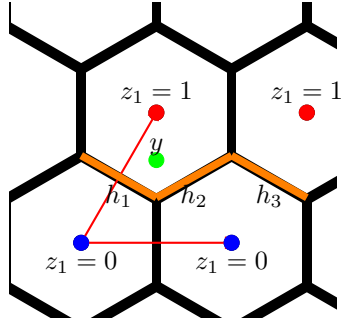

%%We call the hyperplanes involved in the decision boundary the boundary hyperplanes.

%%-------------------------

\subsection{Decoding via a Boolean equation}

Let $\mathcal{B}$ be VR basis. The CVP in $\PPB$
is solved componentwise, by comparing the position of $y$ with the Voronoi facets partitioning $\PPB$. 
This can be expressed in the form of a Boolean equation,
where the binary (Boolean) variables are the positions with respect to the facets (on one side or another).
Therefore, one should compute the position of $y$
relative to the decision boundary via a Boolean equation to guess whether
$\hat{z}_1 = 0$ or $\hat{z}_1 = 1$.

Consider the orthogonal vectors to the hyperplanes containing the Voronoi facet of a point $x \in \CBONE$ and a point from $\CN(x) \cap \CBZERO$. These vectors are denoted by $v_j$ as in \eqref{equ_vec_vj}.
A Boolean variable $u_{j}(y)$ is obtained as:
\begin{equation}
\label{eq_bool}
u_{j}(y)=\text{Heav}(y \cdot v_j - \bias_j) \in \{ 0, 1\},
\end{equation}
where $\text{Heav}(\cdot)$ stands for the Heaviside function.
Since $\mathcal{V}(x)=\mathcal{V}(0)+x$, orthogonal vectors $v_j$ to all facets partitioning $\PPB$
are determined from the facets of $\mathcal{V}(0)$.

\begin{example}
\label{ex_HLD}
 Let $\hz=(\hz_1, \hz_2)$ and $y \in \PPB$ the point to be decoded. Given the red basis on Figure~\ref{fig_exa_boundary}, the first component $\hz_1$ is $1$ (true) if
$y$ is above hyperplanes $h_1$ and $h_2$ simultaneously or above $h_3$. 
Let $u_1(y),u_2(y),$ and $u_3(y)$ be Boolean variables, the state of which depends on the location of $y$ with respect to the hyperplanes $h_1,h_2,$ and $h_3$, respectively. 
We get the Boolean equation $\hz_1=u_1(y) \cdot u_2(y) + u_3(y)$, where $+$ is a logical OR and $\cdot$ stands for a logical AND.
%Similarly, $\hz_2=d+a \cdot \overline{e}$, where $\overline{e}$ is the Boolean complement of $e$.
\end{example}

%For $\Lambda \subset \R^n$ of rank $n$, to find the Boolean equation of a coordinate $\hz_k$,
%select the $2^{n-1}$ corners of $\mathcal{P}$ where $z_{k}=1$ and perform the two following steps:

Given a lattice $\Lambda \subset \R^n$ of rank $n$, Algorithm~\ref{alg_bool_eq} enables 
to find the Boolean equation of a coordinate $\hz_i$. 
It also finds the equation of each hyperplane needed to get the value of the Boolean variables involved in the equation. 
This algorithm can be seen as a ``training'' step to ``learn'' the structure of the lattice.
It is a brute-force search that may quickly become too complex as the dimension increases. However, we shall see in Section~\ref{sec_dec_bound} and \ref{sec_complex_ana} that these Boolean equations can be analyzed without this algorithm, via a study of the basis. 
%Once found, these parameters can be loaded into a neural network similar to the one of Figure~\ref{fig_A2_net}.
Note that the decoding complexity does not depend on the complexity of this search algorithm.
%but on the size of the resulting Boolean equation.
% is not the main concern as it is done offline. 
%As we shall see below, 
\begin{algorithm}
\caption{Brute-force search to find the Boolean equation of a coordinate $\hat{z}_i$ for a lattice $\Lambda$}
\label{alg_bool_eq}
%\textbf{Input:} XXX   \\
\begin{algorithmic}[1]
\STATE Select the $2^{n-1}$ corners of $\PPB$ where $z_{i}=1$ and all relevant Voronoi vectors of $\Lambda$.
 \FOR{each of the  $2^{n-1}$ corners where $z_{i}=1$} 
 \FOR{each relevant Voronoi vector of $\Lambda$} 
\STATE {
Move in the direction of the selected relevant Voronoi vector by half its norm + $\epsilon$ ($\epsilon$ being a small number).
} 
\IF{The resulting point is outside $\PPB$.}
\STATE Do nothing. //There is no decision boundary hyperplane in this direction.
\ELSE 
\STATE Find the closest lattice point $x'=z'G$ (e.g. by sphere decoding \cite{Agrell2002}).
\IF{$z'_{i}=1$}
\STATE Do nothing. //There is no decision boundary hyperplane in this direction.
\ELSE
\STATE Store the decision boundary orthogonal to this direction. //$z'_{i}=0$  
\ENDIF
\ENDIF
\ENDFOR
\FOR{each decision boundary hyperplane found (at this corner)}
\STATE Associate and store a Boolean variable to this hyperplane (corresponding to the position of the point to be decoded with respect to the hyperplane).
\ENDFOR
\STATE The Boolean equation of $\hz_i$ contains a Boolean AND of these variables.
\ENDFOR
\STATE The equation is the Boolean OR of the $2^{n-1}$ AND coming from all corners.
\end{algorithmic}
\end{algorithm}

%For clarity reasons, we omitted technical details in the above steps in Algorithm~\ref{alg_bool_eq} that involve facets of $\PPB$ with a tie.
%Note that the outputed Boolean equation always has the form of a Disjunctive Normal Form (DNF). 
%In practice, the constructed Boolean equation with its $2^{n-1}$ terms is significantly reduced
%into a simpler equation, mainly as a result of identical terms.
%If $\Lambda$ does not admit a VR or quasi-VR basis, the HLD should be constructed from more lattice shells and some coordinates of $\hz$ are not binary anymore.

\subsection{The HLD} 
The HLD is a brute-force algorithm to compute the Boolean equation provided by Algorithm~\ref{alg_bool_eq}.
%The HLD shall operate in $\PPB$ as for Step~2 in the decoding steps
%listed in Section~\ref{sec_deco_step}.
%Once the Boolean equation of each coordinate and the decision boundary hyperplanes are known,
The HLD can be executed via the three steps summarized in Algorithm~\ref{alg_hld}. 
%By abuse of terminology,
%the inner product of two points in $\R^n$ refers to the inner product
%between the two vectors defined by these points. 

\begin{algorithm}
\caption{HLD}
\label{alg_hld}
%\textbf{Input:} XXX
\begin{algorithmic}[1]
\STATE Compute the inner product of $y$ with the vectors orthogonal to the decision boundary hyperplanes.
\STATE Apply the Heaviside function on the resulting quantities to get its relative positions
under the form of Boolean variables.
\STATE Compute the logical equations associated to each coordinate.
\end{algorithmic}
\end{algorithm}

\subsubsection{Implementation of the HLD}

Since Steps 1-2 are simply linear combinations followed by activation functions,
these operations can be written as:
\begin{align}
\label{eq_nn_1}
l_1 = \sigma (y \cdot G_1 +b_1),
\end{align}
where $\sigma$ is the Heaviside  function, $G_1$ a matrix having the vectors $v_j$ as columns, and $b_1$ a vector of biases containing the $p_j$. Equation~\eqref{eq_nn_1} describes the operation performed by a layer of a neural network (see \eqref{eq_ff_nn}) .
The layer $l_1$ is a vector containing the Boolean variables $u_j(y)$.

Let $l_{i-1}$ be a vector of Boolean variables.
It is well known that both Boolean AND and Boolean OR can be expressed as:
 \begin{align*}
l_i = \sigma(l_{i-1} \cdot G_i +b_i),
\end{align*}
where  $G_i$ a matrix composed of 0 and 1, and $b_i$ a vector of biases.
Therefore, the mathematical expression of the HLD is:

\begin{equation}
z_1 = \sigma(\sigma(\sigma(y \cdot G_1 + b_1)\cdot G_2 + b_2) \cdot G_3 + b_3).
\label{equ_feed_fo}
\end{equation}
%where the non-linear activation function $\sigma$ is the Heaviside function.
Equation~\eqref{equ_feed_fo} is exactly the definition of a feed-forward neural network (see \eqref{eq_ff_nn}) with three layers.
%the natural way to represent these steps is to use perceptrons \cite{Goodfellow2016},
%where the edges are labeled with the decision hyperplane parameters,
%i.e. the perceptron weights define the vector orthogonal to the decision hyperplane.
%(a) and (b) form the first layer of a neural network.
%The second layer implements the logical AND and the third layer the logical OR.
%As a result, the HLD can be naturally thought of as a neural network with with three layers.
Figure~\ref{fig_A2_net} illustrates the topology of the neural network obtained
when applying the HLD to the lattice $A_2$. 
Heav($\cdot$) stands for Heaviside($\cdot$). The first part of the network computes the position of $y$ with respect to the boundary hyperplanes to get the variables $u_j(y)$. The second part (two last layers) computes the Boolean ANDs and Boolean ORs of the decoding Boolean equation.

%
%\begin{example}
%\label{ex_HLD_2}
%Consider the point $y$ depicted on Figure~\ref{fig_exa_boundary}.
%We use the neural network of Figure~\ref{fig_A2_net} to recover $z_1$.
%The first layer project $y$ on the vector orthogonal to, respectively, 
%$h_1$, $h_2$, and $h_3$ (and add the bias) and apply the Hevaiside function 
%on the resulting quantity.
%The value obtained at the output of the first layer are $1,1,0$.
%\end{example}

\begin{figure}
\centering
\includegraphics[scale=0.835]{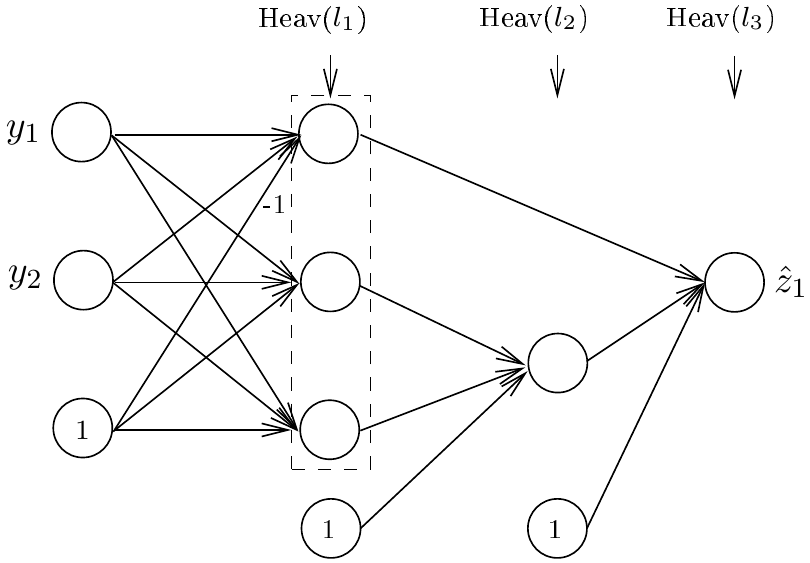}
\caption{Neural network performing HLD decoding on the first symbol $\hat{z}_{1}$ of a point in $\PPB$ for the lattice
$A_2$ (see Example~\ref{ex_HLD}). 
%Unlabeled edges have weight 1. The bias nodes are required to perform AND and OR.
}
\label{fig_A2_net}
\end{figure}

\subsection{The decision boundary as a piecewise affine function}
\label{sec_dec_bound}

In order to better understand the decision boundary, we characterize it as a function rather than a Boolean equation. We shall see in the sequel that it is sometimes possible to efficiently compute this function and thus reduce the decoding complexity.

Let $\{ e_i \}_{i=1}^n$ be the canonical orthonormal basis of the vector space $\R^n$.
For $y \in \R^n$, the $i$-th coordinate is $y_i=y \cdot e_i$.
Denote $\tilde{y}=(y_2, \ldots, y_n) \in \R^{n-1}$ and let $\mathcal{H} = \{ h_j \}$ be the set of affine functions involved in the decision boundary.
The affine boundary function $h_{j}:\R^{n-1} \rightarrow \mathbb{R}$ is
\begin{equation}
\label{equ_hj}
h_{j}(\tilde{y})=y_{1}= \bigg( \bias_{j} -\sum_{k \neq 1} y_{k}v_{j}^{k} \bigg)/v_{j}^{1},
\end{equation}
where $v_{j}^{k}$ is the $k$th component of vector $v_{j}$.
For the sake of simplicity, in the sequel $h_{j}$ shall denote the function defined
in (\ref{equ_hj}) or its associated hyperplane depending on the context. %$\{ y \in \R^n : \ y \cdot v_j - \bias_j =0  \}$

\begin{theorem}
\label{th_func_VR_2}
Consider a lattice defined by a VR basis $\B=\{g_i\}_{i=1}^n$.
Let $\mathcal{H} = \{ h_j \}$ be the set of affine functions involved in the decision boundary.
%Without loss of generality, 
Assume that $g_1^1>0$.
Suppose also that $x_1 > \lambda_1$ (in the basis $\{e_i\}_{i=1}^n$), $\forall x \in \CBONE$ and $\forall \lambda \in \CN(x) \cap \CBZERO$.
Then, the decision boundary is given by a CPWL function
$f:\R^{n-1} \rightarrow~\R$, expressed as
\begin{equation}
\label{eq_boundary_funct}
f(\tilde{y}) = \wedge_{m=1}^{M}\{\vee_{k=1}^{l_m}h_{m,k}(\tilde{y}) \},
\end{equation}
where $h_{m,k} \in \mathcal{H}$, $1 \leq l_m< \tau_f  $, and  $1 \leq M \leq 2^{n-1}$.%, $|\mathcal{H}|<~\tau_f$.
\end{theorem}
The proof is provided in Appendix~\ref{App_proof_func}.
In the previous theorem, the orientation of the axes relative to $\B$
does not require $\{g_i\}_{i=2}^n$ to be orthogonal to $e_1$.
%neither $b_1$ and $e_1$ to be collinear. 
%Note that the assumptions of Theorem~\ref{th_func_VR_2} are more general than the ones of Theorem~\ref{th_func_VR}:
This is however the case for the next corollary, which involves a specific rotation satisfying the assumption of the previous theorem.
%The orientation of the axes chosen for Corollary~\ref{th_func_VR} always satisfies: $x_1 > \lambda_1$, $\forall x \in \CBONE$ and $\forall \lambda \in \CN(x) \cap \CBZERO$ (if $b_1^1$ is negative, the equivalent assumption is: $x_1 < \lambda_1$, $\forall x \in \CBONE$ and $\forall \lambda \in \CN(x) \cap \CBZERO$).
Indeed, with the following orientation, any point in $\CBZERO$ is in the hyperplane $\{y \in \R^n : \ y \cdot e_1 =0\}$ and has its first coordinate equal to 0, and $g_1^1>0$ (if it is negative, simply multiply the basis vectors by $-1$).

\begin{figure*}[h]
\centering
\begin{subfigure}{0.4\columnwidth}
 \includegraphics[width=1\columnwidth]{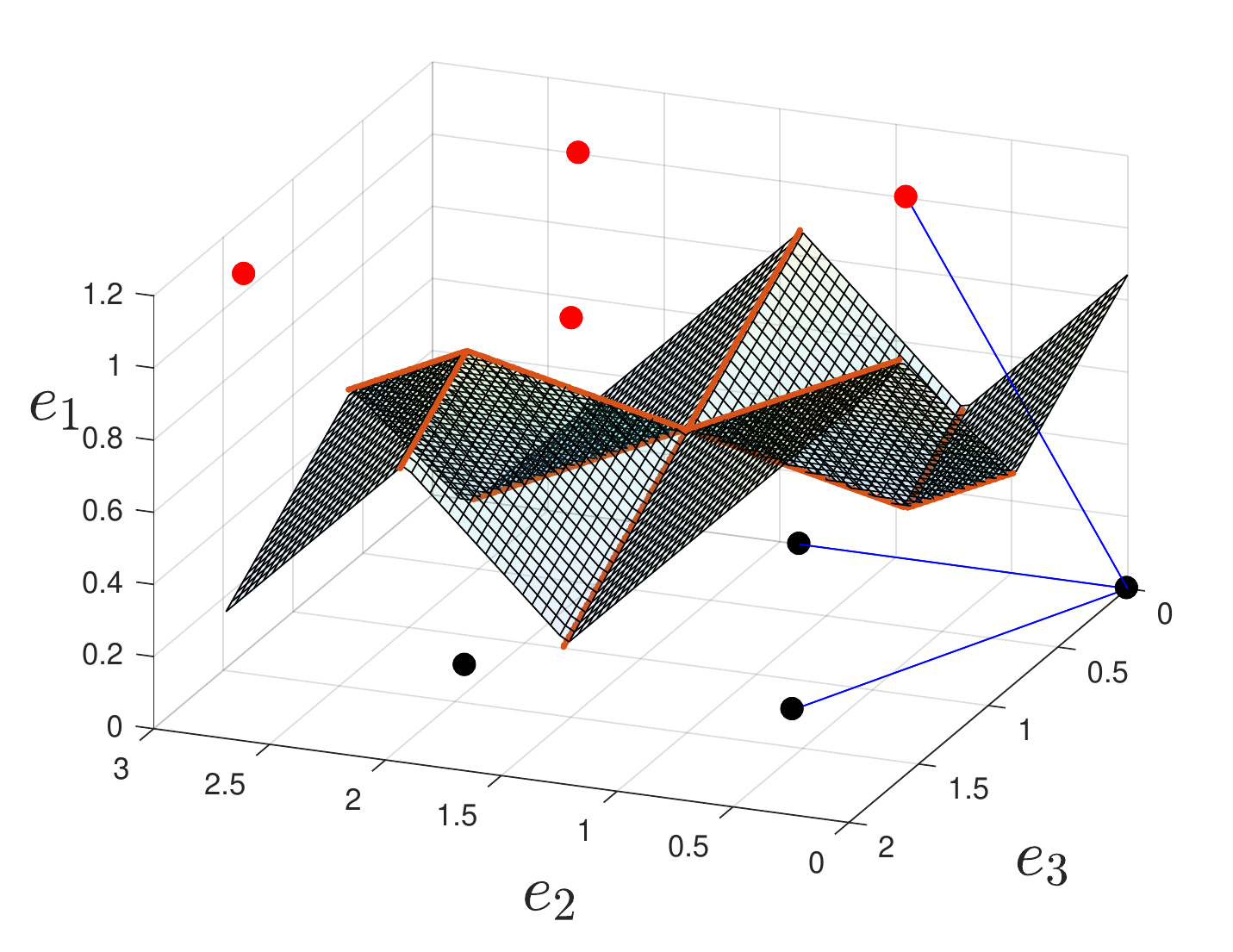}
\caption{The orientation of the basis satisfies the assumptions of Theorem~\ref{th_func_VR_2} and Corollary~\ref{th_func_VR}.}
\label{fig_rot_coro}
\end{subfigure}%
\hspace{6mm}
\begin{subfigure}{0.4\columnwidth}
\centering
\includegraphics[width=1\columnwidth]{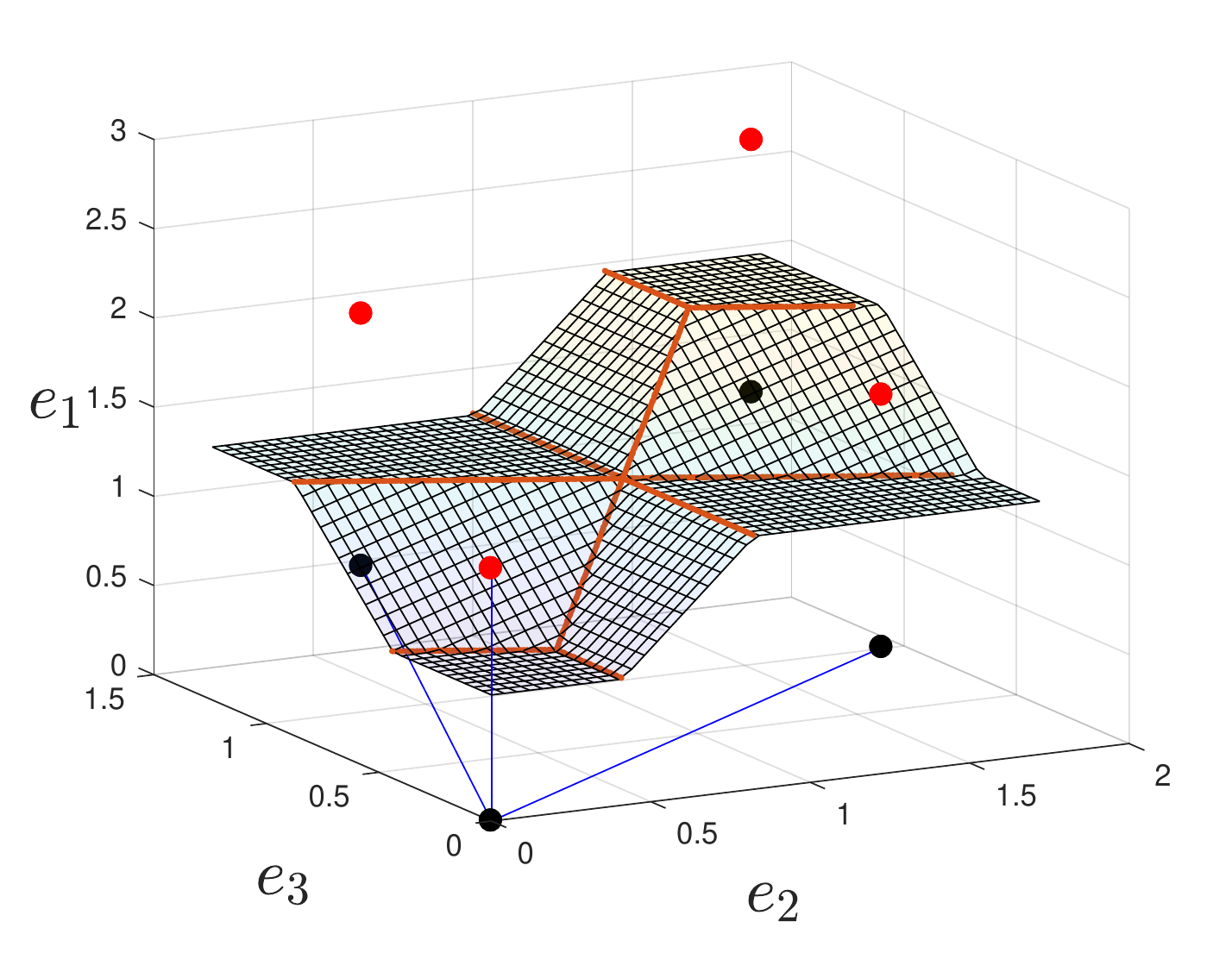}
     \caption{ The orientation of the basis satisfies the assumptions of Theorem~\ref{th_func_VR_2} but not the ones of Corollary~\ref{th_func_VR}.}
     \label{fig_func_A3}
\end{subfigure}%
\caption{CPWL decision boundary function for $A_3$. 
The basis vectors are represented by the blue lines. 
The corner points in $\CBONE$ are in red and the corner points in $\CBZERO$ in black.}
\label{fig_func_rot}
\end{figure*}

\begin{figure}[h]
    \centering
    \includegraphics[scale=0.45]{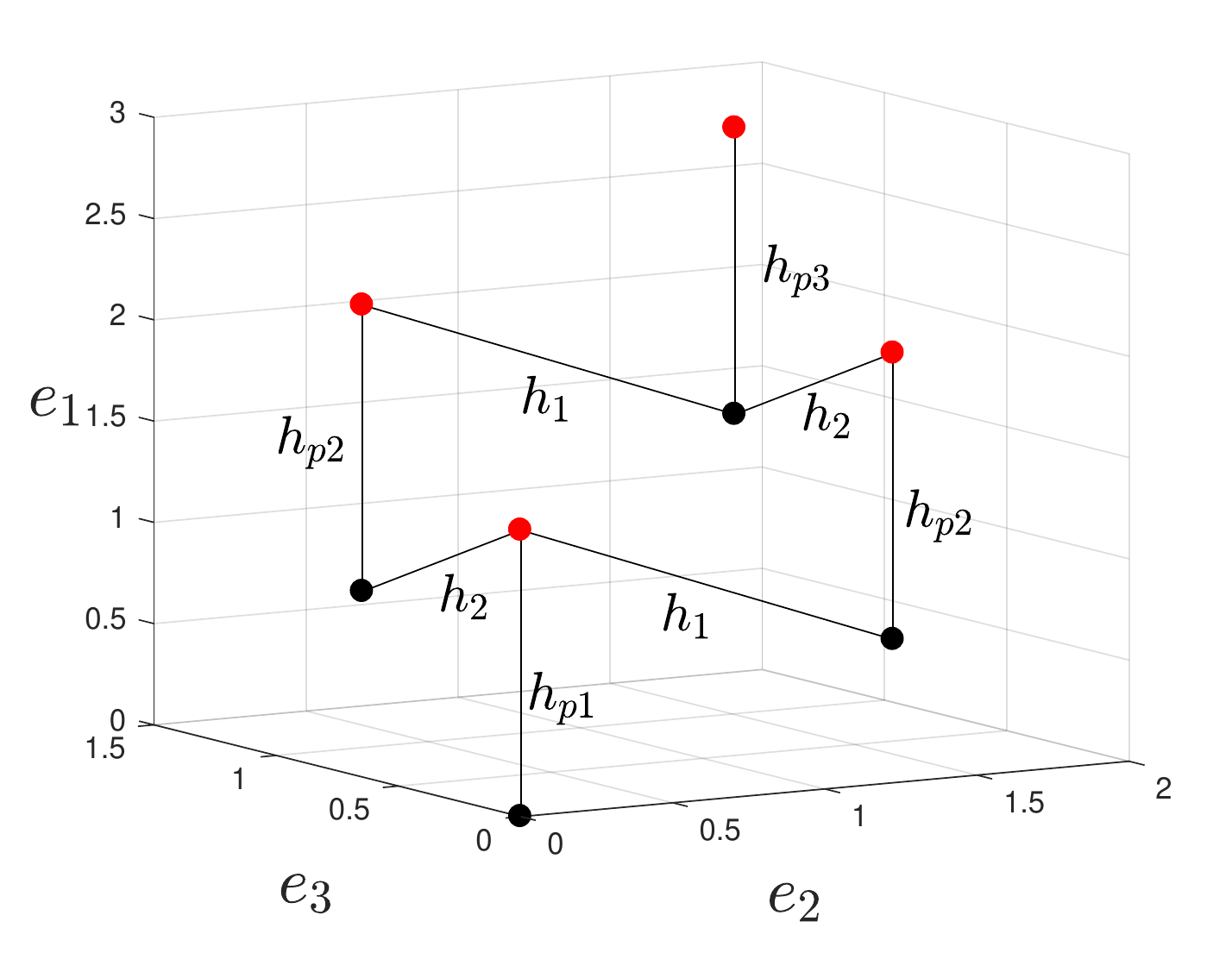}
     \caption{``Neighbor figure" of $\mathcal{C}_{\PPB}$ for the lattice $A_3$. 
		        }
     \label{fig_simplex_A3}
\end{figure}

\begin{corollary}
\label{th_func_VR}
Consider a lattice defined by a VR basis $\B=\{g_i\}_{i=1}^n$.
Suppose that the $n-1$ points $\mathcal{B} \backslash \{ g_1\}$
belong to the hyperplane $\{y \in \R^n : \ y \cdot e_1 =0\}$.
Then, the decision boundary is given by a CPWL function as in (\ref{eq_boundary_funct}).
\end{corollary}

%%-------------------------

%As a result, the proof of Theorem~\ref{th_func_VR_2} (below) also proves Theorem~\ref{th_func_VR}.

\begin{example2}
\label{ex_A3}
Consider the lattice $A_3$ defined by the Gram matrix~\eqref{eq_basis_An}. 
To better illustrate the symmetries we rotate the basis\footnote{Note that the orientation of the basis does not satisfy the assumptions of Corollary~\ref{th_func_VR}.} to have $g_1$ colinear with $e_1$. 
Theorem~\ref{th_func_VR_2} ensures that the decision boundary is a function.
The rotated function is illustrated in Figure~\ref{fig_func_A3} and the non-rotated version in Figure~\ref{fig_rot_coro}.
On Figure~\ref{fig_simplex_A3} each edge is orthogonal to a local affine function of $f$. The edges are labeled with the name of the corresponding affine function.
Each edge connects a point $x \in \CBONE$ to an element of $\CN(x) \cap \CBZERO$. 
		       % The $i$ edges connected to a point $x \in \CBONE$ are $1$-faces of a regular $i$-simplex. 
		      Consequently,  each edge is orthogonal to a local affine function of the decision boundary function $f$.  
		        The edges are labeled with the names of the corresponding affine functions.
		        Theorem~\ref{theo_nbReg_Lin} and its proof show that each set $\CN(x) \cap \CBZERO$ generates a convex part of the decision boundary function $f$ with $|\CN(x) \cap \CBZERO|$ pieces.
		        E.g. on the figure there are one set with $|\CN(x) \cap \CBZERO|=3$, two with $|\CN(x) \cap \CBZERO|=2$, and one with $|\CN(x) \cap \CBZERO|=1$, thus $f$ has 8 pieces. 
The equation of the function is (we omit the $\tilde{y}$ in the formula to lighten the notations):
\begin{align*}
\label{eq_A3}
f = &\Big[ h_{p_1} \vee h_{1} \vee h_{2} \Big] \wedge \Big[ \left( h_{p_2} \vee h_{1}\right)  \wedge \left( h_{p_2}\vee h_{2} \right) \Big] \wedge  \Big[ h_{p_3} \Big], 
\end{align*}
where $h_{p_1}$, $h_{p_2}$ and $h_{p_3}$ are hyperplanes orthogonal to $g_1$ (the $p$ index stands for plateau) and the $[ \cdot ]$ groups all the set of convex pieces of $f$ that includes the same $h_{p_j}$. 
Functions for higher dimensions (i.e. $A_n, n\ge 3$) are available in Appendix~\ref{terms_An_before}.
\end{example2}

%Some interesting lattices may not admit a VR basis, e.g. $E_6$.
%then HLD yields efficient decoding.
%A basis satisfying this condition is called {\em quasi-Voronoi-reduced}.

The notion of decision boundary function can be generalized to non-VR basis under 
the assumptions of the following definition.
%This new definition presumes that $\B$ is quasi-Voronoi-reduced in order
%to make a successful discrimination of $\hat{z}_1$ via the boundary function.
A surface in $\R^n$ defined by a function $g$ of $n-1$ arguments is written as
$\text{Surf}(g)=\{ y = ( g(\tilde{y}), \tilde{y}) \in \R^n, \tilde{y} \in \R^{n-1} \}$. 
%%-------------------------
\begin{definition}
\label{def_semi-Voronoi-reduced}
%%$\vee_{k=1}^{l}g_{m,k}(\tilde{y})$
Let $\B$ be a  is quasi-Voronoi-reduced basis of $\Lambda$. Assume that  $\B$ and $\{ e_i \}_{i=1}^n$ have the same
orientation as in Corollary~\ref{th_func_VR}. The basis is called
semi-Voronoi-reduced (SVR) if there exists at least two points $x_1, x_2 \in \CBONE$
such that 
$\text{Surf}(\vee_{k=1}^{\ell_1}g_{1,k}) \bigcap \text{Surf}(\vee_{k=1}^{\ell_2}g_{2,k})
\ne \varnothing$, where $\ell_1,\ell_2\geq 1$,
$g_{1,k}$ are the facets between $x_1$ and all points in $\CN(x_1) \cap \CBZERO$,
and $g_{2,k}$ are the facets between $x_2$ and all points in $\CN(x_2) \cap \CBZERO$.
\end{definition}
The above definition of a SVR basis imposes that the boundaries around two points of $\CBONE$,
defined by the two convex functions $\vee_{k=1}^{\ell_m}h_{m,k}$, $m=1,2$,
have a non-empty intersection.
Consequently, the min operator $\wedge$ leads to a boundary function as in (\ref{eq_boundary_funct}).
%%-------------------------
\begin{corollary}
\label{coro_SVR}
$\PPB$ for a SVR basis $\B$ admits a decision boundary defined by a CPWL function as in (\ref{eq_boundary_funct}).
\end{corollary}

From now on, the default orientation of the basis with respect to the
canonical axes of $\R^n$ is assumed to be the one of Corollary~\ref{th_func_VR}.
We call $f$ the decision boundary function.
The domain of $f$ (its input space) is $\D(\B) \subset \mathbb{R}^{n-1}$.
The domain $\mathcal{D}(\B)$ is the projection of $\PPB$ on the hyperplane $\{e_i\}_{i=2}^n$.
It is a bounded polyhedron that can be partitioned into convex regions which we call linear regions. 
For any $\tilde{y}$ in one of these regions,
$f$ is described by a unique local affine function $h_{j}$. 
The number of those regions is equal to the number of affine pieces of $f$.

%\begin{remark}
%It is interesting to compare this decision boundary used for
%optimal decoding of a component $z_i$ with the ZF decoder (see Remark~\ref{rem_ZF_1}).
%Assume that the point $y$ to be decoded is
%in $\PPB$.
%%$y=\epsilon_1 g_1 + ... + \epsilon_n g_n$, $x \in \Lambda$, $0 \leq \epsilon_1, ... , \epsilon_n <1$, $b_1,...,b_n \in \B$.
% Assume that there are $K$ distinct truncated simplices of the form 
%$\{ y \in \PPB : y_1 \ge f_m(\tilde{y}), \ y \cdot e_1 \le \frac{1}{2}\times(b_1 \cdot e_1) \}$  or $\{ y \in \PPB : y_1 \le f_m(\tilde{y}), \ y \cdot e_1 \ge \frac{1}{2}\times(b_1 \cdot e_1) \}$. XXX
%%Obvisouly $y \in x + \PPB$. The ZF decoding algorithm simply replaces each $\epsilon_i$ by the
%%closest integer, i.e. 0 or 1.
%\end{remark}

%\begin{figure}
%    \centering
%    \includegraphics[scale=0.35]{perf_An_uniform_noise.eps}
%     \caption{coucou, for $z_1$ uniform noise.}
%     \label{fig_perf_An_uniformnoise}
%
%\end{figure}
                       
\subsection{Complexity analysis: the number of affine pieces of the decision boundary}
\label{sec_complex_ana}

%\subsubsection{Counting the number of affine pieces of the decision boundary function}

An efficient neural lattice decoder should have a reasonable size, i.e. a reasonable number of neurons. 
Obviously, the size of the neural network implementing the HLD (such as the one of Figure~\ref{fig_A2_net})
depends on the number of affine pieces in the decision boundary function.
It is thus of high interest to characterize the number of pieces 
in the decision boundary as a function of the dimension. 
Unfortunately, it is not possible to treat all lattices in a unique framework. Therefore,
we investigate this aspect for some well-known lattices.
\\

\textbf{The lattice $A_n$} \\
We count the number of affine pieces of the decision boundary function $f$ obtained
for $z_1$ with the basis defined by the Gram matrix~\eqref{eq_basis_An}.

\begin{theorem}
\label{theo_nbReg_Lin}
%Rotate the basis to have $b_i$ collinear with the $i_{th}$ axis. 
Consider an $A_n$-lattice basis defined by the Gram matrix~\eqref{eq_basis_An}.
Let $o_i$ denote the number of sets $\CN(x) \cap \CBZERO$, $x \in \CBONE$, where  $|\CN(x) \cap \CBZERO|= i$.
The decision boundary function $f$ has a number of affine pieces equal to
\begin{equation}
\label{eq_An_pieces}
\sum_{i=1}^{n} i \cdot o_i,
%\sum_{i=1}^{n} \underset{(l_i)}{\underbrace{i}} \times \underset{(o)}{\underbrace{ \binom{n-1}{n-i} }}. 
\end{equation}
with $o_i =  \binom{n-1}{n-i}$.
\end{theorem}

\begin{proof}
%Recall that for any convex region of $f$, the set of affine functions is composed of the facets (i.e. the hyperplanes where the facets lie) between a point $x \in \mathcal{C}^1_{\PPB}$ and $\CN(x) \in \CBZERO$.
%It is obvious that $\forall x \in~\CBZERO$: $x + b_1 \in~\CBONE$. 
For any given point $x   \in \mathcal{C}^1_{\PPB}$, 
each element in the set  $\CN(x)~\cap~\CBZERO$ generates a Voronoi facet 
 of the Voronoi region of $x$. Since any Voronoi region is convex,
 the $|\CN(x)~\cap~\CBZERO|=i$ facets are convex. Consequently,
 the set $\CN(x)~\cap~\CBZERO$ generates a convex part of the decision boundary function with $i$ pieces.

%
%and its neighbors form a regular simplex
%$\mathcal{S}$ of dimension $|\CN(x) \cap \CBZERO|$. 
%% (A regular $k$-simplex is generated by the same basis as $A_k$)
%This clearly appears on Figure~\ref{fig_simplex_A3}.
%%, where we connected each point of $\CBONE$ to its neighbor in $\CBZERO$
%Now, consider the decision boundary function of a $i$-simplex separating the top corner (i.e. $\mathcal{C}^1_{\mathcal{S}}$) from all the other corners  (i.e. $\mathcal{C}^0_{\mathcal{S}}$). 
%This function is convex and has $i$ pieces. 

%The maximal dimensionality of such simplex is obtained
%by taking the points 0, $b_1$, and the $n-1$ points $b_j$, $j\ge 2$.
%\end{proof}

%\begin{proof}[Proof (of Theorem~\ref{theo_nbReg_Lin})]
%From Lemma~\ref{lem_simplex},
We now count the number of sets $\CN(x)~\cap~\CBZERO$ with  cardinality  $i$.
It is obvious that $\forall x \in~\CBZERO$: $x + g_1 \in~\CBONE$. 
We walk in $\CBZERO$ and for each of the $2^{n-1}$ points $x \in \CBZERO$ 
we investigate the cardinality of the set $\CN(x+~g_1)~\cap~\CBZERO$. 
This is achieved %by counting the number of elements in $\CN(x+~g_1)~\cap~\CBZERO$, 
via the following property of the basis. 
\begin{gather}
\label{eq_prop}
\begin{split}
&\forall x \in~\CBZERO, \ x' \in A_n \backslash \{g_j,0\},\ 2 \leq j\leq n: \\
&x+ g_j \in \CN(x+g_1), \ x+ x' \not\in \CN(x+g_1) \cap \CBZERO.
\end{split}
\end{gather}
Starting from the lattice point 0, the set $\CN(0+g_1) \cap \CBZERO$ is composed of $0$  and the $n-1$ other basis vectors. 
Then, for all $g_{j_1}$, $2 \leq j_1\leq n$, the sets $\CN(g_{j_1}+g_1) \cap \CBZERO$
are obtained by adding any the $n-2$ remaining basis vectors to $g_{j_1}$. 
Indeed, if we add~$g_{j_1}$ to $g_{j_1}$, the resulting point is outside of $\PPB$. 
Hence, the cardinality of these sets is $n-1$ and there are $\binom{n-1}{1}$
ways to choose $g_{j_1}$: any basis vectors except~$g_1$. 
Similarly, for $g_{j_1}+g_{j_2}$, $j_1 \neq j_2$, the cardinality of the sets  $\CN(g_{j_1}+g_{j_2} +g_1) \cap \CBZERO$ is  $n-2$ and there are $\binom{n-1}{2}$ ways to choose $g_{j_1}+g_{j_2}$. 
More generally, there are $\binom{n-1}{k}$ sets $\CN(x) \cap \CBZERO$ of cardinality $n-k$.

Summing over $k=n-i=0 \ldots n-1$ gives the announced result.
\end{proof}

%\subsection{Decoding via a shallow neural network}
%\label{sec_shallow}
Theorem~\ref{theo_nbReg_Lin} implies that the HLD, applied on $A_n$, induces a neural network (having the form given by \eqref{equ_feed_fo}) of exponential size. Indeed, remember that the first layer of the neural network implementing the HLD performs projections on the orthogonal vectors to each affine piece. 

Nevertheless, one can wonder whether a neural network with a different architecture can compute the decision boundary more efficiently. We first address another category of shallow neural networks: ReLU neural networks with two layers.
Deep neural networks shall be discussed later in the paper.
 Note that in this case we do not consider a single function computed by the neural network, like the HLD, but any function that can be computed by this class of neural network.

%We can even use this result on the number of pieces in the decision bundary function to compute a bound on the minimum number of neurons required 
%by any ReLU neural network with one hidden layer.

\begin{theorem}
\label{theo_shallow}
A ReLU neural network with two layers needs at least 
\begin{equation}
\label{eq_An_shallow}
%\sum_{i=0}^{n-1}(n-1-i) \cdot \binom{n-1}{i}
\sum_{i=2}^{n}(i-1) \times \binom{n-1}{n-i}
\end{equation}
neurons for optimal decoding of the lattice $A_n$.
\end{theorem}

The proof is provided in Appendix~\ref{App_theo_relu}.
Consequently, this class of neural networks is not efficient. However, we shall see in the sequel that deep neural networks are better suited. \\

\textbf{Other dense lattices} \\
Similar proof techniques can be used to compute the number of pieces obtain with some bases
of other dense lattices such as $D_n, n\geq 2$, and $E_n$, $6 \leq n \leq 8$. %We present results for two basis 

Consider the Gram matrix of $D_n$ given by \eqref{eq_second_kind}.
%As we shall see in the sequel, the formula obtained 
%with this basis as a lot of similarities with the one of $E_n$.
All basis vectors have the same length 
but we have either $\pi/3$ or $\pi/2$ angles between the basis vectors. 
This basis is not VR but SVR.
It is defined by the following Gram matrix.
\begin{equation}
\label{eq_second_kind}
\Gamma_{D_{n}}=
\left(
\begin{array}{cccccccc}
2 & 0 &1 &. . . &1 \\
0 &2& 1& . . . & 1 \\
1 &1& 2&  . . . & 1 \\
. & . & .& . .  . & . \\
1 & 1 & 1  & . . . & 2
\end{array}
\right).
\end{equation}

\begin{theorem}
\label{theo_nbReg_Lin_Dn_second_kind}
Consider a $D_n$-lattice basis defined by the Gram matrix~\eqref{eq_second_kind}.
Let $o_i$ denote the number of sets $\CN(x) \cap \CBZERO$, $x \in \CBONE$, where:
\begin{itemize}
\item \small $|\CN(x) \cap \CBZERO|= \underset{(l_i)}{\underbrace{\left[ 1+(n-2-i) \right]}}$\normalsize, and
\item \small $|\CN(x) \cap \CBZERO|= \underset{(ll_i)}{\underbrace{\left[\underset{(1)}{\underbrace{1+2(n-2-i)}}+ \underset{(2)}{\underbrace{\binom{n-2-i}{2}}}\right]}}$ \normalsize.
\end{itemize}
The decision boundary function $f$ has a number of affine pieces equal to
\small
\begin{align}
\label{eq_nbReg_Dn_second_kind}
% \sum_{i=2}^{n-1} \binom{n-2}{n-i} \times \Big(  (1+n-i)+ (1+2 (n-i)+ \binom{n-i}{2}) \Big)+1.
\begin{split}
\sum_{i=0}^{n-2} \left(  (l_i) + (ll_i) \right)\times o_i -1,
\end{split}
\end{align}
\normalsize
with $o_i=\binom{n-2}{i}$.
\end{theorem}

\begin{figure}
\centering

\begin{minipage}{.4\textwidth}
    \centering
    \includegraphics[scale=0.35]{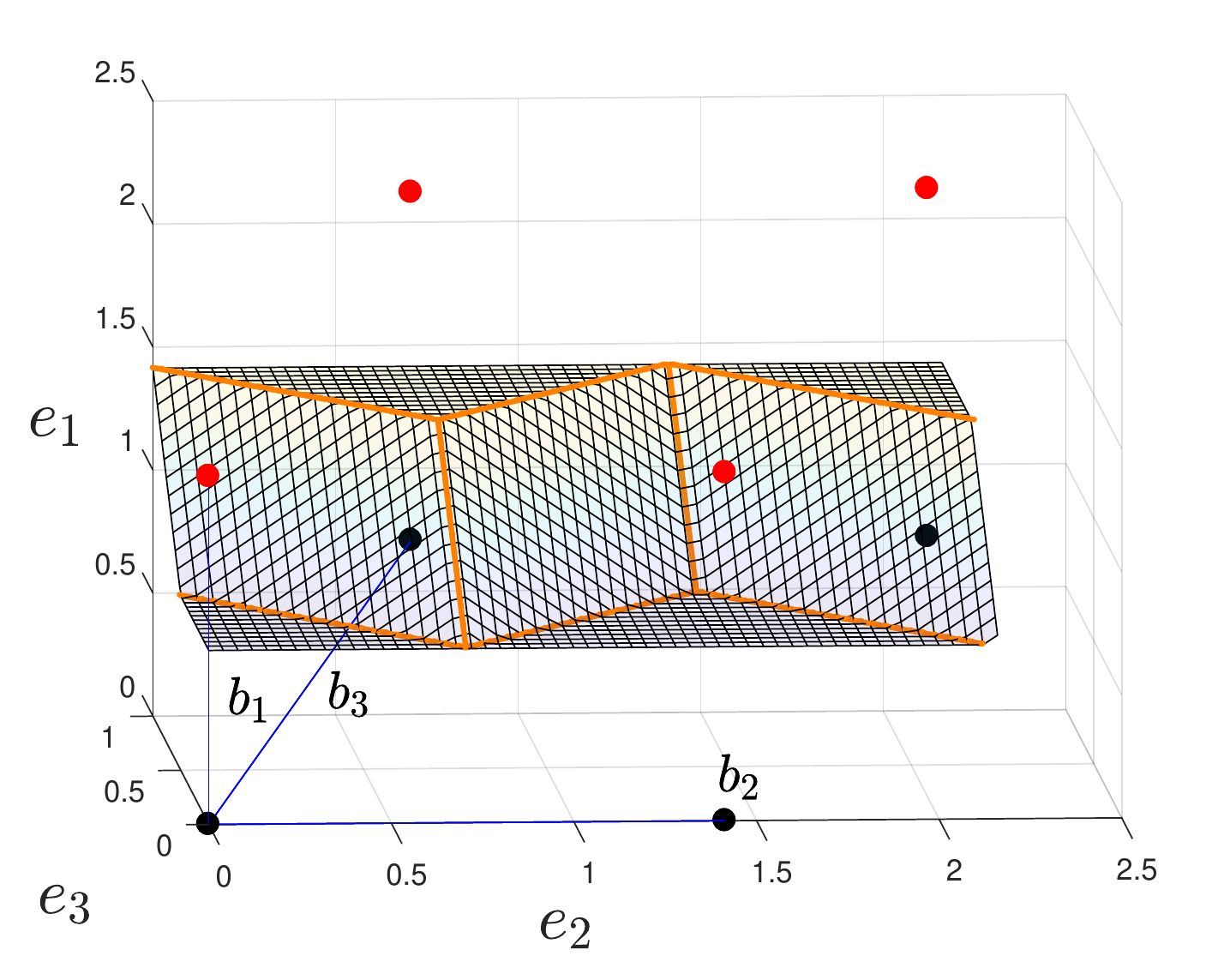}
     \caption{CPWL boundary function for $D_3$. 
	        The basis is rotated to better illustrate the symmetry: $g_1$ is colinear with $e_1$.}
     \label{fig_func_D3_second_kind}
\end{minipage}%
\hspace{8mm}
\begin{minipage}{.4\textwidth}
    \centering
    \includegraphics[scale=0.35]{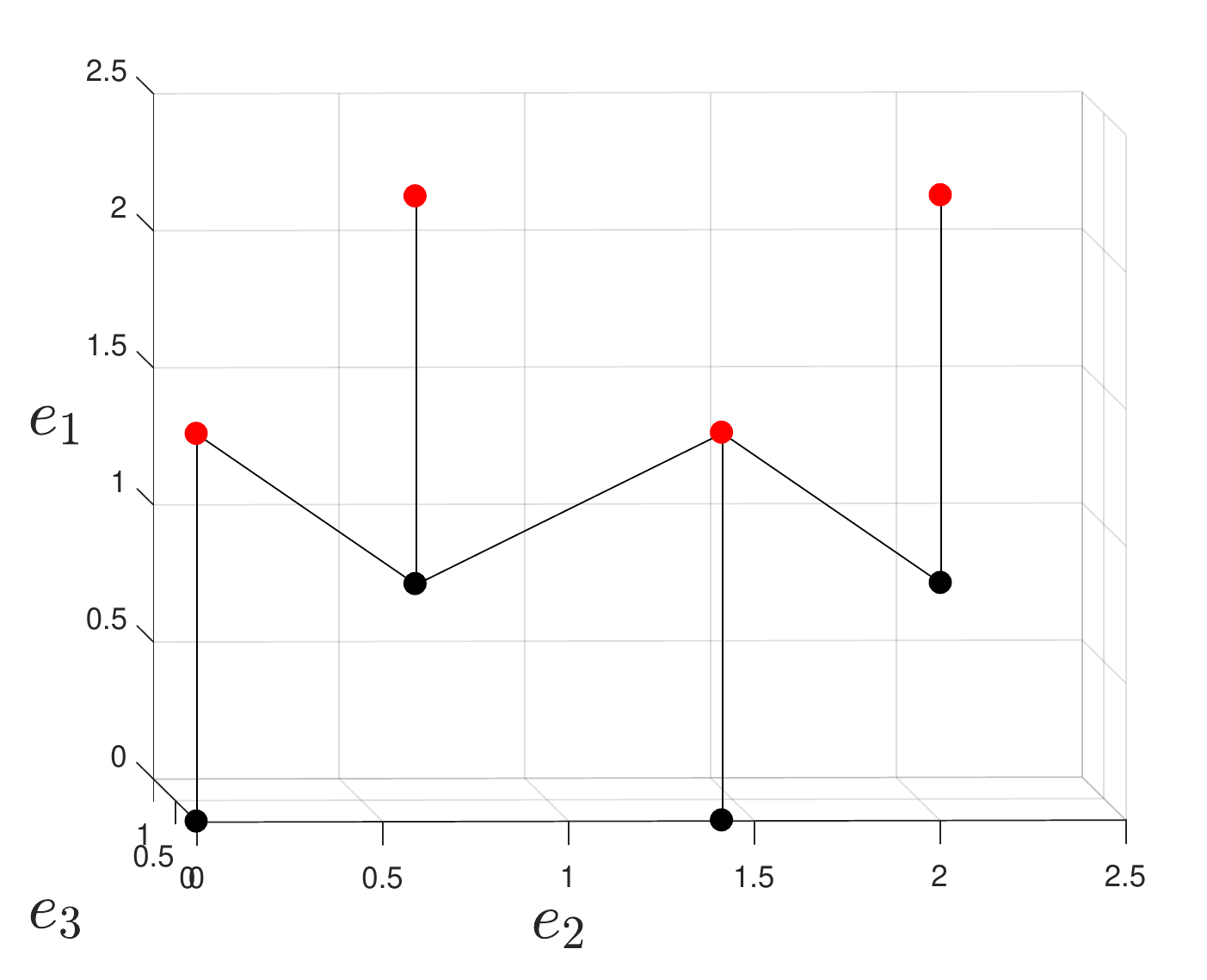}
     \caption{``Neighbor figure'' of $\mathcal{C}_{\PPB}$ for $D_3$. Each edge connects a point $x \in \CBONE$ to an element of $\CN(x) \cap \CBZERO$.}
     \label{fig_simplex_D3_2}
\end{minipage}
\end{figure}

%We give an example to gain insight into the above formula. The proof is deferred to Appendix~\ref{App_Dn_second_kind}.
%The previous sketch of proof highlights that we need to count the
%$i$-simplices to get the number of pieces of $f$.
%This is achieved by finding the different ``neighborhood patterns" (this gives $(l_i)$ and $(ll_i)$) and 
%counting the occurrence of $i$-simplices for each of these patterns (this gives $(o)$).
We presents the two different ``neighborhood patterns"  encountered with this basis of $D_n$ (this gives $(l_i)$ and $(ll_i)$).
In the proof available in Appendix~\ref{App_theo_Dn_nb_pieces}, we then count the number of simplices  (i.e. $(o_i)$) in each of these two categories.

%\begin{example}
%\label{ex_second_kind}
The decision boundary function for $D_3$ is illustrated on Figure \ref{fig_func_D3_second_kind}. 
We investigate the different ``neighborhood patterns" by studying Figure~\ref{fig_simplex_D3_2}: 
I.e. we are looking for the different ways to find the neighbors of $x \in \CBONE$ in $  \CN(x) \cap \CBZERO$, depending on $x$.
In the sequel, $(l_i)$, $(ll_i)$, and $(1)$, $(2)$ refer to Equation~\eqref{eq_nbReg_Dn_second_kind}
and $\sum_j g_j$ denotes any sum of points in the set $\{0,g_j\}_{j=3}^{n}$, where $g_2$ is the basis vector orthogonal to $g_1$. 
We recall that adding $g_1$ to any point $x \in \CBZERO$ leads to a point in $\CBONE$.

$(l_i)$ This pattern is the same as the (only) one encountered for $A_n$ with the basis given by Equation~$\eqref{eq_basis_An}$.
We first consider any point in $\CBONE$ of the form $\sum_j g_j + g_1$.
Its neighbors in $\CBZERO$ are $\sum_j g_j$ and any $\sum_j g_j+g_i$, 
where $g_i$ is any basis vector having an angle of $\pi/3$ with $g_1$ such that $\sum_j g_j+g_i$ is not outside $\PPB$.
Hence, $|\CN(\sum_{j=1}^i g_j + g_1) \cap \CBZERO|=1+n-2-i$.
E.g. for $n=3$, the closest neighbors of $0+g_1$  in $\CBZERO$ are $0$ and $g_3$. 
$g_2$ is perpendicular to $g_1$ and is not a closest neighbor of $g_1$. 

%We get a $1+n-2$-simplex generating 2 pieces $f$.
%We get the same pattern for $b_2+b_1$. %except that no $b_j, 1 \le j \le 2$ can be added.
%The point $b_3+b_1$ also belongs to this category except that no basis vectors having an angle of $\pi/3$ with $b_1$ can be added to $b_3$ without leaving $\PPB$
%(i.e $b_3+b_3+b_1$ is not in $\PPB$). 
%Hence, the only closest neighbor of $b_3+b_1$ in $\CBZERO$ is $b_3$: we have a $1+n-2-1$-simplex.
%Note that . 
%and $b_2+b_3 + b_1$ and $b_2+b_3+b_1$. 

$(ll_i)$ The second pattern is obtained with any point of the form $\sum_j g_j + g_2+g_1$ and its neighbors in $\CBZERO$.
$\sum_j g_j + g_2$ and any $\sum_j g_j + g_2 +g_i$, $\sum_j g_j+g_k$ are neighbors of this point in $\CBZERO$, where $g_i$, $g_k$ are any basis vectors having an angle of $\pi/3$ 
with $g_1$ such that (respectively) $\sum_j g_j + g_2 +g_i$, $\sum_j g_j+g_k$  are not outside $\PPB$.
This terms generate the $(1)$ in the formula.
E.g. for $n=3$, the closest neighbors of $0+g_2+g_1$  in $\CBZERO$ are $g_2$, $g_2 + g_3$, and
$g_3$. %We get a $1+2 (n-2)$~-~simplex.
%For $b_3+b_2+b_1$, it is the same pattern, except that in this case no $b_i$ can be added to $b_3+b_2$ without leaving $\PPB$: we have a $1+2 (n-2-1)$-simplex.
Moreover, for $n=3$ one ``neighborhood case" is not happening: 
from $n=4$, the points $g_i+g_j \in \CBZERO$, $3 \leq i < j \leq n $, are also closest neighbors of $g_2 + g_1$. 
This explains the binomial coefficient $(2)$.
Hence, $|\CN(\sum_{j=1}^i g_j + g_2 +  g_1) \cap \CBZERO|=1+2 (n-2-i) +\binom{n-2-i}{2}$.

Finally, we investigate $E_n$, $6 \le n \le 8$. 
$E_8$ is one of the most famous and remarkable lattices due to its exceptional density relatively to its dimension (it was recently proved that $E_{8}$ 
is the densest packing of congruent spheres in 8-dimensions \cite{Viazovska2017}). 
%Recently, it was proved that $\Lambda_{24}$ is the densest packing of congruent spheres in 24-dimensions \cite{Cohn2017}. 
The basis we consider is almost identical to the basis of $D_n$ given by \eqref{eq_second_kind}, except one main difference:
there are two basis vectors orthogonal to $g_1$  instead of one. 
This basis is not VR but SVR.
It is defined by the following Gram matrix.

\begin{equation}
\label{eq_En}
\Gamma_{E_{n}}=
\left(
\begin{array}{cccccccc}
2 & 0 &0 &1&. . . &1 \\
0 &2& 1&1& . . . & 1 \\
0 &1& 2&1&  . . . & 1 \\
1 &1&1&2& . . .  &1 \\ 
. & . & .& . & . .  . & . \\
1 & 1 & 1 &1 & . . . & 2
\end{array}
\right).
\end{equation}

\begin{theorem}
\label{theo_nbReg_Lin_En}
Consider an $E_n$-lattice basis, $6 \le n \le 8$, defined by the Gram matrix~\eqref{eq_second_kind}.
The decision boundary function $f$ has a number of affine pieces equal to
\small
\begin{align}
\label{eq_En}
\begin{split}
&\sum_{i=0}^{n-3} \Big( \underset{(l_i)}{\underbrace{\left[ 1+(n-3-i) \right]}} +\underset{(ll_i)}{\underbrace{2 \left[ 1+ 2(n-3-i)+ \binom{n-3-i}{2} \right]}} +  \\
&\underset{(lll_i)}{\underbrace{\left[ \underset{(1)}{\underbrace{1+3(n-3-i) }}+  \underset{(2)}{\underbrace{3\binom{n-3-i}{2}}} + \underset{(3)}{\underbrace{\binom{n-3-i}{3}}} \right]} }\Big)\underset{(o_i)}{\underbrace{\binom{n-3}{n-i}}}-3.
\end{split}
\end{align}
\normalsize
\end{theorem}

We first highlight the similarities with the function of $D_n$ defined by~\eqref{eq_second_kind}.  As with $D_n$, we have case $(l_i)$. Case $(ll_i)$ of $D_n$ is also present but obtained twice because of the two orthogonal vectors.  The terms $n-2-i$ in $(l_i)$ and $(ll_i)$ of Equation~\eqref{eq_nbReg_Dn_second_kind} are replaced by $n-3-i$ also because of the additional orthogonal vector.

Then, there is a new pattern $(lll_i)$: Any point of the form $\sum_j g_j + g_3+g_2+g_1$ and its neighbors in $\CBZERO$, where $\sum_j g_j$ represents any sum of points in the set $\{0,g_j\}_{j=4}^{n}$.
For instance, the closest neighbors in $\CBZERO$ of $g_3 + g_2+g_1 \in \CBONE$  are the following points, which we can sort in three groups as on Equation~\eqref{eq_En}: 
(1) $g_2+g_j$, $g_3+g_j$, $g_2+g_3+g_j$, (2) $g_j+g_k$, $g_2+g_j+g_k$, $g_3+g_j+g_k$, (3) $g_j+g_i+g_k$,
 $4 \leq i<j<k \leq n$. The formal proof is available in Appendix~\ref{App_func_En}.\\

%\subsubsection{The HLD is hardware-frienly}

  %Avantages plus ``haut niveau" par rapport aux autres algos: on peut tout paralléliser (GPU) et on a une faible latence.

%\begin{itemize}
%\item The decision boundary (NOT the function)
%\item Decoding via classification with the HLD (projection + Bool equ)
%\item The HLD is optimal with a VR-basis.
%\item The equation of the HLD is $\hat{z}=f(f(f(y\cdot G_1)\cdot G_2)\cdot G_3)$: definition of a feedforward neural network.
%\item The decision boundary can be thought as a piecewise affine function
%\item Complexity analysis 
%\begin{itemize}
%\item	Comptage des pi\`eces affines
%\item  Avantages plus ``haut niveau" par rapport aux autres algos: on peut tout paralléliser (GPU) et on a une faible latence.
%\end{itemize}
%\end{itemize}

\section{Complexity reduction}
\label{sec_complex_reduc}

In this section, we first show that a technique called the folding strategy enables to compute the decision boundary function at a reduced (polynomial) complexity. The folding strategy can be seen as a preprocessing step to simplify the function to compute.
The implementation of this technique involves a deep neural network. 
As a result, the exponential complexity of the HLD is reduced to a polynomial complexity by moving from a shallow neural network to a deep neural network.
The folding strategy and its implementation is first presented for the lattice $A_n$. We then show that folding is also possible for $D_n$ and $E_n$. 

In the second part of the section, we argue that, on the Gaussian channel, the problem to be solved by neural networks is easier for MIMO lattices than for dense lattices: In low to moderate dimensions, many pieces of the decision boundary function can be neglected for quasi-optimal decoding.  
%Consequently, the problem to be solved by neural networks is easier with MIMO lattices.  
Assuming that usual training techniques naturally neglect the useless pieces,
this explains why neural networks of reasonable size are more efficient with MIMO lattices than with dense lattices.

\subsection{Folding strategy}
\label{sec_fold}

\subsubsection{The algorithm}

Obviously, at a location $\tilde{y}$,
we do not want to compute all affine pieces in (\ref{eq_boundary_funct}),
whose number is for instance given by (\ref{eq_An_pieces}) for $A_n$.
To reduce the complexity of this evaluation,
the idea is to exploit the symmetries of $f$
by ``folding" the function and mapping distinct regions
of the input domain to the same location.
If folding is applied sequentially, i.e. fold a region that has already been folded, the gain becomes exponential. 
The notion of folding the input space in the context of neural networks
was introduced in \cite{Szymanski2013} and \cite{Montufar2014}.
We first present the folding procedure for the lattice $A_n$
and explain how this translate into a deep neural networks.
We then show that this strategy
can also be applied to the other dense lattices studied in Section~\ref{sec_complex_ana}. \\

\textbf{Folding of $A_n$} \\
The input space $\mathcal{D}(\B)$ is defined as in Section~\ref{sec_dec_bound}.
Given the basis orientation as in~Corollary~\ref{th_func_VR},
the projection of $g_j$ on $\mathcal{D}(\B)$ is $g_j$ itself, for $j\ge 2$. 
We also denote the bisector hyperplane
between two vectors $g_j, g_k$ by $BH(g_j, g_k)$
and its normal vector is taken to be $v_{j,k}=~g_j-g_k$. 
Let $\tilde{y} \in \mathcal{D}(\B)$ and let $\tilde{v}_{j,k}$ be a vector with the $n-1$ last coordinates of $v_{j,k}$. 
First,  we define the function $F_{j,k}$, where $2\le j < k \le~n$, which performs the following reflection.
Compute $\tilde{y} \cdot \tilde{v}_{j,k}$.
If the scalar product is non-positive, replace $\tilde{y}$
by its mirror image with respect to $BH(g_j, g_k)$.
Since $2\le j < k \le~n$, there are $\binom{n-1}{2}=(n-1)(n-2)/2$ functions $F_{j,k}$. 
The function $F_{A_n}$ performs sequentially these $O(n^2)$ reflections:
\begin{align}
F_{A_n}=F_{2,2} \text{ o } F_{2,3}  \text{ o } F_{3,3}  \text{ o }  ...  \text{ o } F_{n,n},
\end{align} 
and 
\begin{align}
F_{A_n}:\mathcal{D}(\B) \rightarrow \mathcal{D}(\B)'.
\end{align} 

\begin{theorem}
\label{theo_An_linear}
Let us consider the lattice $A_n$ defined by the Gram matrix~\eqref{eq_basis_An}. 
We have (i) $\D(\B)' \subset \D(\B)$, (ii) for all $\tilde{y} \in \D(\B)$, $f(\tilde{y}) = f(F_{A_n}(\tilde{y}))$ and (iii) $f$ has exactly 

\begin{equation}
\label{eq_nb_pieces_folding}
2n-1
\end {equation}
pieces on $\D(\B)'$.
\end{theorem}
Equation~\eqref{eq_nb_pieces_folding} is to be compared with \eqref{eq_An_pieces}.
\begin{example2}[Continued]
\label{ex_A3_2}
The function $f$ for $A_3$ restricted to $\D(\B)'$ (i.e. the function to evaluate after folding), say $f_{\D(\B)'} $, is
\begin{align}
\label{eq_A3}
f_{\D(\B)'} = \Big[ h_{p1} \vee h_{1} \Big] \wedge \Big[  h_{p2} \vee h_{2}  \Big] \wedge \Big[ h_{p3} \Big]. 
\end{align}
The general expression of $f^n_{\D(\B)'} $ for any dimension $n$ is 
\begin{align*}
f^n_{\D(\B)'} = & \Big[  h_{p_1} \vee h_{1}  \Big] \wedge \Big[ h_{p_2} \vee h_{2} \Big] \wedge ... \ \wedge    \Big[  h_{p_{n-1}} \vee h_{n-1} \Big] \wedge \Big[ h_{p_n} \Big].
\end{align*}
\end{example2}

\begin{proof}
To prove (i) we  use the fact that $BH(g_j,g_k)$, $2\le  j < k \le~n$,
is orthogonal to  $\D(\B)$, then the image of $\tilde{y}$  via the folding
$F$ is  in $\D(\B)$.   

(ii) is the  direct result of  the symmetries  in the
$A_n$  basis where  the  $n$ vectors have the same length and the angle between any two basis vectors is $\pi/3$.  
A reflection with respect $BH(g_j, g_k)$ switches $g_j$  and $g_k$ in
the hyperplane containing $\D(\B)$ and orthogonal to $e_1$. Switching $g_j$
and $g_k$ does  not change the decision boundary because  of the basis
symmetry, hence  $f$ is unchanged.

Now, for (iii), how  many pieces are left after all reflections?
Similarly to the proof of Theorem~\ref{theo_nbReg_Lin}, 
we walk in $\CBZERO$ and for a given point $x \in \CBZERO$ we count 
the  number of elements of $\CN(x+b_1) \cap \CBZERO$ (via Equation~\eqref{eq_prop}) 
that are on the proper side of all bisector hyperplanes.
%we use the property given by equation $(\ref{eq_prop}) $
Starting with $\CN(x+b_1) \cap \CBZERO$, only $0$ and $g_2$ are on the proper side: 
any other point $g_{j}$, $j \ge 3$, is on the  other  side  of  the  the bisector  hyperplanes $BH(g_2, g_j)$.
Hence, the lattice  point $g_1$, which had
$n$ neighbors in $\CBZERO$ before folding, only has  2 now. 
$f$ has only two pieces around $g_1$ instead of $n$.
Then, from  $g_{2}$ one can add $g_{3}$ but no other for
the same reason. 
The point $g_2+g_1$ has only 2  neighbors in  $\CBZERO$ on the proper side.
The pattern replicates until the last corner reaching
$g_1+g_2+\ldots+g_n$ which has only one neighbor.
So we get $2(n-1)+1$ pieces.
\end{proof}

\textbf{From folding to a deep ReLU neural network} \\
For sake of simplicity and without loss of generality, in addition to the standard ReLU activation function ReLU$(a)=\max(0, a)$, we also allow the function $\max(0,-a)$ and the identity as activation functions in the neural network.
%Note, that this function can easily be obtained via a standard ReLU network.

To implement a reflection $F_{j,k}$, one can use the following strategy. 

\begin{itemize}
\item Step~1: rotate the axes to have the $i$th axis $e_i$ perpendicular to the reflection hyperplane and shift the point (i.e. the $i$th coordinate) to have the reflection hyperplane at the origin.
\item Step~2: take the absolute value of the $i$th coordinate. 
\item Step~3: do the inverse operation of step 1.
\end{itemize}

Now consider the ReLU neural network\footnote{This neural network uses both ReLU and linear activation functions. It can still be considered as a ReLU neural network as a linear activation function can be implemented with ReLU neurons.} illustrated in Figure~\ref{fig_relu_ref}. 
The edges between the input layer and the hidden layer (the dashed square) represent the rotation matrix (Step 1), where the $i$th column is repeated twice, and $p$ is a bias applied on the $i$th coordinate.
Within the dashed square, the absolute value of the $i$th coordinate is computed
and shifted by $-p$. The activation functions in the dashed square, $\max(0,a)$, $\max(0,-a)$, and $a$, implement the absolute value operation (Step 2). 
Finally, the edges between the hidden layer and the output layer represent the inverse rotation matrix (Step 3). 
This ReLU neural network computes a reflection $F_{j,k}$. We call it a reflection block.
Note that the width of a reflection block is $O(n)$.

The function $F_{A_n}$ can be implemented by a simple
concatenation
of reflection blocks. 
This leads to a very deep 
and narrow neural network
of depth $\mathcal{O}(n^2)$ (the number of functions $F_{j,k}$) and width $\mathcal{O}(n)$ (the width of a reflection block is linear in $n$).
\begin{figure}
    \centering
    \includegraphics[scale=0.75]{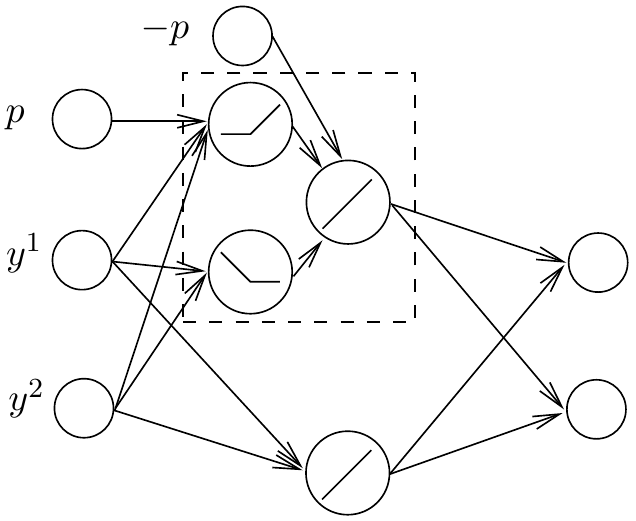}
     \caption{Reflection ReLU neural network (called reflection block). }
     \label{fig_relu_ref}
\end{figure}

Regarding the $2n-1$ remaining pieces after folding, we have two options 
(in both cases, the number of operations involved is negligible compared to the previous folding operations). 
To directly discriminate the point with respect to $f$, 
we implement the HLD on these remaining pieces with two additional hidden layers 
(as in Figure~\ref{fig_A2_net}): 
project $y_{folded}$ on the $2n-1$ hyperplanes (see Theorem~\ref{theo_An_linear}), with one layer of width $2n +1$, 
and compute the associated Boolean equation with an additional hidden layer.
If needed, we can evaluate $f(\tilde{y})$ via $\mathcal{O}(\log(n))$ additional hidden layers. 
First, compute the $n-1$ 2-$\vee$ via two layers of size $\mathcal{O}(n)$ containing several ``max ReLU neural networks" 
(see e.g. Figure 3 in \cite{Arora2018}). Then, compute the $n$-$\wedge$ via $\mathcal{O}(\log(n))$ layers. 
%Note that $f(\tilde{y})$ can also be used for discrimination via the sign of $y_{i}-f(\tilde{y})$. 

Consequently, $f$ can be computed by a ReLU network of depth $\mathcal{O}(n^2)$ and width $\mathcal{O}(n)$. \\ 
%The total number of parameters in the whole neural network is $\mathcal{O}(n^4)$.
%In \cite{Corlay2019}~-~\ref{App_number_params}, 
%we quickly discuss whether or not this can be improved.

%Eventually, the closest corner of $\PPB$ to $y$ is found (and thus the CVP is solved because the basis is VR) 
%by using $n$ such neural networks in parallel, one for each $z_i$, $1 \leq i \leq n$ (this could be optimized). 
%The final neural network has depth $\mathcal{O}(n^2)$ and width $\mathcal{O}(n^2)$. \\

\textbf{Folding of other dense lattices} \\
We now present the folding procedure for other lattices.
%The folding transformation $F$ is defined as in Theorem~\ref{theo_An_linear} for the following result.
%\begin{theorem}
%\label{theo_Dn_const_A_linear}
%Let us consider the lattice $D_n$ defined by the Gram matrix~\eqref{eq_Dn_const_A}. 
%We have (i) for all $\tilde{y} \in \D(\B)$, $f(\tilde{y}) = f(F(\tilde{y}))$ and (ii) $f$ has exactly 
%
%\begin{equation}
%2n-1
%\end {equation}
%pieces on $\D'(\B)$.
%This is to be compared with (\ref{eq_Dn_constA}). 
%\end{theorem}
%
%The folding procedure is identical to the one used for $A_n$: 
%The number of pieces to evaluate is reduced to a linear number via $\mathcal{O}(n^2)$ reflections with respect to the bisector hyperplanes
%between any pair of vectors in $\B \backslash \{ b_1 \}$. 
%The proof presents no novelty and is deferred to Appendix~\ref{App_folding_const_A}.
%
%As a result, $f$ can be computed by a ReLU network of depth $\mathcal{O}(n^2)$ and width $\mathcal{O}(n)$.

%%%%%%%%%

First, we consider $D_n$
defined by the Gram matrix~\eqref{eq_second_kind}.
$F_{D_n}$ is defined as $F_{A_n}$
except that we keep only the $F_{j,k}$ for $j,k\geq3$. 
Moreover, the $g_i$ are now the basis vectors of $D_n$ instead of $A_n$,
where $g_2$ is the basis vector orthogonal to $g_1$. 
There are $\binom{n-2}{2}=(n-2)(n-3)/2$ functions $F_{j,k}$
and the function $F_{D_n}$ performs sequentially the $O(n^2)$ reflections.
%Given the basis orientation as in Corollary~\ref{th_func_VR},
%the folding transformation
%$F:\D(\B) \rightarrow \D'(\B)$ is defined as follows.
%Let $\tilde{y} \in \D(\B)$, for all $3\le j < k \le~n$,
%compute $\tilde{y} \cdot v_{j,k}$ (the first coordinate of $v_{j,k}$ is zero).
%If the scalar product is non-positive, replace $\tilde{y}$
%by its mirror image with respect to $BH(g_j, g_k)$.
%There exist $\binom{n-2}{2}$ hyperplanes for mirroring.

\begin{theorem}
\label{theo_Dn_second_kind_lin}
Let us consider the lattice $D_n$ defined by the Gram matrix~\eqref{eq_second_kind}. 
We have (i) for all $\tilde{y} \in \D(\B)$, $f(\tilde{y})~=~f(F_{D_n}(\tilde{y}))$ and (ii) $f$ has exactly 

\begin{equation}
\label{equ_fold_Dn}
6n-12
\end {equation}
pieces on $\D'(\B)$.
\end{theorem}
Equation~\eqref{equ_fold_Dn} is to be compared with \eqref{eq_nbReg_Dn_second_kind}.\\ 
\textbf{Sketch of proof.} 
\ To count the number of pieces of $f$, defined on $\D'(\B)$, 
we need to enumerate the cases where both $x \in \CBONE$ and $x'\in \CN(x) \cap \CBZERO$ are 
on the non-negative side of all reflection hyperplanes. 
Among the points in $\mathcal{C}_{\PPB}$ only the points
\begin{enumerate}
\item $x_1=g_3+...+g_{i-1}+g_i$ and $x_1+g_1$,
\item $x_2=g_3+...+g_{i-1}+g_i+g_2$ and $x_2+g_1$,
\end{enumerate}
$i \leq n$, are on the non-negative side of all reflection hyperplanes.
It is then easily seen that the number of pieces of $f$, defined on $\D'(\B)$,
is given by equation~\eqref{eq_nbReg_Dn_second_kind} reduced as follows. 
The three terms $(n-2-i)$ (i.e. $2(n-2-i)$ counts for two), the term $\binom{n-2-i}{2}$, and the term $\binom{n-2}{i}$ become 1 at each step $i$,
for all $0 \leq i \leq n-3$ (except $\binom{n-2-i}{2}$ which is equal to 0 for $i=n-3$).
Hence, \eqref{eq_nbReg_Dn_second_kind} becomes $(n-3)\times(2+4)+(2+3)+1$, which gives the announced result. \\

Consequently, $f$ can be computed by a ReLU network of depth $\mathcal{O}(n^2)$ and width $\mathcal{O}(n)$ 
(i.e. the same size as the one for $A_n$).

%%%%%%%%
Second, we show how to fold the function for $E_n$.
$F_{E_n}$ is defined as $F_{A_n}$
except that, for the functions $F_{j,k}$, $4\le j < k \le~n$ and $j=2,k=3$ instead of 
$2\le j < k \le~n$,
where $g_2,g_3$ are the basis vectors orthogonal to $g_1$. 
There are $\binom{n-3}{2}+1=(n-3)(n-4)/2+1$ functions $F_{j,k}$
and the function $F_{E_n}$ performs sequentially the $O(n^2)$ reflections.
%Finally, we show how to fold the function for $E_n$.
%Given the basis orientation as in Corollary~\ref{th_func_VR},
%the folding transformation
%$F:\D(\B) \rightarrow \D'(\B)$ is defined as follows.
%Let $\tilde{y} \in \D(\B)$, for all $4\le j < k \le~n$ and $j=2,k=3$,
%compute $\tilde{y} \cdot v_{j,k}$ (the first coordinate of $v_{j,k}$ is zero).
%If the scalar product is non-positive, replace $\tilde{y}$
%by its mirror image with respect to $BH(g_j, g_k)$.
%There exist $\binom{n-3}{2}$+1 hyperplanes for mirroring.
%We get the following theorem, whose proof is available in Appendix~\ref{App_folding_En}.

%There exists $\binom{n-2}{2} + n-1$ hyperplanes for mirroring.
\begin{theorem}
\label{theo_En_folding}
Let us consider the lattice $E_n$, $6\leq n \leq 8$, defined by the Gram matrix~\eqref{theo_nbReg_Lin_Dn_second_kind}. 
%Also, $f_e$ is the extended boundary function (see Proposition~\ref{lemma_superexp}), defined on $\D'(\B)$.
We have (i) for all $\tilde{y} \in \D(\B)$, $f(\tilde{y})~=~f(F_{E_n}(\tilde{y}))$ and (ii) $f$ has exactly 

\begin{equation}
\label{equ_fold_En}
12n-40
\end {equation}
pieces on $\D'(\B)$.
\end{theorem}
Equation~\eqref{equ_fold_En} is to be compared with~\eqref{eq_En}. 
%\textbf{Sketch of proof} 
%\ Similarly to the sketch of proof of Theorem~\ref{theo_Dn_second_kind_lin}, the remaining cases where both 
%$x \in \CBONE$ and $x'\in \CN(x) \cap \CBZERO$ are on the ``non-negative side" of all reflection hyperplane is given by a 
%reduced version of equation~\eqref{eq_En}, namely: 
%$(n-4)\times(2+4+6)+(2+4+5)+(2+3+3)+1$, which gives the announced result. $\blacksquare$ \\
Consequently, $f$ can be computed by a ReLU network of depth $\mathcal{O}(n^2)$ and width $\mathcal{O}(n)$. \\

%\subsubsection{The advantage of depth over width}

%\subsubsection{Comparison with existing algo}

%To write \\
% But higher latency due to depth. 

%%%%%%%%%%%%%%%%%%%%%%%%

\subsection{Neglecting many affine pieces in the decision boundary } 

In the previous section, we showed that complexity reduction can be achieved for some structured
lattices by exploiting their symmetries.
What about unstructured lattices? We consider the problem of decoding on the Gaussian channel. The goal is to obtain quasi-MLD performance.

\subsubsection{Empirical observations}
In \cite{Corlay2018}, we performed several computer simulations with dense lattices (e.g. $E_8$) and MIMO lattices
(such as the ones considered in \cite{Samuel2017}), which are typically not dense in low to moderate dimensions. 
We aimed at minimizing the number of parameters in a standard fully-connected feed-forward sigmoid neural network \cite{Goodfellow2016} while maintaining quasi-MLD performance. The training was performed with usual gradient-descent-like techniques~\cite{Goodfellow2016}. The network considered is shallow, similar to the HLD, as it contains only three hidden layers. 
Let $W$ be the number of parameters in the neural networks (i.e. the number of edges). 
To be competitive, $W$ should be smaller than $2^n$. For $E_8$ we obtained a complexity ratio $\frac{\log_2 W}{n}=2.0$ whereas for the MIMO lattice the ratio is $\frac{\log_2 W}{n}=0.78$.
%As expected, the number of parameters required for the dense lattice increases exponenially with the dimension, as the HLD. Nevertheless, we managed to get a sub-exponential number of parameters for the less structured MIMO lattice.

We also compared the decoding complexity of MIMO lattices and dense lattices ($BW_{16}$ in this case) in \cite{Corlay2018_a}, with a different network architectures (but still having the form of a feed-forward neural network). The conclusion was the same: While it is possible to get a reasonable complexity for MIMO lattices, it is much more challenging for dense lattices.

\subsubsection{Explanation}
%We propose an explanation of this observation in this subsection. 
We explained in the first part of this paper that all pieces of the decision boundary function are facets of Voronoi regions. As a result, the (optimal) HLD needs to consider all Voronoi relevant vectors, which is equal to $\tau_f=2^{n+1}-2$ for random lattices. 
However, \eqref{equ_bound_prel} shows that a term in the union bound decreases exponentially with $\|x\|^2$, which is a standard behavior on the Gaussian channel.
Numerical evaluations of a union bound truncated at a squared distance of  $2 \cdot d^2(\Lambda)$ (3dB margin in VNR) yield 
very tight results at moderate and high VNR.  Therefore, only the first lattice shells need to be considered for quasi-MLD performance on the Gaussian channel. 

Consequently, we performed simulations to know how many Voronoi facets contribute to the 3dB-margin quasi-MLD error probability for random MIMO lattices generated by a matrix $G$ with random i.i.d $\mathcal{N}(0,1)$ components. We numerically generated 200000 random MIMO lattices $\Lambda$ and computed the average number of lattice points in a sphere of squared radius $2\cdot d^2(\Lambda)$ centered at the origin. 
%For moderate $\Delta$, only the lattice shells within $\approx 3$ dB of the first lattice shell (i.e. at a squared distance of $2\cdot d^2(\Lambda)$ from the origin)  need to be considered to get a good estimate of the MLD performance.
The results are reported in Table~\ref{table_nbPts}. Figure~\ref{fig_nbPts_approx} also provide the distribution for $n=14$.
The random lattices in dimension $n=14$ are generated by a matrix $G$ with random i.i.d. $\mathcal{N}(0,1)$ components. We numerically generated 200000 random lattices in dimension $n=14$ to estimate the probability distribution.
For comparison, the number of points in such a sphere is  25201 for the dense Coxeter-Todd lattice in dimension 12 and 588481 for the dense Barnes-Wall lattice in dimension 16 \cite[Chap.~4]{Conway1999}. Note however that while the numbers shown in Table~\ref{table_nbPts} are relatively low, the increase seems to be exponential: The number of lattice points in the sphere almost doubles when adding two dimensions.

\begin{table}
%\begin{align}
\begin{center}
\begin{tabular}{|c|c|c|c|c|c|}
 \hline
Dimension $n$  & 10 & 12 &  14 & 16   \\
  \hline
 %  &  &  &   & & \\
Average number of points & 59 & 109 & 201 & 361  \\
\hline
\end{tabular}
\end{center}%\end{align}
\caption{Average number of points in a sphere of squared radius $2\cdot d^2(\Lambda)$ centered at the origin for random MIMO lattices $\Lambda$ .}
\label{table_nbPts}
\end{table}
\begin{figure}[h]
    \centering
   	 \includegraphics[angle=270,scale=0.36]{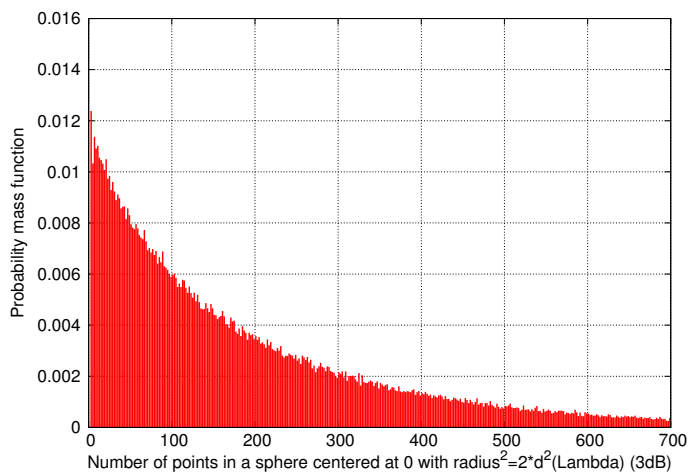}
     \caption{Distribution of the number of lattice points in a sphere of squared radius $2  d^2(\Lambda)$ for $n=14$. }
     \label{fig_nbPts_approx}
\end{figure}
%This means that the Voronoi regions of unstructured random MIMO lattices of low to moderate dimensions are not spherical.
This means that the number of Voronoi facets significantly contributing to the error probability is much smaller for random unstructured MIMO lattices compared to structured lattices in these dimensions. As a result, the number of hyperplanes
that should be taken into account for quasi-MLD is much smaller for random unstructured MIMO lattices. 
In other words, the function to compute for quasi-optimal decoding is ``simpler": A piecewise linear boundary with a relatively low amount of affine pieces can achieve quasi-MLD for random MIMO lattices.  

\subsection{Learning perspective}

We argue that regular learning techniques for shallow neural networks, such as gradient-descent, using Gaussian distributed data at moderate SNR for the training, naturally selects the Voronoi facets contributing to the error probability.
We estimated in the previous subsection, via computer search, that the number of Voronoi facets from this category is low for unstructured MIMO lattices. 
This explains why,
for quasi-optimal decoding in low to moderate dimensions, shallow neural networks can achieve satisfactory performance at reasonable complexity with unstructured MIMO lattices. 
However, the number of Voronoi facets to consider is much higher for structured lattices. This elucidates why it is much more challenging to train a shallow neural network with structured lattices.

%Alternatively, a reduced complexity decoding could also be performed via a modified HLD, where the brute-force search (Algorithm~\ref{alg_bool_eq})  should consider only the Voronoi facets contributing to the quasi-MLD error probability.

In the first part of this section, we explained that for this latter category of lattices, such as $A_n$, one should consider a deep neural network. It is thus legitimate to suppose that training a deep neural network to decode $A_n$ should be successful.
%Regarding the training aspect, learning via gradient-descent techniques allows to achieve satisfactory decoding performance  when shallow neural networks are considered along with a target function having a limited number of affine pieces.
However, when this category of neural networks is used, even when we know that their function class contains the target function, the training is much more challenging.
In particular, even learning simple one dimensional oscillatory function, such as the triangle wave function illustrated on Figure~\ref{fig_triangle}, is very difficult whereas they can be easily computed via folding. This can only be worst for high-dimensional oscillatory functions such as the boundary decision functions.

\begin{figure}[H]
\centering
\includegraphics[scale=0.835]{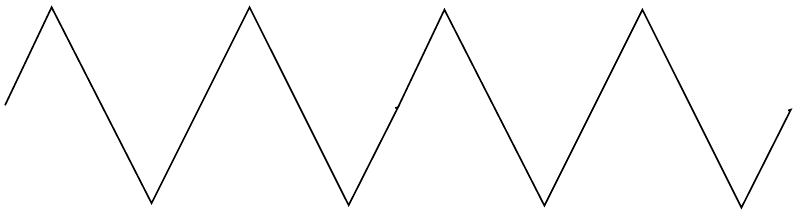}
\caption{Simple one-dimensional function which is challenging to learn via usual techniques.} 
\label{fig_triangle}
\end{figure}

This might explain the success of model-based techniques, where the neural network architectures are established by unfolding known decoding algorithms and where the weights are initialized based on these algorithms \cite{Nachmani2016}. Learning is then used to explore the functions in the function class of the neural network that are not ``too far'' from the initial point in the optimization space. Nevertheless, the initial point should already be of good quality to get satisfactory performance and learning amounts to fine tuning the algorithm.

\section{Conclusions}

The decoding problem has been investigated from a neural network perspective. 
We discussed what can and cannot be done with feed-forward neural networks in light of the complexity of the decoding problem.
We have highlighted that feed-forward neural networks should compute a CPWL boundary function to decode.
When the number of pieces in the boundary function is too high, the size of the shallow neural networks becomes prohibitive and
deeper neural networks should be considered. For dense structured lattices, this number of pieces is high even in moderate dimensions whereas it remains reasonable in low and moderate dimensions for unstructured random lattices.

%% file: 2_Appendix.tex
\section{Appendix}
%\vspace{1mm}

%%-----------------------------------------------------------------------
%%-----------------------------------------------------------------------
%\textcolor{blue}{
\subsection{Proof of Equation~\eqref{eq_pe_opt}}
\label{App_proof_small_o}
$P_e(ub)=\frac{1}{2}\sum_{x\in\Lambda\setminus\{0\}}
\exp\left(-\frac{\|x\|^2}{8\sigma^2} \right)=
\frac{1}{2}\sum_{x\in\Lambda\setminus\{0\}}
\exp\left(-\frac{\Delta}{8\sigma_{max}^2}\cdot \|x\|^2 \right)$,
where the signal-to-noise ratio, here called VNR,
is $\Delta=\sigma_{max}^2/\sigma^2$.
After grouping the lattice points shell by shell, with shell of index $k$ located
at distance $d_k$ from the origin, we obtain
\begin{equation}
\label{equ_f_Delta}
P_e(ub)=f(\Delta)=\sum_{k=1}^{\infty} \tau_k \exp(-\frac{\Delta}{8\sigma_{max}^2}\cdot d_k^2),
\end{equation}
where $\tau_1=\tau$ is the kissing number and $d_1=d(\Lambda)=2\rho(\Lambda)$ is the lattice minimum distance. It is well-known that the series $f(\Delta)$ converges for $\Delta>0$, because
the Theta series itself converges for $|q|<1$ and it is holomorphic in $z$ for $q=e^{i \pi z}$
and $\Im{z}\ge 0$ \cite[Chap.2, Sec.2.3]{Conway1999}. 
Another direct method is to upperbound $\tau_k$, for $k$ large, by the number
of points on a sphere in $\R^n$ of radius $d_k$ where each point is occupying
an area given by a sphere in $\R^{n-1}$ of radius $\rho$ to prove that $\tau_k$ is polynomial in $d_k$.
The sequence $d_k$ is unbounded and strictly increasing, hence $f(\Delta)$ converges
for $\Delta>0$. We will be just using the fact that $f(1)$ is finite to prove (15).
Indeed, we can write
\begin{align*}
\frac{\sum_{k=2}^{\infty} \tau_k \exp(-\frac{\Delta}{8\sigma_{max}^2}\cdot d_k^2)}
{\tau_1 \exp(-\frac{\Delta}{8\sigma_{max}^2}\cdot d_1^2)} 
&=\sum_{k=2}^{\infty} \frac{\tau_k}{\tau_1} \exp\left(-\frac{\Delta}{8\sigma_{max}^2}\cdot (d_k^2-d_1^2)\right)\\ \nonumber
&=\sum_{k=2}^{\infty} \frac{\tau_k}{\tau_1} \left[\exp\left(-\frac{d_k^2-d_1^2}{8\sigma_{max}^2} \right)\right]^1
\left[\exp\left(-\frac{d_k^2-d_1^2}{8\sigma_{max}^2} \right)\right]^{\Delta-1}\\ \nonumber
&\le f(1) \cdot \left[\exp\left(-\frac{d_2^2-d_1^2}{8\sigma_{max}^2} \right)\right]^{\Delta-1},
\end{align*}
where the latest right term vanishes for $\Delta \rightarrow \infty$.
This proves that $P_e(ub)=f(\Delta)=\tau_1 \exp(-\frac{\Delta}{8\sigma_{max}^2}\cdot d_1^2)
+ o\left(\exp(-\frac{\Delta}{8\sigma_{max}^2}\cdot d_1^2)\right)$
with the Bachmann-Landau small o notation.
This is (15) after replacing $-\frac{\Delta}{8\sigma_{max}^2}\cdot d_1^2$ by $-\frac{\pi e \Delta\gamma}{4}$. 
The interpretation of (15) is that the error-rate performance of a lattice on a Gaussian channel is dominated
by the nearest neighbors in the small-noise regime. 

\subsection{Proofs of Section~\ref{sec_proofs_VR}}
\label{App_proof_VR_bases}

\subsubsection{Proof of Theorem~\ref{theo_An}}
\label{App_VR_An}

We need to show that none of $y \in \mathcal{V}(x)$, $x \in \Lambda \backslash \mathcal{C}_{\PPB}$, crosses a facet of $\PPBbar$.
In this scope, we first find the closest point to a facet of $\PPBbar$ and show that its Voronoi region do not cross $\PPBbar$.
It is sufficient to prove the result for one facet of $\PPBbar$ as the landscape is the same for all of them.

Let $H_{\mathcal{F}_1}$ denote the hyperplane defined by $\mathcal{B} \backslash g_1$ where the facet $\mathcal{F}_{1}$ of $\PPBbar$ lies.
While $g_{1}$ is in $\PPBbar$ it is clear that $-g_1$ is not in $\PPBbar$. 
Adding to $-g_1$ any linear combination of the $n-1$ vectors generating $\mathcal{F}_{1}$ is equivalent to moving in a hyperplane, 
say $H_{P_1}$, parallel to $\mathcal{F}_{1}$ and it does not change the distance from $H_{\mathcal{F}_1}$. 
Additionally, any integer multiplication of $-g_{1}$ results in a point which is further from the hyperplane (except by $\pm 1$ of course). 
Note however that the orthogonal projection of $-g_1$ onto $H_{\mathcal{F}_1}$ is not in $\mathcal{F}_{1}$. The only lattice point in $H_{P_1}$ having this property is obtained by adding all $g_j$, $2 \leq j \leq n $, to $-g_1$, i.e. it is the point $-g_1 + \sum_{j=2}^n g_j$.

This closest point to $\PPBbar$, along with the points $\mathcal{B} \backslash g_1$, form a simplex. 
The centroid of this simplex is a hole of the lattice (but it is not a deep hole of $A_n$ for $n \geq 3$). 
It is located at a distance of  $\alpha/(n+1)$, $\alpha>0$,  to the center of any facet of the simplex and thus to $\mathcal{F}_1$ and $\PPBbar$.

\subsubsection{Proof of Theorem~\ref{theo_E8}}
\label{App_VR_En}

In this appendix, we prove Lemma~\ref{theo_E8_bis}.
One can check that any generator matrix $G$ obtained from the following Gram matrix generates $E_8$ and satisfies the assumption of Lemma~\ref{theo_E8_bis}. Consequently, it proves Theorem~\ref{theo_E8}.

\begin{equation}
\label{eq_E8_bis}
\Gamma_{E_8}=
\left(
\begin{array}{cccccccc}
  4  & 2  & 0  & 2  & 2  & 2  & 2  & 2 \\
  2  & 4  & 2  & 0  & 2  & 2  & 2  & 2 \\
  0  & 2  & 4  & 0  & 2  & 2  & 0  & 0 \\
  2  & 0  & 0  & 4  & 2  & 2  & 0  & 0 \\
  2  & 2  & 2  & 2  & 4  & 2  & 2  & 0 \\
  2  & 2  & 2  & 2  & 2  & 4  & 0  & 2 \\
  2  & 2  & 0  & 0  & 2  & 0  & 4  & 0 \\
  2  & 2  & 0  & 0  & 0  & 2  & 0  & 4
\end{array}  
\right).
\end{equation}

\begin{lemma}
\label{theo_E8_bis}
Let $G$ be a generator matrix of $E_{8}$, where the set of basis vectors $\mathcal{B}$ are all from the first lattice shell. 
Let $H_{\mathcal{F}_i}$ denote the hyperplane defined by $\mathcal{B} \backslash g_i$ where the facet $\mathcal{F}_{i}$ of $\PPBbar$ lies.
Let $\mathring{\mathcal{P}}(\mathcal{B})$ be the interior of the fundamental parallelotope of $E_8$. 
If $(G^{-1})^{T}$ is a generator matrix of $E_{8}$ with basis vectors from the first shell,
then the $G$ basis is Voronoi-reduced with respect to $\mathring{\mathcal{P}}$. 
\end{lemma}

To prove Lemma~\ref{theo_E8_bis}, we need the next lemma.

\begin{lemma}
\label{lem_vr}
Let $G$ be a generator matrix of a lattice $\Lambda$, where the rows of $G$ form a basis $\mathcal{B}$ of $\Lambda$ with lattice points from the first shell. 
Let $H_{\mathcal{F}_i}$ denote the hyperplane defined by $\mathcal{B} \backslash g_i$ where the facet $\mathcal{F}_{i}$ of $\PPBbar$ lies.
If $(G^{-1})^{T}$ generates $\Lambda^{*}$ with lattice points from the first shell of this dual lattice, then the minimum distance between any $H_{\mathcal{F}_i}$ and a lattice point in $\Lambda \backslash \PPBbar$ is
\begin{equation}
\label{equ_dist}
d(\Lambda \backslash \PPBbar), H_{\mathcal{F}_i})=\frac{d(\Lambda)}{\sqrt{\gamma(\Lambda^{*})} \times \sqrt{\gamma(\Lambda)}}.
\end{equation}
\end{lemma}

\begin{proof}
%We  derive the distance between the closest lattice point outside of $\PPBbar$ and a facet of $\PPBbar$ when both the generator matrices $G$ of a lattice $\Lambda$ and its inverse $(G^{-1})^{T}$ are composed of vectors from the first shell. 
%$n-1$ vectors of $\mathcal{B}$ define the (linear) hyperplane where a given facet $\mathcal{F}_{i}$ of $\PPBbar$ lies. %They are included in the facet. 
%Hence, the vector $u_{i}$, orthogonal to $\mathcal{F}_{i}$, is orthogonal to these $n-1$ vectors of $\mathcal{B}$. It follows that the scalar product $b\cdot u_{i}$ is non null only for one vector of $\mathcal{B}$, says $g_{i}$, and we know that it is equal to 1 since $G \cdot G^{-1}=I$. 
%Trivially, we have
%\begin{equation}
%\label{scal_prod}
%g_{i} \cdot u_{i} = 1.
%\end{equation}

We derive the minimum distance between a lattice point outside of $\PPBbar$, $x \in \Lambda \backslash \PPBbar$, and $H_{\mathcal{F}_i}$. This involves two steps: First, we find one of the closest lattice point by showing that any other lattice point is at the same distance or further and then we compute the distance between this point and $H_{\mathcal{F}_i}$.  In the following, $u_i$ is the basis vector of the dual lattice $\Lambda^*$ orthogonal to $\mathcal{F}_{i}$ and $g_i$ the only basis vector of $\Lambda$ where $u_i \cdot g_i \neq 0$, $g_i\in \mathcal{B}$.  

As explained in the proof for $A_n$, while $g_{i}$ is in $\PPBbar$ it is clear that $-g_{i}$ is not in $\PPBbar$. Adding any linear combination of the $n-1$ vectors generating the facet is equivalent to moving in a hyperplane parallel to $H_{\mathcal{F}_i}$. It does not change the distance from $H_{\mathcal{F}_i}$. 
Additionally, any integer multiplication of $-g_{i}$ results in a point which is further from the facet (except by $\pm$1 of course). Therefore, $-g_{i}$ is one of the closest lattice points in $\Lambda \backslash \PPBbar$ from $H_{\mathcal{F}_i}$. 
%On Figure \ref{fig_voro_reduced}, the point which has is Voronoi region drawn is an instance of a vector $-g_{i}$.\\

%\begin{figure}
%    \centering
%    \begin{picture}(160,160)
%	 \put(20,20){$g_2$}
%   	 \includegraphics[scale=0.26]{voroReduced_illu.eps}
%   \end{picture}
%     \caption{The three vecors in blue represent a basis of a lattice. }
%     \label{fig_voro_reduced}
%\end{figure}

How far is this point from $\PPBbar$? This distance is obtained by projecting $-g_{i}$ on $u_i$, the vector orthogonal to $\mathcal{F}_{i}$ % (i.e. the vector defined by one column of $G^{-1}$): 
\begin{equation}
\label{equ_dist_def}
d(\Lambda \backslash \PPBbar, H_{\mathcal{F}_i})=\frac{|g_{i} \cdot u_{i}|}{||u_{i})||}.
\end{equation} 
First, the term $g_{i} \cdot u_{i}=1$ since $G\cdot G^{-1}=I$.
Second, from the Hermite constant of the dual lattice $\Lambda^{*}$, and using $\text{det} \ G~\cdot~\text{det} \ G^{-1} =1 $, we get:
%\[
%\gamma_{\Lambda^{*}} = \left( \frac{d_{min_{\Lambda^{*}}}}{|\text{det} \ G ^{-1}|^{1/n}}\right)^{2}.
%\]
%Since $\text{det} \ G \cdot \text{det} \ G^{-1} =1 $, we get
\begin{equation}
\label{equ_dmin_inv}
d(\Lambda^{*})  = \frac{\sqrt{\gamma(\Lambda^{*})}}{|\text{det} \ G|^{1/n}}.
\end{equation}
Since all vectors of $\Lambda^*$ are from the first shell 
(i.e. their norm is  $d(\Lambda^{*})$, assumption of the lemma), \eqref{equ_dist_def} becomes
\begin{equation}
d(\Lambda \backslash \PPB, H_{\mathcal{F}_i}) = \frac{1}{d(\Lambda^*)} =\frac{|\text{det} \ G|^{1/n}}{\sqrt{\gamma(\Lambda^{*})}}.
\end{equation} 
The result follow by expressing det $G$ as a function of $\gamma(\Lambda)$ and $d(\Lambda)$.

\end{proof}
We are now ready to prove Lemma~\ref{theo_E8_bis}.
\begin{proof}[Proof (of Lemma~\ref{theo_E8})]
$g_i$, $u_i$, and $H_{\mathcal{F}_i}$ are defined as in the previous proof.
We apply \eqref{equ_dist} to $E_{8}$. 
 Since this lattice is self-dual, $\gamma(E_{8}^{*})=\gamma(E_{8}) = 2$ and \eqref{equ_dist} becomes
 \[
 d(E_{8} \backslash \PPBbar, H_{\mathcal{F}_i})=\frac{d(E_{8})}{2}=\rho(E_{8}),
 \]
% where $\rho$ is the packing radius. 
As a result, the closest lattice point outside of $\PPBbar$ is at a distance equal to the packing radius. %Remember that the basis is VR if and only if the Voronoi region of this point do not cross any facet of $\PPB$. 
Since the covering radius is larger than the packing radius, the basis is VR only if the Voronoi region of the closest points have a specific orientation relatively to the parallelotope. %(note that all Voronoi regions have the same orientation as they tile the space by translations, hence showing the orientation for one closest point is enough). 

The rest of the proof consists in showing that $H_{\mathcal{F}_i}$ is a reflection hyperplane for $-g_{i}$. Indeed, this would mean that there is a lattice point of $E_{8}$ on the other side of $H_{\mathcal{F}_i}$, located at a distance $\rho(E_{8})$ from $H_{\mathcal{F}_i}$. 
It follows that this lattice point is at a distance $d(E_{8})$ from $-g_{i}$ and is one of its closest neighbor. Hence, one of the facet of its Voronoi region lies in the hyperplane perpendicular to the vector joining the points, %(i.e. the vector orthogonal to the reflexion hyperplane)
at a distance $\rho(E_{8})$ from the two lattice points. Consequently, this facet and $H_{\mathcal{F}_i}$ lie in the same hyperplane. Finally, the fact that a Voronoi region is a convex set implies that the basis is VR. 

To finish the proof, we show that $H_{\mathcal{F}_i}$ is indeed a reflection hyperplane for $-g_{i}$. The reflection of a point with respect to the hyperplane perpendicular to $u_{i}$ (i.e. $H_{\mathcal{F}_i}$) is expressed as
 \[
 s_{u_{i}}(-g_{i}) =- g_{i}+ 2 \cdot \frac{ u_{i}\cdot g_{i}}{||u_{i}||^{2}}\cdot u_{i}.
 \]
 We have to show that this point belongs to $E_{8}$. The dual of the dual of a lattice is the original lattice. Hence, if the scalar product between $s_{u_{i}}(-g_{i})$ and all the vectors of the basis of $E_{8}^{*}$ is an integer, it means that this point belongs to $E_{8}$.
 \[
 s_{u_{i}}(-g_{i}) \cdot u_{j} = -g_{i} \cdot u_{j} + 2 \cdot \frac{ u_{i}\cdot g_{i}}{||u_{i}||^{2}}\cdot u_{i} \cdot u_{j}.
 \]
We analyse the terms of this equation: $g_{i} \cdot u_{j} \in \mathbb{Z}$ since they belong to dual lattices. We already know that $u_{i}\cdot g_{i}=1$. Also $u_{i} \cdot u_{j} \in \mathbb{Z}$ as $E_{8}^{*}$ is an integral lattice.  With Equation~\eqref{equ_dmin_inv}, we get that $\frac{2}{||u_{i}||^{2}}=1$.  We conclude that $s_{u_{i}}(-g_{i}) \cdot u_{j} \in \mathbb{Z}$.
 \end{proof}
 
 \subsubsection{Proof of Theorem~\ref{no_VR_24}}

 \begin{proof}
 $\Lambda_{24}$ is self-dual with $\gamma(\Lambda_{24})=2$ and $d(\Lambda_{24})=2$. 
 Asumme that we have two generator matrices $G$ and $G^{-1}$ satisfying the assumption of Lemma~\ref{lem_vr}. Equation~\eqref{equ_dist} gives
\begin{align}
\label{equ_dist_leech}
d(\Lambda_{24} \backslash \PPB, H_{\mathcal{F}_i})=\frac{d(\Lambda_{24})}{4}=\frac{\rho(\Lambda_{24})}{2}.
\end{align}
This distance is clearly smaller than the packing radius of $\Lambda_{24}$.

Moreover, Equation~\eqref{equ_dist_def} shows that if $G^{-1}$ contains a point which is not from the first shell, $\underset{i}{\min} d(\Lambda \backslash \PPBbar, H_{\mathcal{F}_i})$ becomes smaller has $\underset{i}{\text{max}} \ ||u_{i}||$ is greater. Hence, \eqref{equ_dist_leech} is an upper bound on $d(\Lambda_{24} \backslash \PPB, H_{\mathcal{F}_i})$.
\end{proof}

\subsection{Proof of Theorem~\ref{th_func_VR_2}}
\label{App_proof_func}

%Recall that any Voronoi cell is a convex polyhedron. 
%Any subset of intersecting hyperplanes bounding the Voronoi cell is thus a polytope. 
All Voronoi facets of $f$ associated  to a same point of $\CBONE$ form a polytope.
The variables within a AND condition of the HLD discriminate a point with respect to the boundary hyperplanes where these facets lie: 
The condition is true if the point is on the proper side of all these facets.
For a given point $y \in \PPB$, we write a AND condition $m$ as Heav($yA_m + q_m) > 0$, where $A_m \in \mathbb{R}^{ n \times l_m }$, $ q_m \in \R^{l_m}$.
Does this convex polyhedron lead to a convex CPWL function?

Consider Equation~$\eqref{eq_bool}$. 
The direction of any $v_j$ is chosen so that the Boolean variable is true for the point in $\CBONE$ whose Voronoi facet is in the corresponding boundary hyperplane.
Obviously, there is a boundary hyperplane, which we name $\psi$, between the lattice point $0 \in \CBZERO$ and $g_1 \in \CBONE$. 
This is also true for any $x \in \CBZERO$ and $x + g_1 \in \CBONE$.
Now, assume that one of the vector $v_j$ has its first coordinate $v^1_j$ negative.
It implies that for a given location $\tilde{y}$, if one increases $y_1$ the term $y \cdot v_{j}^{T} - \bias_{j}$ decreases and eventually becomes negative if it was positive.
Note that the Voronoi facet corresponding to this $v_j$  is necessarily above $\psi$, with respect to the first axis $e_1$, as the Voronoi region is convex.
It means that there exists $\tilde{y}$ where one can do as follows. 
For a given $y_1$ small enough, $y$ is in the decoding region $z_1=0$.
If one increases this value, $y$ will cross $\psi$ and be in the decoding region  $z_1=1$. 
If one keeps increasing the value of $y_1$, $y$ eventually crosses the second hyperplane and is back in the region $z_1=0$.
In this case $f$ has two different values at the location $\tilde{y}$ and it is not a function.
If no $v^1_j$ is negative, this situation is not possible.
All $v^1_j$ are positive if and only if all  $x \in \CBONE$ have their first coordinates $x_1$ larger than the first coordinates of all $\CN(x) \cap \CBZERO$.
Hence, the convex polytope leads to a function if and only if this condition is respected.
If this is the case, we can write Heav($yA_m + q)> 0 \Leftrightarrow \wedge_{k=1}^{l_m} y \cdot a_{m,k} + q_{m,k} > 0$, $a_{m,k},q_{m,k} \in \{v_j, \bias_j\}$. 
We want $y_{1}>h_{m,k}(\tilde{y})$, for all $1\leq k \leq l_m$, which is achieved if $y_{1}$ is greater than the maximum of all values. 
The maximum value at a location $\tilde{y}$ is the active piece in this convex region and we get $y_{1} = \vee_{k=1}^{l_m} h_{m,k}(\tilde{y})$. 

A Voronoi facet of a neighboring Voronoi region is concave with the facets of the other Voronoi region it intersects.
%We say that two convex polyhedra, generated by Voronoi facets of distinct cell and intersecting, are concave. 
The region of $f$ formed by Voronoi facets belonging to distinct points in $\CBONE$ form concave regions that are linked by a OR condition in the HLD. The condition is true if $y$ is in the Voronoi region of at least one point of $\CBONE$: $\vee_{m=1}^{M} \{ \wedge_{k=1}^{l_m} y \cdot a_{m,k} + q_{m,k} \}> 0$. We get $f(\tilde{y}) = \wedge_{m=1}^{M}\{\vee_{k=1}^{l_m}h_{m,k}(\tilde{y}) \}$.

Finally, $l_m$ is strictly inferior to~$\tau_{f}$ because all Voronoi facets lying in the affine function of a convex part of $f$ are facets of the same corner point. 
Regarding the bound on $M$, the number of logical OR term is upper bounded by half of the number of corner of $\PPB$ which is equal to $2^{n-1}$. 
%Finally, we have $K < \tau_f$ thanks to the puzzle lemma (Lemma~\ref{lem_number_parti_p}) as the basis is VR.

\subsection{First order terms of the decision boundary function before folding for $A_n$}
\label{terms_An}
%\subsubsection{Before folding}
\label{terms_An_before}
The equations of the boundary function for $A_n$ are the following.

\begin{align*}
f^{n=2} = \Big[   h_{p1} \vee h_{1}  \Big] \wedge \Big[ h_{p2} \Big].
\end{align*}
\begin{align*}
f^{n=3}= & \Big[   h_{p1} \vee h_{1} \vee h_{2}  \Big] \ \wedge  \Big[ \left( h_{p2} \vee h_{1}\right)  \wedge \left( h_{p2}\vee h_{2} \right)  \Big] \ \wedge \Big[   h_{p3} \Big]. 
\end{align*}
\begin{align*}
f^{n=4} = & \Big[ h_{p1} \vee h_{1} \vee h_{2} \vee h_{3} \Big] \wedge 
             \Big[ \left( h_{p2}  \vee h_{1} \vee h_{2}\right) \wedge \left( h_{p2} \vee h_{2} \vee h_{3} \right)  \wedge 
             \left( h_{p2} \vee h_{1} \vee h_{3} \right)\Big] \wedge \\
            & \Big[ \left( h_{p3} \vee h_{1} \right) \wedge\left( h_{p3} \vee h_{2}\right) \wedge \left( h_{p3} \vee h_{3}\right)\Big] \wedge 
              \Big[ h_{p4}\Big].
\end{align*}

\subsection{Proof of Theorem~\ref{theo_shallow}}
\label{App_theo_relu}

A  ReLU neural network with $n$ inputs and $W_1$ neurons in the hidden layer can compute a CPWL function with at most $\sum_{i=0}^{n}\binom{W_1}{i}$ pieces \cite{Pascanu2013}. 
This is easily understood by noticing that the non-differentiable part of $\max(0,a)$ is a $n-2$-dimensional hyperplane that separates two linear regions. 
If one sums $W_1$ functions $\max(0,d_{i} \cdot y)$, where $d_i$, $1\leq i \leq w_1$, is a random vector, one gets $W_1$ of such $n-2$-hyperplanes.  
The result is obtained by counting the number of linear regions that can be generated by these $W_1$ hyperplanes. 
%The number provided by the previous formula is attained if and only if the hyperplanes are in general position. 
%Clearly, in our situation the $n-2$-hyperplanes partitioning $\D$ are not in general position: the hyperplane arrangement is not $simple$.

The proof of the theorem consists in finding a lower bound on the number of such $n-2$-hyperplanes 
(or more accurately the $n-2$-faces located in $n-2$-hyperplanes) partitioning $\D(\B)$. This number is a lower-bound on the number of linear regions.
Note that these $n-2$-faces are the projections in $\D(\B)$ of the 
$n-2$-dimensional intersections of the affine pieces of $f$.

We show that many intersections between two affine pieces linked by a $\vee$ operator (i.e. an intersection of affine pieces within a convex region of $f$) are located in distinct $n-2$-hyperplanes.
To prove it, consider all sets $\CN(x) \cap \CBZERO$ of the form $\{x, x + g_1, x+ g_j\}$, $x \in \CBZERO$, $x+ g_j \in \CBZERO$.
The part of decision boundary function $f$ generated by any of these sets has 2 pieces and their intersection is a $n-2$-hyperplane.   
Consider the set $\{0,g_1, g_2\}$.  Any other set is obtained by a composition of reflections and translations from this set. 
For two $n-2$-hyperplanes associated to different sets to be the same, the second set should be obtained from the first one by a translation along a vector orthogonal to the $2$-face defined by the points of this first set. %(i.e. a vector parallel to the $n-2$-hyperplane). 
However, the allowed translations are only in the direction of a basis vector. None of them is orthogonal to one of of these sets.  
 
Finally, note that any set $\{x \cup (\CN(x) \cap \CBZERO) \}$ where $|\CN(x) \cap \CBZERO|=i$, encountered in the proof of 
Theorem~\ref{theo_nbReg_Lin}, can be decomposed into $i-1$ of such sets (i.e. of the form $\{x,x-g_1,x-g_1+g_j \}$). 
 Hence, from the proof of Theorem~\ref{theo_nbReg_Lin}, 
we get that the number of this category of sets, and thus a lower bound on the number of $n-2$-hyperplanes, 
is $\sum_{k=0}^{n-1}(n-1-k)~\binom{n-1}{k}$. Summing over $k=n-i=0 \ldots n-1$ gives
the announced result.

\subsection{Proof of Theorem~\ref{theo_nbReg_Lin_Dn_second_kind}}
\label{App_theo_Dn_nb_pieces}
%\begin{gather}
%\begin{split}
%&\forall x \in~\CBZERO, \ x' \in D_n \backslash \{b_j,0\} \backslash \{b_j + b_j \},\ 2 \leq j < i \leq n: \\
%&\{x+ b_j\} \backslash {x} \in \CN(x+b_1),  \ \{x+ b_j+b_i\} \backslash {x} \in \CN(x+b_1),  \\ 
%&x+ x' \not\in \CN(x+b_1) \cap \CBZERO.XXX
%\end{split}
%\end{gather}
%\begin{proof}
%Similarly to the proof of Theorem~\ref{theo_nbReg_Lin_Dn_const_A}, 
We count the number of sets $\CN(x)~\cap~\CBZERO$ with  cardinality  $i$.
We walk in $\CBZERO$ and for each of the $2^{n-1}$ points $x \in \CBZERO$ 
we investigate the cardinality of the set $\CN(x+~g_1)~\cap~\CBZERO$. 
In this scope, the points in $\CBZERO$ can be sorted into two categories: $(l_i)$ and $(ll_i)$. 
In the sequel, $\sum_j g_j$ denotes any sum of points in the set $\{0,g_j\}_{j=3}^{n}$.
These two categories and their properties (see also the explanations below Theorem~\ref{theo_nbReg_Lin_Dn_second_kind}), are: 
%via the property given by (\ref{eq_prop}). 
\small
\begin{gather}
\label{eq_prop_Dn_2}
\begin{split}
(l_i) \ &\forall \ x=\sum_j g_j \in~\CBZERO, \ x' \in D_n \backslash \{g_k,0\},\ 3 \leq k\leq n: \\
&x+ g_k \in \CN(x+g_1), \ x+ x' \not\in \CN(x+g_1) \cap \CBZERO. \\
\end{split}
\end{gather}
\begin{gather}
\label{eq_prop_Dn_2__2}
\begin{split}
(ll_i) \ &\forall \ x= \sum_j g_j + g_2\in~\CBZERO, \\ 
&x' \in D_n \backslash \{g_i,-g_2 + g_i,-g_2 + g_i +g_k,0\},\ 3 \leq i < k\leq n: \\
&(1)  \ (a) \ x+ g_i \in \CN(x+g_1), \ (b) \ x -g_2 + g_i \in \CN(x+g_1), \\
&(2) \ x -g_2 + g_i +g_k  \in \CN(x+g_1), \\
&(3) \  x+ x' \not\in \CN(x+g_1) \cap \CBZERO.
\end{split}
\end{gather}
\normalsize
We count the number of sets $\CN(x)~\cap~\CBZERO$ with  cardinality  $i$ per category.

$(l_i)$ is like $A_n$. 
Starting from the lattice point 0, the set $\CN(0+g_1) \cap \CBZERO$ is composed of $0$  and the $n-2$ other basis vectors (i.e. without $g_2$ because it is perpendicular to $g_1$). 
Then, for all $g_{j_1}$, $3 \leq j_1\leq n$,  the sets $\CN(g_{j_1}+g_1) \cap \CBZERO$
are obtained by adding any of the $n-3$ remaining basis vectors to $g_{j_1}$  (i.e. not $g_1$, $g_2$, or $g_{j_1}$). % to generate a simplex in $\PPB$ where the top corner is $g_{j_1}+g_1$. 
Indeed, if we add again~$g_{j_1}$, the resulting point is outside $\PPB$ and should not be considered. 
Hence, the cardinality of these sets is $n-2$ and there are $\binom{n-2}{1}$
ways to choose $g_{j_1}$: any basis vectors except~$g_1$ and $g_2$. 
Similarly, for $g_{j_1}+g_{j_2}$, $j_1 \neq j_2$, the cardinality of the sets  $\CN(g_{j_1}+g_{j_2} +g_1) \cap \CBZERO$ is  $n-3$ and there are $\binom{n-2}{2}$ ways to choose $g_{j_1}+g_{j_2}$. 
More generally, there are $\binom{n-2}{i}$ sets $\CN(x) \cap \CBZERO$ of cardinality $n-1-i$.

%From the origin, one can form a $n-1$-simplex with all basis vectors except $g_2$. Then, from any $g_j$, $j \neq 1,2$, one can form a $n-2$ simplex and there are $\binom{n-2}{1}$ ways to choose $g_j$. In general, there are $\binom{n-2}{i}$ ways to do this $n-1-i$ simplex.
$(ll_i)$ To begin with, we are looking for the neighbors of $g_2+g_1$. First (i.e. property $(1)$), we have the following $1+2 \times (n-2)$  points in $\CN(g_2+g_1) \cap \CBZERO$:
$g_2$, any $g_j+g_2$, $3 \leq j \leq n$, and any $g_j$, $3 \leq j \leq n$.  
Second (i.e. property $(2)$), the $\binom{n-2}{2}$ points $g_j + g_k$, $3 \le j<k\le n$, 
are also neighbors of $g_2+g_1$. Hence, $g_2+g_1$ has $1+2 \times (n-2)+\binom{n-2}{2}$ neighbors in $\CBZERO$.
Then, the points $g_1+g_2+g_{j_1} $, $3 \leq j_1 \leq n$, have $1+2 \times (n-2-1) + \binom{n-2-1}{2}$ neighbors of this kind, 
using the same arguments, and there are $\binom{n-2}{1}$ ways to chose $g_{j_1}$. 
In general, there are $\binom{n-2}{i}$ sets of cardinality $1+2 \times (n-2-i) + \binom{n-2-i}{2}$.
%One can replicate the pattern $\sum_k \binom{n-2}{k}$ for the same reason as (i), where at each step $k$, $1+2 \times (n-2-k) + \binom{n-2-k}{2}$-simplices are obtained.

To summarize, each set replicates $\sum_i \binom{n-2}{i}$ times, where for each $i$ we have both $(l_i)$ sets of cardinality $1+(n-2-i)$ and $(ll_i)$ sets of cardinality $1+2 \times (n-2-i) + \binom{n-2-i}{2}$.
As a result, the total number of pieces of $f$ is obtained as
%\begin{equation*}
% \sum_{k=0}^{n-2} \binom{n-2}{k} \times \left(  (n-1-k)+ \left(1+2(n-2-k)+ \binom{n-2-k}{2} \right) \right).
%\end{equation*}
\small
\begin{equation}
\label{eq_proof_big}
% \sum_{i=2}^{n-1} \binom{n-2}{n-i} \times \Big(  (1+n-i)+ (1+2 (n-i)+ \binom{n-i}{2}) \Big)+1.
\sum_{i=0}^{n-2} \left(   \underset{(l_i)}{\underbrace{\left[ 1+(n-2-i) \right]}}+  \underset{(ll_i)}{\underbrace{\left[\underset{(1)}{\underbrace{1+2(n-2-i)}}+ \underset{(2)}{\underbrace{\binom{n-2-i}{2}}} \right]}} \right) 
\times \underset{(o_i)}{\underbrace{\ \binom{n-2}{i}}}-1,
\end{equation}
\normalsize
where the -1 comes from the fact that for $i=n-2$, the piece generated by $(l_i)$ and the piece generated by $(ll_i)$ are the same. Indeed, the bisector hyperplane of $x$, $x+g_1$ and the bisector hyperplane of $x+g_2$, $x+g_2+g_1$ are the same since $g_2$ and $g_1$ are perpendicular.
%\end{proof}

\subsection{Proof of Theorem~\ref{theo_Dn_second_kind_lin} 
%Folding of $f$ with the second basis of $D_n$
}
\label{App_folding_second_kind}

\begin{lemma}
\label{lem_fold_2}
Among the elements of $\mathcal{C}_{\PPB}$, only the points
of the form 
\begin{enumerate}
\item $x_1=g_3+...+g_{i-1}+g_i$ and $x_1 +g_1$, 
\item $x_2=g_3+...+g_{i-1}+g_i+g_2$ and $x_2 + g_1$,
\end{enumerate}
$i \leq n$, 
are on the non-negative side of all $BH(g_j,g_k)$, $3\le  j < k \le~n$.  
\end{lemma}

%\begin{lemma}
%\label{lem_fold}
%Among the elements of $\mathcal{C}_{\PPB}$, only the points
%of the form $x=g_2+g_3+...+g_{i-1}+g_i$ and  $x+g_1$, $i \leq n$, 
%are on the non-negative side of all $BH(g_j,g_k)$, $2\le  j < k \le~n$.  
%\end{lemma}

\begin{proof}
In the sequel, $\sum_i g_i$ denotes any sum of points in the set $\{0,g_i\}_{i=3}^{n}$.
For 1), consider a point of the form $g_3+...+g_{j-1}+g_{j+1}+...+g_{i-1}+g_i$, $j+1<i-1\le n-1$.
This point is on the negative side of all $BH(g_j,g_k)$, $ j  < k \leq i$.
More generally, any point $\sum_i g_i$, where $\sum_i g_i$ includes in the sum $g_k$ but not $g_j$, $j<k \leq n$,
is on the negative side of $BH(g_j,g_k)$.
Hence, the only points in $\CBZERO$ that are on the non-negative side of all hyperplanes have the form $g_3+...+g_{i-1}+g_i$, $i \leq n$.

Moreover, if $x \in \CBZERO$ is on the negative side of one of the hyperplanes $BH(g_j,g_k)$, $3\le  j < k \le~n$, so is $x+g_1$
since $g_1$ is in all $BH(g_j,g_k)$.

2) is proved with the same arguments.
\end{proof}

\begin{proof} (of Theorem~\ref{theo_Dn_second_kind_lin})
(i) The folding  via $BH(g_j, g_k)$, $3 \leq j < k  \leq n$, switches $g_j$  and $g_k$ in
the hyperplane containing $\D(\B)$, which is orthogonal to $e_1$. Switching $g_j$
and $g_k$ does  not change the decision boundary because of the basis
symmetry, hence  $f$ is unchanged. 

Now, for (ii), how many pieces are left after all reflections?
To count the number of pieces of $f$, defined on $\D'(\B)$, 
we need to enumerate the cases where both $x \in \CBONE$ and $x'\in \CN(x) \cap \CBZERO$ are 
on the non-negative side of all reflection hyperplanes. 
%The result is easily obtained via the effect of the reflections on Equation~\eqref{eq_prop_Dn_2}\&~\eqref{eq_prop_Dn_2__2}.

Firstly, we investigate the effect of the folding operation on the term $\sum_{i=0}^{n-2}[1+(n-2-i)] \times \binom{n-2}{i}$ in Equation~\eqref{eq_proof_big}.
Remember that it is obtained via $(l_i)$ (i.e.  Equation~\eqref{eq_prop_Dn_2}).
Due to the reflections, among the points in $\CBONE$ of the form $\sum_j g_j+g_1$ only $x=g_3+g_4+...+g_{i-1}+g_i+g_1$, $j\leq n$, is on the non-negative side
of all reflection hyperplanes (see result 1. of Lemma~\ref{lem_fold_2}).
%Indeed, any other point is on the non-negative side of at least one of the hyperplane $BH(g_l,g_m)$, $3\leq l <m \leq n$. Hence, at each step $i$, the term $\binom{n-2}{i}$ becomes 1.
Similarly, among the elements in $\CN(x) \cap \CBZERO$, only $x-g_1$ and $x-g_1+g_{i+1}$ (instead of  $x-g_1+g_k$, $3\leq k \leq n$)  are on the non-negative side of all reflection hyperplanes.
Hence, at each step $i$, the term $[1+(n-2-i)]$ becomes 2 (except for $i=n-2$ where it is 1).
Therefore, the folding operation reduced the term $\sum_{i=0}^{n-2}[1+(n-2-i)] \times \binom{n-2}{i}$ 
to $(n-2)\times2 + 1$.

Secondly, we investigate the reduction of the term $\sum_{i=0}^{n-2}\left[1+2(n-2-i)+\binom{n-2-i}{2}\right] \times \binom{n-2}{i}$ obtained via $(ll_i)$ (i.e.  Equation~\ref{eq_prop_Dn_2__2}).
The following results are obtained via item 2. of Lemma~\ref{lem_fold_2}.
Among the points denoted by $\sum_j g_j+g_2+g_1 \in \CBONE$ only $x=g_3+g_4+...+g_{i-1}+g_i+g_2+g_1$ is on the proper side of all reflection hyperplanes.
Among the neighbors of any of these points, of the form $(ll_i)-(2)$, only $x+g_{i+1}+g_{i+2}$ is on the proper side of all hyperplanes.
Additionally, among the neighbors of the form $(ll_i)-(1)$ and $(ll_i)-(b)$, i.e. $x+g_k$ or $x-g_2+g_k$, $3\leq k \leq n$, $g_k$ can only be $g_{i+1}$.
Therefore, the folding operation reduces the term  $\sum_{i=0}^{n-2}[1+2(n-2-i)+\binom{n-2-i}{2}] \times \binom{n-2}{i}$ 
to $(n-3)\times4 + 3+1$.

\end{proof}

\subsection{Proof of Theorem~\ref{theo_nbReg_Lin_En}}
\label{App_func_En}
%This proof is similar to the one of Appendix~\ref{App_Dn_second_kind}.
\begin{proof}
%Similarly to the proof of Theorem~\ref{eq_nbReg_Dn_second_kind}, 
We count the number of sets $\CN(x)~\cap~\CBZERO$ with  cardinality  $i$.
%The number of pieces of $f$ is then obtained by summing 
%the number of pieces of the boundary functions of each simplex. 
%(again, see the proof of Theorem~\ref{theo_nbReg_Lin_Dn_const_A}).
%Hence, we walk in $\CBZERO$ and for each of the $2^{n-1}$ points $x \in \CBZERO$ 
%we investigate the dimension of the simplex where the top corner is $x + b_1 \in \CBONE$. 
%This is achieved by counting the number of elements in $\CN(x+~b_1)~\cap~\CBZERO$. 
We walk in $\CBZERO$ and for each of the $2^{n-1}$ points $x \in \CBZERO$ 
we investigate the cardinality of the set $\CN(x+~g_1)~\cap~\CBZERO$. 
In this scope, we group the lattice points $x \in \CBZERO$ in three categories.
The numbering of these categories matches the one given 
in the sketch of proof (see also Equation~\ref{eq_En_bis} below). $\sum_j g_j$ denotes any sum of points in the set $\{0,g_j\}_{j=4}^{n}$.

\footnotesize
\begin{gather}
\label{eq_En_1}
\begin{split}
(l_i) \ &\forall \ x=\sum_j g_j \in~\CBZERO, \ x' \in D_n \backslash \{g_j,0\},\ 4 \leq k\leq n: \\
&x+ g_k \in \CN(x+g_1), \ x+ x' \not\in \CN(x+g_1) \cap \CBZERO. \\
\end{split}
\end{gather}
\begin{gather}
\label{eq_En_2}
\begin{split}
(ll_i)-A \ &\forall \ x= \sum_j g_j + g_2\in~\CBZERO, \\
& x' \in D_n \backslash \{g_i,-g_2 + g_i,-g_2 + g_i +g_k,0\},\ 4 \leq i < k\leq n: \\
&(1)  \ x+ g_i \in \CN(x+g_1), \ x -g_2 + g_i \in \CN(x+g_1), \\
&(2) \ x -g_2 + g_i +g_k  \in \CN(x+g_1), \\
&(3) \ x+ x' \not\in \CN(x+g_1) \cap \CBZERO. 
\end{split}
\end{gather}
\begin{gather}
\label{eq_En_3}
\begin{split}
(ll_i)-B \ &\forall \ x= \sum_j g_j + g_3\in~\CBZERO, \\ 
&x' \in D_n \backslash \{g_i,-g_3 + g_i,-g_3 + g_i +g_k,0\},\ 4 \leq i < k\leq n: \\
&(1) \ x+ g_i \in \CN(x+g_1), \ x -g_3 + g_i \in \CN(x+g_1), \\
&(2) \ x -g_3 + g_i +g_k  \in \CN(x+g_1), \\
&(3) \ x+ x' \not\in \CN(x+g_1) \cap \CBZERO. 
\end{split}
\end{gather}
\begin{gather}
\label{eq_En_4}
\begin{split}
(lll_i) \ &\forall \ x= \sum_j g_j + g_2+g_3 \in~\CBZERO, \\
&x' \in D_n \backslash \{g_i,g_i + g_k,g_i + g_k + g_l,0\}, \ 4 \leq i < k <l\leq n: \\
&(1) \ x -g_2+ g_k \in \CN(x+g_1), \ x -g_3+ g_k \in \CN(x+g_1), \\
& x  + g_k \in \CN(x+g_1),\\
&(2) \ x -g_3 - g_2 + g_i + g_k \in \CN(x+g_1), \\
& x  - g_2 + g_i + g_k \in \CN(x+g_1), \ x -g_3 + g_i + g_k \in \CN(x+g_1),\\
&(3) \ x+ g_i + g_k +g_l \in \CN(x+g_1), \\
& (4) \ x+ x' \not\in \CN(x+g_1) \cap \CBZERO. 
\end{split}
\end{gather}
\normalsize
We count the number of $i$-simplices per category.

$(l_i)$ is like $A_n$. 
Starting from the lattice point 0, the set $\CN(0+g_1) \cap \CBZERO$ is composed of $0$  and the $n-3$ other basis vectors (i.e. without $g_2$ and $g_3$ because they are perpendicular to $g_1$). 
Then, for all $g_{j_1}$, $4 \leq j_1\leq n$,  the sets $\CN(g_{j_1}+g_1) \cap \CBZERO$
are obtained by adding any of the $n-4$ remaining basis vectors to $g_{j_1}$  (i.e. not $g_1$, $g_2$, $g_3$ or $g_{j_1}$). % to generate a simplex in $\PPB$ where the top corner is $g_{j_1}+g_1$. 
%Indeed, if we add again~$g_{j_1}$, the resulting point is outside $\PPB$ and should not be considered. 
Hence, the cardinality of these sets is $n-3$ and there are $\binom{n-3}{1}$
ways to choose $g_{j_1}$: any basis vectors except~$g_1$, $g_2$, and $g_3$. 
Similarly, for $g_{j_1}+g_{j_2}$, $j_1 \neq j_2$, the cardinality of the sets  $\CN(g_{j_1}+g_{j_2} +g_1) \cap \CBZERO$ is  $n-4$ and there are $\binom{n-3}{2}$ ways to choose $g_{j_1}+g_{j_2}$. 
More generally, there are $\binom{n-3}{i}$ sets $\CN(x) \cap \CBZERO$ of cardinality $n-2-i$.

%From the origin, one can form a $n-1$-simplex with all basis vectors except $g_2$. Then, from any $g_j$, $j \neq 1,2$, one can form a $n-2$ simplex and there are $\binom{n-2}{1}$ ways to choose $g_j$. In general, there are $\binom{n-2}{i}$ ways to do this $n-1-i$ simplex.
$(ll_i)$ is like the basis of $D_n$ (see $(ll_i)$ in the proof in Appendix~\ref{App_theo_Dn_nb_pieces}), repeated twice because we now have two basis vectors orthogonal to $g_1$ instead of one.
Hence, we get that there are $\binom{n-3}{i}$ sets of cardinality $2 \times \left(1+2(n-3-i)+ \binom{n-3-i}{2}\right)$.

$(lll_i)$ is the new category. We investigate the neighbors of a given point $x=\sum_j g_j+ g_3+g_2+g_1$.
First (1), any $\sum_j g_j+g_3+ g_2$ is in $\CN(x)\cap \CBZERO$. Any $\sum_j g_j + g_2+g_k$, $\sum_j g_j+g_3+g_k$, and $\sum_j g_j+ g_3+g_2+g_k$, where  $4\leq k \leq n$ and $k \not\in \{ j\}$ are also in $\CN(x)\cap \CBZERO$. 
Hence, there are $3 \times (n-3-i)$ of such neighbors, where $i = |\{j\}|$ (in $\sum_j g_j$). Then, (2) any $\sum_j g_j + g_i+g_k$, $\sum_j g_j + g_2+ g_i+g_k$, and $\sum_j g_j + g_3+ g_i+g_k$,  where $4\leq i<k \leq n$ and $i,k \not\in \{ j\}$,
are in $\CN(x)\cap \CBZERO$. There are $3\times \binom{n-3-i}{2}$ possibilities, where $i = |\{j\}|$. 
Finally (3), any $\sum_j g_j + g_i+ g_k+g_l$, $4 \leq i<k<l \leq n$ and $i,k,l \not\in \{ j\}$ are in $\CN(x)\cap \CBZERO$. There are $\binom{n-3-i}{3}$ of them, where $i = |\{j\}|$.

%We are looking for the neigbors of $g_1+g_2$: the $1+2 \times (n-2)$-simplex points  $g_2$, any $g_2+g_j$, $3 \leq j \leq n$ and any $g_j$, $3 \leq j \leq n$ are in $\CN(g_1+g_2) \cap \CBZERO$.  The $\binom{n-2}{2}$ points $g_j + g_k$, $3 \le j,k\le n$, $j \neq k$, are neighbors to $g_1+g_2$ and the points $g_1+g_2+g_j \CBONE$, $3 \leq j \leq n$ have $\binom{n-2-1}{2}$ neighbors of this kind.
%One can replicate the pattern $\sum_k \binom{n-2}{k}$ for the same reason as (i), where at each step $k$ a $1+2 \times (n-2-k) + \binom{n-2-k}{2}$-simplex is obtained.
To summarize, each set replicates $\sum_i \binom{n-3}{i}$ times, where for each $i$ we have $(l_i)$ sets of cardinality $1+n-3-i$, $(ll_i)$  $2 \times \left(1+2(n-3-i)+ \binom{n-3-i}{2}\right)$, and $(lll_i)$
$1 + 3 \times (n-3-i) + 3\times \binom{n-3-i}{2}+ \binom{n-3-i}{3}$. As a result, the total number of pieces of $f$ is obtained as
\small
\begin{align}
\label{eq_En_bis}
&\sum_{i=0}^{n-3} \Bigg( \underset{(l_i)}{\underbrace{\left[ 1+(n-3-i) \right]}} +\underset{(ll_i)}{\underbrace{2 \left[ 1+ 2(n-3-i)+ \binom{n-3-i}{2} \right]}} + \\
&\underset{(lll_i)}{\underbrace{\left[ \underset{(1)}{\underbrace{1+3(n-3-i) }}+  \underset{(2)}{\underbrace{3\binom{n-3-i}{2}}} + \underset{(3)}{\underbrace{\binom{n-3-i}{3}}} \right]} }\Bigg) \times \underset{(o_i)}{\underbrace{\binom{n-3}{n-i}}}-3,
\end{align}
\normalsize
where the -3 comes from the fact that for $i=n-3$, the four pieces generated by $(l_i)$, $(ll_i)$, and $(lll_i)$ are the same. 
Indeed, the bisector hyperplane of $x$, $x+g_1$, is the same as the one of $x+g_2$, $x+g_2+g_1$, of $x+g_3$, $x+g_3+g_1$,
and of $x+g_2+g_3$, $x+g_2+g_3+g_1$, since both $g_2$ and $g_3$ are perpendicular to $g_1$.
\end{proof}

\subsection{Proof of Theorem~\ref{theo_En_folding}}
\label{App_folding_En}

%\newpage
\begin{lemma}
\label{lem_fold_3}
Among the elements of $\mathcal{C}_{\PPB}$, only the points
of the form 
\begin{enumerate}
\item $x_1=g_4+...+g_{i-1}+g_i$ and $x_1 +g_1$, 
\item $x_2=g_4+...+g_{i-1}+g_i+g_2$ and $x_2+g_1$,
\item $x_3=g_4+...+g_{i-1}+g_i+g_2+g_3$ and $x_3+g_1$, 
\end{enumerate}
$i \leq n$, 
are on the non-negative side of all $BH(g_j,g_k)$, $4\le  j < k \le~n$.  
\end{lemma}
\begin{proof}
See the proof of Lemma~\ref{lem_fold_2}.
\end{proof}

\begin{proof}(of Theorem~\ref{theo_En_folding})
(i) The folding via $BH(g_j, g_k)$, $4 \leq j < k  \leq n$ and $j=2,k=3$, switches $g_j$  and $g_k$ in
the hyperplane containing $\D(\B)$, which is orthogonal to $e_1$. Switching $g_j$
and $g_k$ does  not change the decision boundary because  of the basis
symmetry, hence $f$ is unchanged. 

Now, for (ii), how many pieces are left after all reflections?
To count the number of pieces of $f$, defined on $\D'(\B)$, 
we need to enumerate the cases where both $x \in \CBONE$ and $x'\in \CN(x) \cap \CBZERO$ are 
on the non-negative side of all reflection hyperplanes. 
%The result is easily obtained via the effect of the reflections on Equation~\eqref{eq_prop_Dn_2}\&~\eqref{eq_prop_Dn_2__2}.

Firsly, we investigate the effect of the folding operation on the term $\sum_{i=0}^{n-3}[1+n-3-i] \times \binom{n-3}{i}$ in Equation~\eqref{eq_En_bis}.
Remember that it is obtained via $(l_i)$ (i.e.  Equation~\eqref{eq_En_1}).
Due to result 1 of Lemma~\ref{lem_fold_3} and similarly to the corresponding term in the proof of Theorem~\ref{theo_Dn_second_kind_lin}, this term reduces to
$(n-3)\times2 + 1$.

Secondly, we investigate the reduction of the term $2 \left[ 1+ 2(n-3-i)+ \binom{n-3-i}{2} \right] \times \binom{n-3}{i}$, obtained via $(ll_i)$ (i.e.  Equation~\eqref{eq_En_2}).
The following results are obtained via item 2 of Lemma~\ref{lem_fold_3}. 
$\binom{n-3}{i}$ reduces to 1 at each step $i$ because in $\CBONE$, only the points $x=g_2+g_3+g_{i-1}+g_i+g_1$ are on the non-negative side of all hyperplanes, $i\leq n$.
Then, since any $\sum_j g_j + g_3 +g_1$ is on the negative side of the hyperplane $BH(g_2,g_3)$, $(ll_i)-(B)$ generates no piece in $f$ (defined to $\D'(\B)$).
$(ll_i)-(A)$ is the same situation as the situation $(ll_i)$ in the proof of Theorem~\ref{theo_Dn_second_kind_lin}.
Hence, the term reduces to $(n-3)\times(4)+3+1$.

Finally, what happens to the term $\left[1+3(n-3-i) +  3\binom{n-3-i}{2} + \binom{n-3-i}{3} \right] \binom{n-3}{n-i}$, obtained via $(lll_i)$ (i.e.  Equation~\eqref{eq_En_3})?
The following results are obtained via item 3 of Lemma~\ref{lem_fold_3}. 
As usual, $\binom{n-3}{n-i}$ reduces to 1 at each step $i$.
Then, $3(n-3-i)$, due to $(lll_i)-(1)$, becomes $2\times 1$ at each step $i$ because any $x-g_2+g_k$ (in $(lll_i)-(1)$), $k \leq 4 \leq n$, is on the negative side of $BH(g_2,g_3)$. 
For $x-g_3+g_k$ and $x+g_k$, only one valid choice of $g_k$ remains at each step $i$, as explained in the proof  of Theorem~\ref{theo_Dn_second_kind_lin}.
Regarding the term $3\binom{n-3-i}{2}$, due to $(lll_i)-(2)$, any point $x-g_2+g_i+g_k$ (in $(lll_i)-(2)$) is on the negative side of $BH(g_2,g_3)$ and 
at each step $i$ there is only one valid way to
chose $g_j$ and $g_k$ for both $x-g_3-g_2+g_j+g_k$ and $x-g_3+g_j+g_k$.
Eventually, for the last term due to $(lll_i)-(3)$ only one valid choice remain at each step $i$.
Therefore, the term due to $(lll_i)$ is reduced to
to $(n-4)\times6 + 5+3+1$.
\end{proof}
%
%\subsection{Distribution of the number of lattice points in a sphere for random lattices}
%\label{app_distri}
%
%Figure~\ref{fig_nbPts_approx} depicts the distribution of the number of lattice points in a sphere of squared radius $2 \cdot d(\Lambda)$, centered at the origin, for random MIMO lattices $\Lambda$ in dimension $n=14$.
%
%\begin{figure}[h]
%    \centering
%   	 \includegraphics[angle=270,scale=0.36]{pmf_graph_n=14_3dB_gaussian_lattices.eps}
%     \caption{Distribution of the number of lattice points in a sphere of squared radius $2 \cdot d^2(\Lambda)$. The random lattices in dimension $n=14$ are generated by a matrix $G$ with random i.i.d. $\mathcal{N}(0,1)$ componens. We numerically generated 200000 random lattices in dimensions $n=14$ to compute the distribution.}
%     \label{fig_nbPts_approx}
%\end{figure}

%% file: journal_split_without_blue.bbl
\begin{thebibliography}{99}

\bibitem{Agrell2002} E.~Agrell, T.~Eriksson, A.~Vardy, and K.~Zeger,
``Closest point search in lattices,''
{\em IEEE Trans. on Inf. Theory}, vol.~48, no.~8, pp.~2201-2214, 2002.

\bibitem{Arora2018} R.~Arora, A.~Basu, P.~Mianjy, and A.~Mukherjee, 
``Understanding deep neural networks with rectified linear units,"
{ \em International Conference on Learning Representations}, Nov.~2016.


\bibitem{Conway1992} J.~H.~Conway and N.~J.~A.~Sloane, 
``Low-Dimensional Lattices. VI. Voronoi Reduction of Three-Dimensional Lattices,"
{ \em Proceedings: Mathematical and Physical Sciences by the Royal Society}, vol.~436, no.~1896,  pp.~55-68, Jan.~1992.


\bibitem{Conway1999} J.~H.~Conway and N.~J.~A.~Sloane. {\em Sphere packings, lattices and groups}.
Springer-Verlag, New York, 3rd ed., 1999.

\bibitem{Cohen1996}
H.~Cohen,
\emph{A course in computational algebraic number theory}.
Springer-Verlag, New York, 1993.

\bibitem{Corlay2018_a} V.~Corlay, J.J.~Boutros, P.~Ciblat, and L.~Brunel,
``Multilevel MIMO Detection with Deep Learning,''
 { \em $52^{th}$ Asilomar Conference on Signals, Systems and Computers}, May 2018.

\bibitem{Corlay2018}  V.~Corlay, J.J.~Boutros, P.~Ciblat, and L.~Brunel, ``Neural Lattice Decoders," 
{\em 6th IEEE Global Conference on Signal and Information Processing}, Nov. 2018.

\bibitem{Forney1988}  G.~D.~Forney, Jr., 
``Coset codes. I. Introduction and geometrical classification,"  
\emph{IEEE Transactions on Information Theory}, 
vol.~34, no.~5, pp.~1123-1151, Sep.~1988. 

\bibitem{Coxeter1973} H.~Coxeter. {\em Regular Polytopes}. 
Dover, New York, 3rd ed., 1973.

\bibitem{Goodfellow2016} I.~Goodfellow, Y.~Bengio, and A.~Courville. {\em Deep Learning}. 
The MIT Press, 2016.

\bibitem{Gruber2017} T.~Gruber, S.~Cammerer, J.~Hoydis, and S.~ten~Brink, 
``On deep learning-based channel decoding,'' {\em Conference on Information Sciences and Systems}, March 2017.

\bibitem{He2020}
H.~He, C.-K.~Wen, S.~Jin, G.~Ye~Li, "Model-Driven Deep Learning
for MIMO Detection," 
\textit{IEEE Trans. on Signal Processing}, vol.~68 , pp.~1702-1715, Feb.~2020.


\bibitem{Hoydis2017} 
T.~O’Shea  and J.~Hoydis,
``An introduction to deep learning for the physical layer,''
 \emph{IEEE Trans. on Cognitive Communications and Networking}, vol.~3, no.~4, pp.~563 - 575, Dec.~2017.
 
\bibitem{Krizhevsky2012}
A.~Krizhevsky, I.~Sutskever,  and G.~Hinton,
``ImageNet Classification with Deep Convolutional Neural Networks,''
{\em Advances in Neural Information Processing Systems 25},
pp.~1097-1105, 2012.


%\bibitem{Daniely2017} A.~Daniely,
%``Depth Separation for Neural Networks,"
%\emph{34th Annual Conference on Learning Theory},
%pp.~690-696, 2017.
%
%\bibitem{Eldan2016} R.~Eldan and O.~Shamir,
%``The power of depth for feedforward neural networks,"
%\emph{29th Annual Conference on Learning Theory},
%pp.~907–940, 2016.

\bibitem{Micciancio2002} D.~Micciancio and S.~Goldwasser,
{\em Complexity of lattice problems, a cryptographic perspective}.
Kluwers Academic Publishers, 2002.

\bibitem{Mohammadkarimi2019}
M. Mohammadkarimi, M. Mehrabi, M. Ardakani and Y. Jing, 
"Deep Learning-Based Sphere Decoding,"  
\textit{IEEE Trans. on Wireless Communications}, vol.~18, no.~9, pp.~4368-4378, Sept.~2019.

\bibitem{Montufar2014} G.~Mont\`ufar, R.~Pascanu, K.~Cho, and Y.~Bengio, ``On the Number of Linear Regions of Deep Neural Networks,'' 
{ \em Advances in neural information processing systems},  Dec. 2014.

\bibitem{Nachmani2016} E.~Nachmani, Y.~Be'ery and D.~Burshtein, 
``Learning to decode linear codes using deep learning,'' 
{\em 54th Annual Allerton Conference on Communication, Control, and Computing (Allerton)}, Monticello, Illinois, pp.~341-346, Sept.~2016.
%
\bibitem{Pascanu2013} R.~Pascanu, G.~Montufar, and Y.~Bengio, 
``On the number of inference regions of deep feed forward with piece-wise linear activations,"  
arXiv preprint arXiv:1312.6098, 2013.
%
%\bibitem{Petersen2018} P.~Petersen and F.~Voigtlaender
%``Optimal approximation of piecewise smooth functions
%using deep ReLU neural networks,"
%journal = {Neural Networks, Elsevier},
%vol.~108, pp.~296-330, Dec.~2018.
%
%\bibitem{Poggio2017} T.~Poggio and H.~Mhaskar and L.~Rosasco and B.~Miranda and Q.~Liao
%``Why and When Can Deep – but Not Shallow – Networks Avoid the
%Curse of Dimensionality: a Review",
%\emph{Center for Brains, Minds and Machines (CBMM)
%Memo No. 58}, 2017.


\bibitem{Poltyrev1994}
 G.~Poltyrev,
 ``On coding without restrictions for the AWGN channel,''
 \emph{IEEE Trans. Inf. Theory}, vol.~40, no.~2, pp.~409-417, Mar.~1994.

\bibitem{Proakis2008} J.G.~Proakis and M.~Salehi,
{\em Digital Communications}.
McGraw-Hill, 5th ed., 2008.
%
%\bibitem{Raghu2016}, M.~Raghu, B.~Poole, J.~Kleinberg, S.~Ganguli, and J.~Sohl-Dickstein, 
%``On the expressive power of deep neural networks,"
%\emph{arXiv preprint arXiv:1606.05336},
%2016.
%
%\bibitem{Safran2017} I.~Safran and O.~Shamir,
%``Depth-Width Tradeoffs in Approximating Natural Functions with
%Neural Networks,"
%\emph{34th Annual Conference on Learning Theory}, pp.~2979-2987, 2017.

\bibitem{Samuel2017} N.~Samuel, T.~Diskin, and A.~Wiesel, ``Deep MIMO detection,'' 
\emph{2017 IEEE 18th International Workshop on Signal Processing Advances in Wireless Communications}, July 2017.

\bibitem{Szymanski2013} L.~Szymanski and B.~McCane, ``Learning in deep architectures with folding transformations," 
{ \em International Joint Conference on Neural Networks }, Aug. 2013.

\bibitem{Telgarsky2016} M.~Telgarsky, ``Benefits of depth in neural networks," 
 { \em 29th Annual Conference on Learning Theory},  June 2016.

\bibitem{Viazovska2017} M.~Viazovska, ``The sphere packing problem in dimension 8", 
{ \em Annals of Mathematics}, vol.~185, no.~2, pp.~991–1015, 2017.

\bibitem{Zamir2014} R. Zamir,
{\em Lattice coding for signals and networks}.
Cambridge University Press, 2014.

\end{thebibliography}
